\documentclass[english]{article}
\usepackage[T1]{fontenc}
\usepackage{setspace}
\usepackage[latin9]{inputenc}
\usepackage[a4paper]{geometry}
\usepackage[active]{srcltx}
\usepackage{xcolor}
\usepackage{babel}
\usepackage{verbatim}
\usepackage[square,sort,comma,numbers]{natbib}

\usepackage{float}
\usepackage{mathrsfs}
\usepackage{enumitem}
\usepackage{amsmath}
\usepackage{amsthm}
\usepackage{amssymb}
\usepackage{graphicx}
\usepackage{esint}
\PassOptionsToPackage{normalem}{ulem}
\usepackage{ulem}
\usepackage{breakurl}

\makeatletter

\floatstyle{ruled}
\newfloat{algorithm}{tbp}{loa}
\providecommand{\algorithmname}{Algorithm}
\floatname{algorithm}{\protect\algorithmname}
\providecolor{lyxadded}{rgb}{0,0,1}
\providecolor{lyxdeleted}{rgb}{1,0,0}

\newcommand{\lyxaddress}[1]{
\par {\raggedright #1
\vspace{1.4em}
\noindent\par}
}
 \theoremstyle{definition}
  \newtheorem{example}{\protect\examplename}
  \theoremstyle{plain}
  \newtheorem{prop}{\protect\propositionname}
\theoremstyle{plain}
\newtheorem{thm}{\protect\theoremname}
  \theoremstyle{plain}
  \newtheorem{lem}{\protect\lemmaname}

 \usepackage{url}

\usepackage{bbm}

\usepackage{amsmath,amssymb}
\usepackage{textcomp,mathrsfs}
\usepackage{graphicx,amsthm,mathrsfs}
\usepackage{multirow,makeidx,algorithmic,algorithm}
\usepackage{latexsym,graphicx}
\usepackage{mathrsfs}
\usepackage{booktabs,fancyhdr} 
\usepackage{rotating}
\usepackage{pdflscape}
\usepackage{graphicx}
\usepackage{array}
\usepackage{bigints}
\usepackage{listings}
\usepackage{titlesec}
\usepackage{mathtools}
\usepackage{natbib}

\usepackage{hyperref}
\hypersetup{colorlinks,citecolor=blue,urlcolor=blue,filecolor=blue,backref=page}

\usepackage{notation}

\usepackage{xr}

\newcounter{xxx}
\newcommand\dotprod[2]{\left\langle #1, #2\right\rangle}  

\newcommand\R{{\mathbb R}}        

\renewcommand\P{{\mathbb P}}        
\newcommand\E{{\mathbb E}}        
\def\1{{\mathbf 1}}        

\newcommand{\ud}{\,\mathrm{d}}    

\abbreviation{RFMC}{Rejection Free Monte Carlo}

\notation*[Dimensionality]{d}{d}
\notation[Pair of position and velocity]{z}{z}
\notation{Z}{Z}
\notation[Position]{x}{x}
\notation*[Velocity]{v}{v}
\notation[normalisedIP]{xv}{m}
\notation[Normalising constant]{normConstant}{K}
\notation[NormalisedIP]{XV}{M}
\notation[PoissonP]{PP}{\Pi}
\notation[outer normal]{nor}{n}

\notation[Energy function]{En}{U}
	
\notation{No}{\textrm{No}}
\notation[Intermediate time dummy variable]{s}{s}
\notation[First arrival of an inhomogeneous Poisson process (collision time)]{T}{T}
\notation[Sequence of random collision times]{TVec}{\boldsymbol{T}}
\notation[Density of the first arrival time]{q}{q} 

\notation{qtilde}{\tilde q}
\notation[Collision operator]{Col}{C}
\notation[Collision wall operator]{ColW}{C_b}
\notation[Time to Wall]{twall}{t_W}
\notation[minus v]{m}{m}
\notation{zSpace}{\mathcal{Z}}
\notation[Number of collisions]{NCol}{{\#\text{Col}}}

\notation[Trajectory]{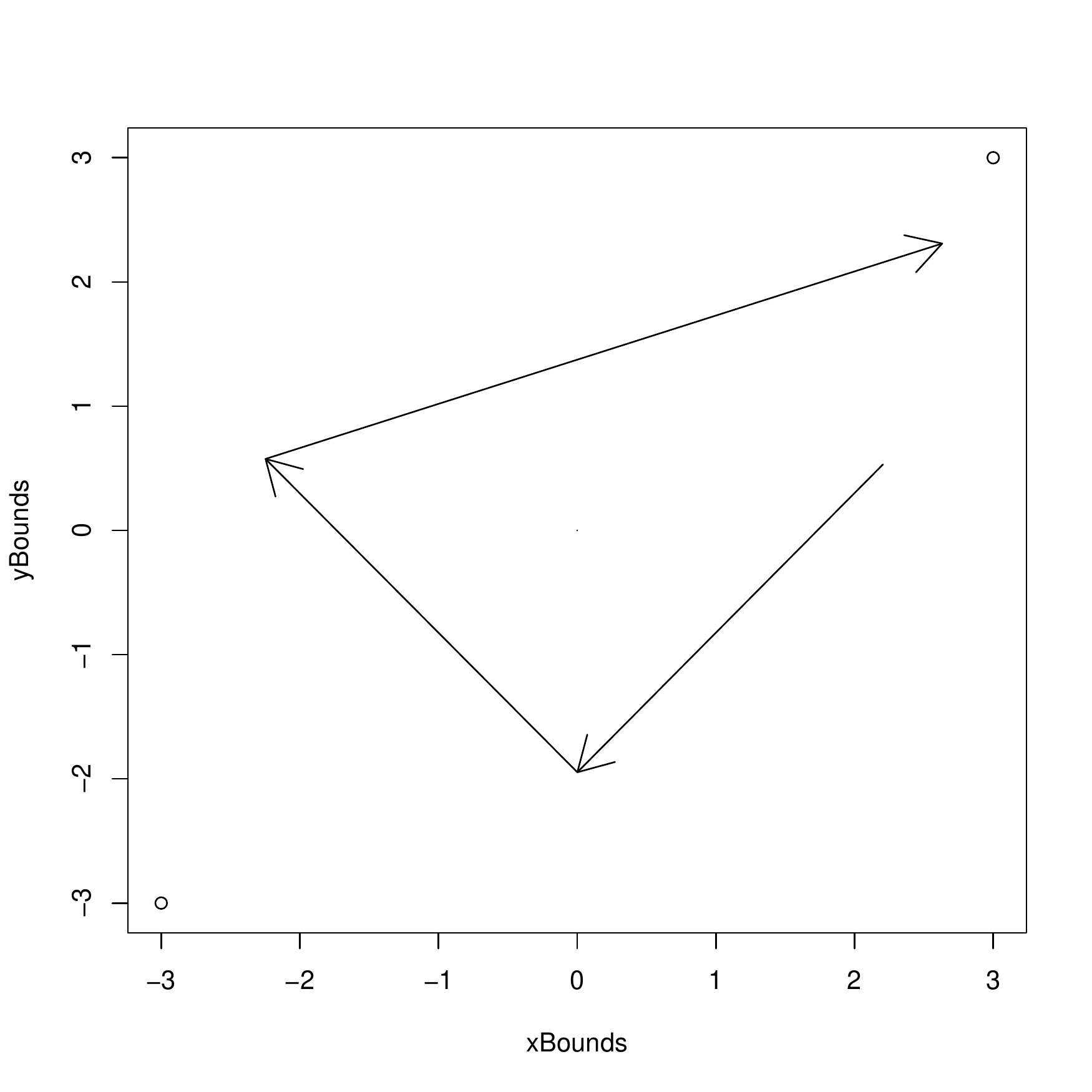}{\tau}
\notation*[Total trajectory length]{L}{L}

\notation[Expected value]{I}{I}

\notation[Expected value]{dom}{D}

\notation{pos}{\text{pos}}
\notation{vel}{\text{vel}}

\notation{budget}{b}  
\notation{tauc}{\tau_c}
\notation{tauref}{\tau_{\text{ref}}}

\notation*[Artificial time indexing the sample path produced by the rejection-free algorithms]{t}{t}
\notation{tVec}{\boldsymbol{t}}
\notation[Intensity]{Int}{\lambda}
\notation[Intensity renewal refresh]{RInt}{\lambda^{\text{ref}}}
\notation[Refresh trans]{Rtrans}{P^{\text{ref}}}
\notation[ trans kernel]{trans}{P}
\notation[Integrated autocorrelation time]{IACT}{IACT}

\notation[ergodicity proof set]{ergSet}{U}

\notation[Test function]{h}{h}
\notation[Test function2]{g}{g}
\notation[Surface Area]{SurfSphere}{S}

\DeclareMathOperator{\vol}{vol}

  \providecommand{\examplename}{Example}
  \providecommand{\lemmaname}{Lemma}
  \providecommand{\propositionname}{Proposition}
\providecommand{\theoremname}{Theorem}

\makeatother

  \providecommand{\examplename}{Example}
  \providecommand{\lemmaname}{Lemma}
  \providecommand{\propositionname}{Proposition}
\providecommand{\theoremname}{Theorem}

\begin{document}

\title{The Bouncy Particle Sampler: A Non-Reversible Rejection-Free Markov
Chain Monte Carlo Method}

\author{Alexandre Bouchard-Côté$^{*}$, Sebastian J. Vollmer$^{\dagger}$
and Arnaud Doucet$\ddagger$ }
\maketitle

\lyxaddress{$^{*}$Department of Statistics, University of British Columbia,
Canada.}

\lyxaddress{$^{\dagger}$Mathematics Institute and Department of Statistics,
University of Warwick, UK. }

\lyxaddress{$^{\ddagger}$Department of Statistics, University of Oxford, UK. }
\begin{abstract}
Many Markov chain Monte Carlo techniques currently available rely
on discrete-time reversible Markov processes whose transition kernels
are variations of the Metropolis\textendash Hastings algorithm. We
explore and generalize an alternative scheme recently introduced in
the physics literature \cite{PetersDeWith2012} where the target
distribution is explored using a continuous-time non-reversible piecewise-deterministic
Markov process. In the Metropolis\textendash Hastings algorithm, a
trial move to a region of lower target density, equivalently of higher
``energy'', than the current state can be rejected with positive
probability. In this alternative approach, a particle moves along
straight lines around the space and, when facing a high energy barrier,
it is not rejected but its path is modified by bouncing against this
barrier. By reformulating this algorithm using inhomogeneous Poisson
processes, we exploit standard sampling techniques to simulate exactly
this Markov process in a wide range of scenarios of interest. Additionally,
when the target distribution is given by a product of factors dependent
only on subsets of the state variables, such as the posterior distribution
associated with a probabilistic graphical model, this method can be
modified to take advantage of this structure by allowing computationally
cheaper ``local'' bounces which only involve the state variables
associated to a factor, while the other state variables keep on evolving.
In this context, by leveraging techniques from chemical kinetics,
we propose several computationally efficient implementations. Experimentally,
this new class of Markov chain Monte Carlo schemes compares favorably
to state-of-the-art methods on various Bayesian inference tasks, including
for high dimensional models and large data sets. 
\end{abstract}
{\small{}{}Keywords: Inhomogeneous Poisson process; Markov chain
Monte Carlo; Piecewise deterministic Markov process; Probabilistic
graphical models; Rejection-free simulation.}{\small \par}

\section{Introduction\label{sec:Intro}}

Markov chain Monte Carlo (MCMC) methods are standard tools to sample
from complex high-dimensional probability measures. Many MCMC schemes
available at present are based on the Metropolis-Hastings (MH) algorithm
and their efficiency is strongly dependent on the ability of the user
to design proposal distributions capturing the main features of the
target distribution; see \cite{liu2008monte} for a comprehensive
review. We examine, analyze and generalize here a different approach
to sample from distributions on $\mathbb{R}^{d}$ that has been recently
proposed in the physics literature \cite{PetersDeWith2012}. Let
the energy be defined as minus the logarithm of an unnormalized version
of the target density. In this methodology, a particle explores the
space by moving along straight lines and, when it faces a high energy
barrier, it bounces against the contour lines of this energy. This
non-reversible rejection-free MCMC method will be henceforth referred
to as the Bouncy Particle Sampler (BPS). This algorithm and closely
related schemes have already been adopted to simulate complex physical
systems such as hard spheres, polymers and spin models \cite{kampmann2015monte,michel2014generalized,michel2015event,nishikawa2015event}.
For these models, it has been demonstrated experimentally that such
methods can outperform state-of-the-art MCMC methods by up to several
orders of magnitude.

However, the implementation of the BPS proposed in \cite{PetersDeWith2012}
is not applicable to most target distributions arising in statistics.
In this article we make the following contributions: 
\begin{description}
\item [{Simulation~schemes~based~on~inhomogeneous~Poisson~processes:}] by
reformulating explicitly the bounces times of the BPS as the first
arrival times of inhomogeneous Poisson Processes (PP), we leverage
standard sampling techniques \citep[Chapter 6]{Devroye1986} and methods
from chemical kinetics \cite{thanh2015simulationrejection} to obtain
new computationally efficient ways to simulate the BPS process for
a large class of target distributions. 
\item [{Factor~graphs:}] when the target distribution can be expressed
as a factor graph \cite{WainwrightJordan2008}, a representation
generalizing graphical models where the target is given by a product
of factors and each factor can be a function of only a subset of variables,
we adapt a physical multi-particle system method discussed in \citep[Section III]{PetersDeWith2012}
to achieve additional computational efficiency. This local version
of the BPS only manipulates a restricted subset of the state components
at each bounce but results in a change of all state components, not
just the one being updated contrary the Gibbs sampler. 
\item [{Ergodicity~analysis:}] we present a proof of the ergodicity of
BPS when the velocity of the particle is additionally refreshed at
the arrival times of an homogeneous PP. When this refreshment step
is not carried out, we exhibit a counter-example where ergodicity
does not hold. 
\item [{Efficient~refreshment:}] %
we propose alternative refreshment schemes and compare their computational
efficiency experimentally. 
\end{description}
Empirically, these new MCMC schemes compare favorably to state-of-the-art
MCMC methods on various Bayesian inference problems, including for
high-dimensional scenarios and large data sets. Several additional
original extensions of the BPS including versions of the algorithm
which are applicable to mixed continuous-discrete distributions, distributions
restricted to a compact support and a method relying on the use of
curved dynamics instead of straight lines can be found in \cite{bpsv1}.
For brevity, these are not discussed here.

The rest of this article is organized as follows. In Section \ref{sec:Basic-bouncy-particle-sampler},
we introduce the basic version of the BPS, propose original ways to
implement it and prove its ergodicity under weak assumptions. Section
\ref{sec:Local-bouncy-particle-sampler} presents a modification of
the basic BPS which exploits a factor graph representation of the
target distribution and develops computationally efficient implementations
of this scheme. In Section \ref{sec:Applications}, we demonstrate
this methodology on various Bayesian models. The proofs are given
in the Appendix and the Supplementary Material.

\section{The bouncy particle sampler\label{sec:Basic-bouncy-particle-sampler}}

\subsection{Problem statement and notation }

Consider a probability distribution $\pi$ on $\mathbb{R}^{d},$ equipped
with the Borel $\sigma$-algebra $\mathcal{B}(\mathbb{R}^{d})$. We
assume that $\pi$ admits a probability density with respect to the
Lebesgue measure ${\rm d}x$ and slightly abuse notation by denoting
also this density by $\pi$. In most practical scenarios, we only
have access to an unnormalized version of this density, that is 
\[
\pi\left(x\right)=\frac{\gamma\left(x\right)}{\mathcal{Z}},
\]
where $\gamma:\mathbb{R}^{d}\rightarrow(0,\infty)$ can be evaluated
pointwise but the normalizing constant $\mathcal{Z}=\int_{\mathbb{R}^{d}}\gamma\left(x\right){\rm d}x$
is unknown. We call 
\[
U\left(x\right)=-\mathrm{log}~\gamma\left(x\right)
\]
the associated energy, which is assumed continuously differentiable,
and we denote by $\nabla U\left(x\right)=\left(\frac{\partial U\left(x\right)}{\partial x_{1}},\ldots,\frac{\partial U\left(x\right)}{\partial x_{d}}\right)^{\top}$
the gradient of $U$ evaluated at $x$. We are interested in approximating
numerically the expectation of arbitrary test functions $\varphi:\mathbb{R}^{d}\rightarrow\mathbb{R}$
with respect to $\pi$.

\subsection{Algorithm description\label{subsec:Algorithm-description}}
\begin{center}
\begin{figure}
\begin{centering}
\includegraphics[scale=0.5]{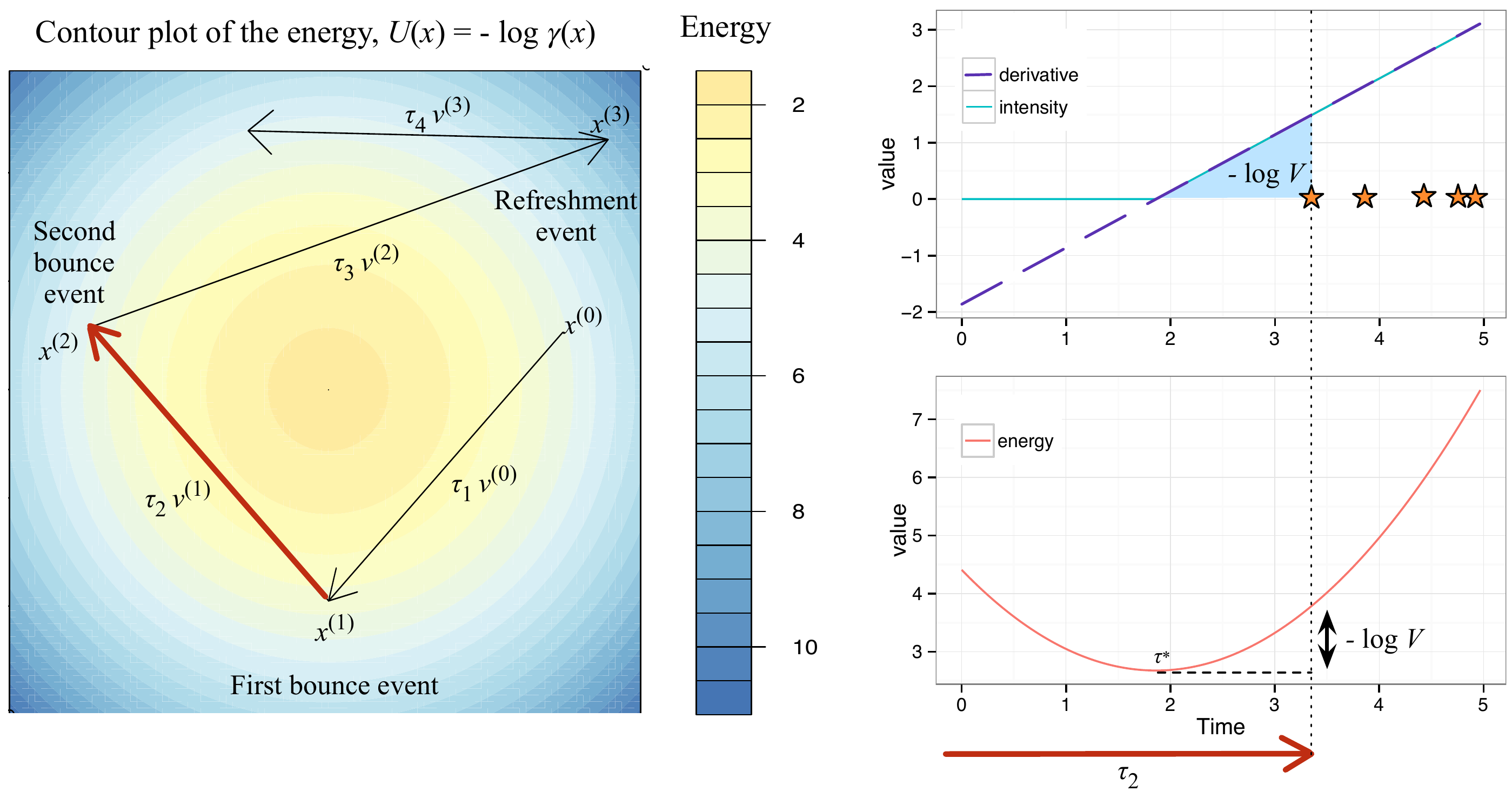} 
\par\end{centering}
\protect\caption{\label{fig:normal-eg}Illustration of BPS on a standard bivariate
Gaussian distribution. Left and top right: see Section \ref{subsec:Algorithm-description};
bottom right: see Example \ref{ex:log-concave}.}
\end{figure}
\par\end{center}

The BPS methodology introduced in \cite{PetersDeWith2012} simulates
a continuous piecewise linear trajectory $\left\{ x\left(t\right)\right\} _{t\geq0}$
in $\mathbb{R}^{d}$. It has been informally derived as a continuous-time
limit of the Metropolis algorithm in \cite{PetersDeWith2012}. Each
segment in the trajectory is specified by an initial position $x^{\left(i\right)}\in\mathbb{R}^{d}$,
a length $\tau_{i+1}\in\mathbb{R}^{+}$ and a velocity $v^{\left(i\right)}\in\mathbb{R}^{d}$
(example shown in Figure \ref{fig:normal-eg}, left). We denote the
times where the velocity changes by $t_{i}=\sum_{j=1}^{i}\tau_{j}$
for $i\geq1$, and set $t_{0}=0$ for convenience. The position at
time $t\in[t_{i},t_{i+1})$ is thus interpolated linearly, $x\left(t\right)=x^{\left(i\right)}+v^{\left(i\right)}\left(t-t_{i}\right)$,
and each segment is connected to the next, $x^{\left(i+1\right)}=x^{\left(i\right)}+v^{\left(i\right)}\tau_{i+1}$.
The length of these segments is governed by an inhomogeneous PP of
intensity function $\lambda:\mathbb{R}^{d}\times\mathbb{R}^{d}\rightarrow[0,\infty)$
\begin{equation}
\lambda\left(x,v\right)=\mathrm{max}\left\{ 0,\left\langle \nabla U\left(x\right),v\right\rangle \right\} .\label{eq:intensityfunction}
\end{equation}
When the particle bounces, its velocity is updated in the same way
as a Newtonian elastic collision on the hyperplane tangential to the
gradient of the energy. Formally, the velocity after bouncing is given
by 
\begin{equation}
R\left(x\right)v=\left(I_{d}-2\frac{\nabla U\left(x\right)\left\{ \nabla U\left(x\right)\right\} ^{\top}}{\left\Vert \nabla U\left(x\right)\right\Vert ^{2}}\right)v=v-2\frac{\left\langle \nabla U\left(x\right),v\right\rangle }{\left\Vert \nabla U\left(x\right)\right\Vert ^{2}}\nabla U\left(x\right),\label{eq:Projectionoperator}
\end{equation}
where $I_{d}$ denotes the $d\times d$ identity matrix, $\|\cdot\|$
the Euclidean norm, and $\left\langle w,z\right\rangle =w^{t}z$ the
scalar product between column vectors $w,z$. \footnote{Computations of the form $R(x)v$ are implemented via the right-hand
side of Equation \ref{eq:Projectionoperator} which takes time $O(d$)
rather than the left-hand side, which would take time $O(d^{2}$).}\cite{PetersDeWith2012} also refresh the velocity at periodic times.
We slightly modify their approach by performing a velocity refreshment
at the arrival times of a homogeneous PP of intensity $\lambda^{\mathrm{ref}}\geq0$,
$\lambda^{\mathrm{ref}}$ being a parameter of the algorithm. A similar
refreshment scheme was used for a related process in \citep{michel2014generalized}.
Throughout the paper, we use the terminology ``event'' for a time
at which either a bounce or a refreshment occurs. The basic version
of the BPS algorithm proceeds as follows:

\begin{algorithm}[H]
\protect\caption{Basic BPS algorithm ~\label{alg:BasicBouncyParticleSamplerRefreshment}}

\begin{enumerate}
\item Initialize $\left(x^{\left(0\right)},v^{\left(0\right)}\right)$ arbitrarily
on $\mathbb{R}^{d}\times\mathbb{R}^{d}$ and let $T$ denote the requested
trajectory length. 
\item For $i=1,2,\ldots$
\begin{enumerate}
\item Simulate the first arrival time $\tau{}_{\mathrm{bounce}}\in(0,\infty)$
of a PP of intensity \label{enu:ArrivalPoisson} 
\[
\chi\left(t\right)=\lambda(x^{\left(i-1\right)}+v^{\left(i-1\right)}t,v^{\left(i-1\right)}).
\]
\item Simulate $\tau{}_{\mathrm{ref}}\sim\mathrm{Exp\left(\lambda^{\mathrm{ref}}\right)}$. 
\item Set $\tau_{i}\leftarrow\mathrm{min}\left(\tau{}_{\mathrm{bounce}},\tau{}_{\mathrm{ref}}\right)$
and compute the next position using 
\begin{equation}
x^{\left(i\right)}\leftarrow x^{\left(i-1\right)}+v^{\left(i-1\right)}\tau_{i}.\label{eq:stateupdate}
\end{equation}
\item If $\tau_{i}=\tau{}_{\mathrm{ref}}$, sample the next velocity $v^{\left(i\right)}\sim\mathcal{N}\left(0_{d},I_{d}\right)$.
\label{enu:Refresh-step} 
\item If $\tau_{i}=\tau{}_{\mathrm{bounce}}$, compute the next velocity
$v^{\left(i\right)}$ using 
\begin{equation}
v^{\left(i\right)}\leftarrow R\left(x^{\left(i\right)}\right)v^{\left(i-1\right)}.\label{eq:bouncingprojection}
\end{equation}
\item If $t_{i}=\sum_{j=1}^{i}\tau_{j}\ge T$ exit For Loop (line 2). 
\end{enumerate}
\end{enumerate}
\end{algorithm}

In the algorithm above, $\mathrm{Exp\left(\delta\right)}$ denotes
the exponential distribution of rate $\delta$ and $\mathcal{N}\left(0_{d},I_{d}\right)$
the standard normal on $\mathbb{R}^{d}$. Refer to Figure \ref{fig:normal-eg}
for an example of a trajectory generated by BPS on a standard bivariate
Gaussian target distribution. An example of a bounce time simulation
is shown in Figure \ref{fig:normal-eg}, top right, for the segment
between the first and second events\textemdash the intensity $\chi(t)$
(turquoise) is obtained by thresholding $\left\langle \nabla U\left(x^{(1)}+v^{\left(1\right)}t\right),v^{(1)}\right\rangle $
(purple, dashed); arrival times of the PP of intensity $\chi(t)$
are denoted by stars.

We will show further that the transition kernel of the BPS process
admits $\pi$ as invariant distribution for any $\lambda^{\mathrm{ref}}\geq0$
but it can fail to be irreducible when $\lambda^{\mathrm{ref}}=0$
as demonstrated in Section \ref{subsec:Isotropic-Multivariate-Normal}.%
{} It is thus critical to use $\lambda^{\mathrm{ref}}>0$. Our proof
of invariance and ergodicity can accommodate some alternative refreshment
steps \ref{enu:Refresh-step}. One such variant, which we call restricted
refreshment, samples $v^{(i)}$ uniformly on the unit hypersphere
$\mathcal{S}^{d-1}=\left\{ x\in\mathbb{R}^{d}:\Vert x\Vert=1\right\} $.
We compare experimentally these two variants and others in Section
\ref{subsec:Comparisons-of-refresh}.

\subsection{Algorithms for bounce time simulation\label{subsec:Algorithm-implementation}}

Implementing BPS requires sampling the first arrival time $\tau$
of a one-dimensional inhomogeneous PP $\Pi$ of intensity $\chi(t)=\lambda(x+vt,v)$
given by (\ref{eq:intensityfunction}). Simulating such a process
is a well-studied problem; see \citep[Chapter 6, Section 1.3]{Devroye1986}.
We review here three methods and illustrate how they can be used to
implement BPS for examples from Bayesian statistics. The first method
described in Section \ref{subsec:Simulation-time-scale-transform}
will be particularly useful when the target is log-concave, while
the two others described in Section \ref{subsec:Simulation-through-adaptive}
and Section \ref{subsec:Simulation-using-composition} can be applied
to more general scenarios.

\subsubsection{Simulation using a time-scale transformation\label{subsec:Simulation-time-scale-transform}}

If we let $\varXi\left(t\right)=\int_{0}^{t}\chi\left(s\right){\rm d}s$,
then the PP $\Pi$ satisfies 
\begin{align*}
\P(\tau>u) & =\P(\Pi\cap[0,u)=\emptyset)=\exp(-\varXi(u)),
\end{align*}
and therefore $\tau$ can be simulated from a uniform variate $V\sim\mathcal{U}\left(0,1\right)$
via the identity 
\begin{equation}
\tau=\varXi^{-1}(-\log(V)),\label{eq:integralEquation}
\end{equation}
where $\varXi^{-1}$ denotes the quantile function of $\varXi,$ $\varXi^{-1}(p)=\inf\mbox{\ensuremath{\left\{ t:p\le\varXi(t)\right\} } }.$
Refer to Figure \ref{fig:normal-eg}, top right for a graphical illustration.
This identity corresponds to the method proposed in \cite{PetersDeWith2012}
to determine the bounce times and is also used in \cite{michel2014generalized,michel2015event,nishikawa2015event}
to simulate related processes.

In general, it is not possible to obtain an analytical expression
for $\tau$. However, when the target distribution is strictly log-concave
and differentiable, it is possible to solve Equation (\ref{eq:integralEquation})
numerically (see Example \ref{ex:log-concave} below). 
\begin{example}
\emph{\label{ex:log-concave}Log-concave densities}. If the energy
is strictly convex (see Figure \ref{fig:normal-eg}, bottom right),
we can minimize it along the line specified by $(x,v)$ 
\[
\tau_{*}=\mathrm{argmin}_{t:t\geq0}~U\left(x+vt\right),
\]
where $\tau_{*}$ is well defined and unique by strict convexity.
On the interval $\left[0,\tau_{*}\right)$, which might be empty,
we have ${\rm d}U\left(x+vt\right)/{\rm d}t<0$ and ${\rm d}U\left(x+vt\right)/{\rm d}t\geq0$
on $\left[\tau_{*},\text{\ensuremath{\infty}}\right)$. The solution
$\tau$ of (\ref{eq:integralEquation}) is thus necessarily such that
$\tau\geq\tau_{*}$ and (\ref{eq:integralEquation}) can be rewritten
using the gradient theorem as 
\begin{equation}
\int_{\tau_{*}}^{\tau}\frac{{\rm d}U\left(x+vt\right)}{{\rm d}t}{\rm d}t=U\left(x+v\tau\right)-U\left(x+v\tau_{*}\right)=-\log V.\label{eq:convex-case}
\end{equation}
Even if we only compute $U$ pointwise through a black box, we can
solve \eqref{eq:convex-case} through line search within machine precision.

We note that \eqref{eq:convex-case} also provides an informal connection
between the BPS and MH algorithms. Exponentiating this equation, we
get indeed 
\[
\frac{\pi(x+v\tau)}{\pi(x+v\tau_{*})}=V.
\]

Hence, in the log-concave case, and when the particle is climbing
the energy ladder (i.e., $\tau_{*}=0$), BPS can be viewed as ``swapping''
the order of the steps taken by the MH algorithm. In the latter, we
first sample a proposal and second sample a uniform $V$ to perform
an accept-reject decision. With BPS, $V$ is first drawn then the
maximum distance allowed by the same MH ratio is travelled. As for
the case of a particle going down the energy ladder, the behavior
of BPS is simpler to understand: bouncing never occurs. We illustrate
this method for Gaussian distributions.

\emph{\label{ex:normal}Multivariate Gaussian distributions}. Let
$U\left(x\right)=\left\Vert x\right\Vert ^{2}$, then simple calculations
yield 
\begin{equation}
\tau=\frac{{1}}{\left\Vert v\right\Vert ^{2}}\begin{cases}
-\left\langle x,v\right\rangle +\sqrt{-\left\Vert v\right\Vert ^{2}\log V} & \text{if}\left\langle x,v\right\rangle \leq0,\\
-\left\langle x,v\right\rangle +\sqrt{\left\langle x,v\right\rangle ^{2}-\left\Vert v\right\Vert ^{2}\log V} & \text{ otherwise.}
\end{cases}\label{eq:ComputeCollisionsIsoGaussian}
\end{equation}
\end{example}

\subsubsection{Simulation using adaptive thinning\label{subsec:Simulation-through-adaptive}}

When it is difficult to solve (\ref{eq:integralEquation}), the use
of an adaptive thinning procedure provides an alternative. Assume
we have access to local-in-time upper bounds $\bar{\chi}{}_{s}\left(t\right)$
on $\chi(t)$, that is 
\begin{align*}
\bar{\chi}_{s}(t) & =0\,\text{ for all }t<s,\\
\bar{\chi}{}_{s}(t) & \ge\chi(t)\text{ for all }s\le t\le s+\Delta(s),
\end{align*}
where $\triangle$ is a positive function (standard thinning corresponds
to $\Delta=+\infty$). Assume additionally that we can simulate the
first arrival time of the PP $\bar{\Pi}_{s}$ with intensity $\bar{\chi}_{s}(t)$.
Such bounds can be constructed based on upper bounds on directional
derivatives of $U$ provided the remainder of the Taylor expansion
can be controlled%
. Algorithm \ref{alg:thinning} shows the pseudocode for the adaptive
thinning procedure.

\begin{algorithm}[H]
\protect\caption{Simulation of the first arrival time of a PP through thinning ~\label{alg:thinning}}

\begin{enumerate}
\item Set $s\leftarrow0$, $\tau\leftarrow0$. 
\item Do
\begin{enumerate}
\item Set $s\leftarrow\tau$. 
\item Sample $\tau$ as the first arrival point of the PP $\bar{\Pi}_{s}$
of intensity $\bar{\chi}_{s}$.\label{enu:Sample-bound} 
\item If $\bar{\Pi}_{s}=\{\emptyset\}$ then set $\tau\leftarrow s+\triangle(s)$. 
\item If $s+\triangle(s)\leq\tau$ set $s\leftarrow s+\triangle(s)$ and
go to (b). 
\item While $V>\{\chi\left(\tau\right)/\bar{\chi}_{s}\left(\tau\right)\}$
where $V\sim\mathcal{U}\left(0,1\right)$. \label{enu:Thin} 
\end{enumerate}
\item Return $\tau$. 
\end{enumerate}
\end{algorithm}

The case $V>\{\chi\left(\tau\right)/\bar{\chi}_{s}\left(\tau\right)\}$
corresponds to a rejection step in the thinning algorithm but, in
contrast to rejection steps that occur in standard MCMC samplers,
in the BPS algorithm this means that the particle does not bounce
and just coasts. Practically, we would like ideally $\triangle$ and
the ratio $\chi\left(\tau\right)/\bar{\chi}_{s}\left(\tau\right)$
to be large. Indeed this would avoid having to simulate too many candidate
events from $\bar{\Pi}_{s}$ which would be rejected as these rejection
steps incur a computational cost.

\subsubsection{Simulation using superposition and thinning\label{subsec:Simulation-using-composition}}

Assume that the energy can be decomposed as 
\begin{equation}
U\left(x\right)=\sum_{j=1}^{m}U^{[j]}\left(x\right),\label{eq:superposition-decomp}
\end{equation}
then 
\[
\chi\left(t\right)\le\sum_{j=1}^{m}\chi^{[j]}\left(t\right),
\]
where $\chi^{[j]}(t)=\max\left(0,\left\langle \nabla U^{[j]}(x+tv),v\right\rangle \right)$
for $j=1,...,m$. It is therefore possible to use the thinning algorithm
of Section \ref{subsec:Simulation-through-adaptive} with $\bar{\chi}_{0}(t)=\sum_{j=1}^{m}\chi^{[j]}\left(t\right)$
for $t\geq0$ (and $\Delta=+\infty$), as we can simulate from $\bar{\Pi}_{0}$
via superposition by simulating the first arrival time $\tau^{[j]}$
of each PP with intensity $\chi^{[j]}\left(t\right)\geq0$ then returning
\[
\tau=\mathrm{min}_{j=1,...,m}~\tau^{[j]}.
\]

\begin{example}
\emph{\label{ex:exp-fam}Exponential families}. Consider a univariate
exponential family with parameter $x$, observation $y,$ sufficient
statistic $\phi(y)$ and log-normalizing constant $A(x)$. If we assume
a Gaussian prior on $x,$ we obtain 
\[
U(x)=\underbrace{{x^{2}/2}}_{U^{[1]}(x)}+\underbrace{{-x\phi(y)}}_{U^{[2]}(x)}+\underbrace{{A(x)}}_{U^{[3]}(x)}.
\]
The time $\tau^{[1]}$ is computed analytically in Example \ref{ex:normal}
whereas the times $\tau^{[2]}$ and $\tau^{[3]}$ are given by 
\[
\tau^{[2]}=\begin{cases}
\frac{{\log V^{[2]}}}{v\phi(y)} & \text{if }v\phi(y)<0,\\
+\infty & \text{{otherwise,}}
\end{cases}
\]
and 
\begin{align*}
\tau^{[3]} & =\begin{cases}
\tilde{\tau}^{[3]} & \text{if }\tilde{\tau}^{[3]}>0,\\
+\infty & \text{{otherwise,}}
\end{cases}
\end{align*}
with $\tilde{\tau}^{[3]}=(A^{-1}(-\log V^{[3]}+A(x))-x)/v$ and $V^{[2]},V^{[3]}\sim\mathcal{U}\left(0,1\right)$.
For example, with a Poisson distribution with natural parameter $x,$
we obtain 
\[
\tilde{\tau}^{[3]}=\frac{\log(-\log V^{[3]}+\exp(x))-x}{v}.
\]
\end{example}
\begin{example}
\emph{\label{ex:Logistic-regression}Logistic regression}. The class
label of the data point $r\in\{1,2,\dots,R\}$ is denoted by $y_{r}\in\left\{ 0,1\right\} $
and its covariate $k\in\{1,2,\dots,d\}$ by $\iota_{r,k}$ where we
assume that $\iota_{r,k}\ge0$ (this assumption can be easily relaxed
as demonstrated but would make the notation more complicated; see
\citep{Galbraith2016} for details). The parameter $x\in\R^{d}$ is
assigned a standard Gaussian prior density denoted by $\psi$, yielding
the posterior density 
\begin{equation}
\pi(x)\propto\psi(x)\prod_{r=1}^{R}\frac{\exp(y{}_{r}\dotprod{\iota_{r}}{x})}{1+\exp\dotprod{\iota_{r}}{x}}.\label{eq:Logisticregressionmodel}
\end{equation}
Using the superposition and thinning method (Section \ref{subsec:Simulation-using-composition}),
simulation of the bounce times can be broken into subproblems corresponding
to $R+1$ factors: one factor coming from the prior, with corresponding
energy 
\begin{equation}
U^{[R+1]}(x)=-\log\psi(x)=\|x\|^{2}/2+\text{constant},\label{eq:LogisticPriorpotential}
\end{equation}
and $R$ factors coming from the likelihood of each datapoint, with
corresponding energy 
\begin{equation}
U^{[r]}(x)=\log(1+\exp\dotprod{\iota_{r}}{x})-y_{r}\dotprod{\iota_{r}}{x}.\label{eq:LogisticDatapotential}
\end{equation}
Simulation of $\tau^{[R+1]}$ is covered in Example \ref{ex:normal}.
Simulation of $\tau^{[r]}$ for $r\in\left\{ 1,2,\dots,R\right\} $
can be approached using thinning. In Appendix \ref{subsec:Bound-on-intensity},
we show that 
\begin{equation}
\chi^{[r]}(t)\le\bar{\chi}^{[r]}=\sum_{k=1}^{d}\1[v_{k}(-1)^{y_{r}}\ge0]\iota_{r,k}|v_{k}|.\label{eq:LogisticDatapotentialbound}
\end{equation}
Since the bound is constant for a given $v$, we sample $\tau^{[r]}$
by simulating an exponential random variable. 
\end{example}

\subsection{Estimating expectations\label{subsec:Estimating-expectations}}

Given a realization of $x\left(t\right)$ over the interval $\left[0,T\right]$,
where $T$ is the total trajectory length, the expectation $\int_{\mathbb{R}^{d}}\varphi\left(x\right)\pi\left({\rm d}x\right)$
of a function $\varphi:\mathbb{R}^{d}\rightarrow\mathbb{R}$ with
respect to $\pi$ can be estimated using 
\[
\frac{1}{T}\int_{0}^{T}\varphi\left(x\left(t\right)\right){\rm d}t=\frac{1}{T}\left(\sum_{i=1}^{n-1}\int_{0}^{\tau_{i}}\varphi\left(x^{\left(i-1\right)}+v^{\left(i-1\right)}s\right){\rm d}s+\int_{0}^{t_{n}-T}\varphi\left(x^{\left(n-1\right)}+v^{\left(n-1\right)}s\right){\rm d}s\right);
\]
see, e.g., \citep{davis1993markov}. When $\varphi\left(x\right)=x_{k}$,
$k\in\left\{ 1,2,\dots,d\right\} $, we have 
\[
\int_{0}^{\tau_{i}}\varphi\left(x^{\left(i-1\right)}+v^{\left(i-1\right)}s\right){\rm d}s=x_{k}^{\left(i-1\right)}\tau_{i}+v_{k}^{\left(i-1\right)}\frac{\tau_{i}^{2}}{2}.
\]
When the above integral is intractable, we may just discretize $x\left(t\right)$
at regular time intervals to obtain an estimator 
\[
\frac{1}{L}\sum_{l=0}^{L-1}\varphi\left(x\left(l\delta\right)\right),
\]
where $\delta>0$ is the mesh size and $L=1+\left\lfloor T/\delta\right\rfloor $.
Alternatively, we could approximate these univariate integrals through
quadrature.

\subsection{Theoretical results\label{subsec:Theoretical-results}}

An informal proof establishing that the BPS with $\lambda^{\mathrm{ref}}=0$
admits $\pi$ as invariant distribution is given in \cite{PetersDeWith2012}.
As the BPS process $z\left(t\right)=\left(x\left(t\right),v\left(t\right)\right)$
is a piecewise deterministic Markov process, the expression of its
infinitesimal generator can be established rigourously using \cite{davis1993markov}.
We show here that this generator has invariant distribution $\pi$
whenever $\lambda^{\mathrm{ref}}\geq0$ and prove that the resulting
process is additionally ergodic when $\lambda^{\mathrm{ref}}>0$.
We denote by $\mathbb{E}_{z}\left[h\left(z\left(t\right)\right)\right]$
the expectation of $h\left(z\left(t\right)\right)$ under the law
of the BPS process initialized at $z\left(0\right)=z$. 
\begin{prop}
\label{Proposition:piinvariance}For any $\lambda^{\mathrm{ref}}\geq0$,
the infinitesimal generator $\mathcal{L}$ of the BPS is defined for
any sufficiently regular bounded function $h:\mathbb{R}^{d}\times\mathbb{R}^{d}\rightarrow\mathbb{R}$
by 
\begin{eqnarray}
\mathcal{L}h(z) & = & \lim_{t\downarrow0}\frac{\mathbb{E}_{z}\left[h\left(z\left(t\right)\right)\right]-h(z)}{t}\nonumber \\
 & = & \left\langle \nabla_{x}h\left(x,v\right),v\right\rangle +\lambda\left(x,v\right)\left\{ h(x,R\left(x\right)v)-h(z)\right\} \nonumber \\
 &  & +\lambda^{\mathrm{ref}}\int\left(h(x,v')-h(x,v)\right)\psi\left(v'\right){\rm d}v',\label{eq:GeneratorBPS}
\end{eqnarray}
where we recall that $\psi\left(v\right)$ denotes the standard multivariate
Gaussian density on $\R^{d}$.

This transition kernel of the BPS is non-reversible and admits $\rho$
as invariant probability measure, where the density of \textup{$\rho$}
w.r.t. Lebesgue measure on\textup{ $\mathbb{R}^{d}\times\mathbb{R}^{d}$
}is\textup{ }given by
\begin{equation}
\rho(z)=\pi\left(x\right)\psi\left(v\right).\label{eq:invariantdensityBPS}
\end{equation}
\end{prop}
If we add the condition $\RInt>0$, we get the following stronger
result. 
\begin{thm}
\label{thm:uniqueInvariant}%
{} If $\lambda^{\mathrm{ref}}>0$ then $\rho$ is the unique invariant
probability measure of the transition kernel of the BPS and for $\rho$-almost
every $z\left(0\right)$ and $h$ integrable with respect to \textup{$\rho$}
\[
\lim_{T\rightarrow\infty}\frac{1}{T}\int_{0}^{T}h(z\left(t\right)){\rm d}t=\int h(z)\rho(z){\rm d}z\quad\text{ a.s.}
\]
\end{thm}
In fact, Lemma \ref{lem:refreshkernel} establishes a minorization so it is only left to establish a Lyapunov function in order to establish polynomial or geometric ergodicity in total variation.
We exhibit in Section \ref{subsec:Isotropic-Multivariate-Normal}
a simple example where $P_{t}$ is not ergodic for $\lambda^{\mathrm{ref}}=0$.

\section{The local bouncy particle sampler\label{sec:Local-bouncy-particle-sampler}}

\subsection{Structured target distribution and factor graph representation}

In numerous applications, the target distribution admits some structural
properties that can be exploited by sampling algorithms. For example,
the Gibbs sampler takes advantage of conditional independence properties.
We present here a ``local'' version of the BPS introduced in \citep[Section III]{PetersDeWith2012}
which can similarly exploit these properties and, more generally,
any representation of the target density as a product of positive
factors 
\begin{equation}
\pi\left(x\right)\propto\prod_{f\in F}\gamma_{f}\left(x_{f}\right),\label{eq:factorgraphtarget}
\end{equation}
where $x_{f}$ is a restriction of $x$ to a subset $N_{f}\subseteq\mbox{\{1,2,\ensuremath{\dots},\ensuremath{d\}}}$
of the components of $x$, and $F$ is an index set called the set
of factors. Hence the energy associated to $\pi$ is of the form 
\begin{equation}
U\left(x\right)=\sum_{f\in F}U_{f}\left(x\right)\label{eq:factorgraphenergy}
\end{equation}
with $\partial U_{f}\left(x\right)/\partial x_{k}=0$ for any variable
absent from factor $f$, i.e. for any $k\in\left\{ 1,2,\dots,d\right\} \backslash N_{f}$.
\begin{center}
\begin{figure}
\begin{centering}
\includegraphics[scale=0.7]{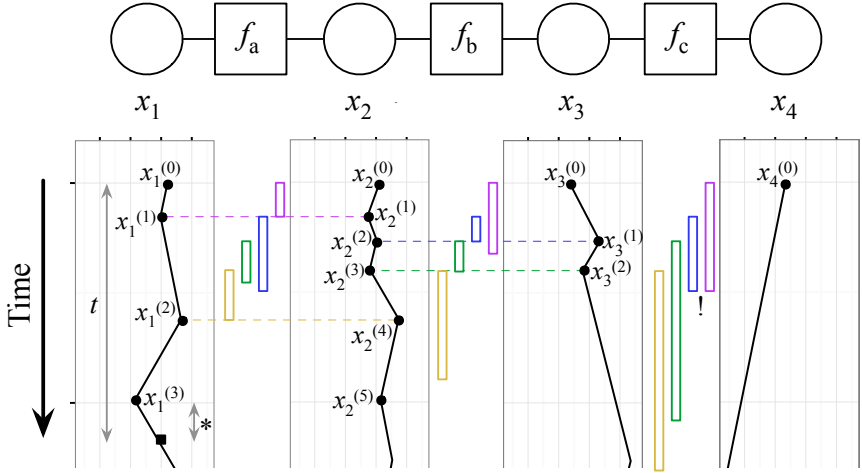} 
\par\end{centering}
\protect\caption{\label{fig:local-rf}Top: a factor graph with $d=4$ variables and
3 binary factors, $F=\left\{ f_{\text{{a}}},f_{\text{{b}}},f_{\text{{c}}}\right\} $.
Bottom: sample paths of $\left(x_{i}(t)\right)_{t\geq0}$ for $i=1,...,4$
for the local BPS. See Sections \ref{subsec:Local-BPS:-algorithm}
and \ref{subsec:Implementation-via-priorityqueue}.}
\end{figure}
\par\end{center}

Such a factorization of the target density can be formalized using
factor graphs (Figure \ref{fig:local-rf}, top). A factor graph is
a bipartite graph, with one set of vertices $N$ called the variables,
each corresponding to a component of $x$ ($|N|=d$), and a set of
vertices $F$ corresponding to the local factors $\left(\gamma_{f}\right)_{f\in F}$.
There is an edge between $k\in N$ and $f\in F$ if and only if $k\in N_{f}.$
This representation generalizes undirected graphical models \citep[Chap. 2, Section 2.1.3]{WainwrightJordan2008}
as, for example, factor graphs can have distinct factors connected
to the same set of components (i.e. $f\neq f'$ with $N_{f}=N_{f'}$)
as in the example of Section \ref{subsec:Logisticregression}.

\subsection{Local BPS: algorithm description\label{subsec:Local-BPS:-algorithm} }

Similarly to the Gibbs sampler, each step of the local BPS manipulates
only a subset of the $d$ components of $x$. Contrary to the Gibbs
sampler, the local BPS does not require sampling from any full conditional
distribution and each local calculation results in a change of \emph{all}
state components, not just the one being updated\textemdash how this
can be done implicitly without manipulating the full state at each
iteration is described below. Related processes exhibiting similar
characteristics have been proposed in \cite{kampmann2015monte,michel2014generalized,michel2015event,nishikawa2015event}.

For each factor $f\in F$, we define a local intensity function $\lambda_{f}:\mathbb{R}^{d}\times\mathbb{R}^{d}\rightarrow\mathbb{R}^{+}$
and a local bouncing matrix $R_{f}:\mathbb{R}^{d}\rightarrow\mathbb{R}^{d\times d}$
by 
\begin{align}
\lambda_{f}\left(x,v\right) & =\mathrm{max}\left\{ 0,\left\langle \nabla U_{f}\left(x\right),v\right\rangle \right\} ,\label{eq:localintensity}\\
R_{f}\left(x\right)v & =v-2\frac{\left\langle \nabla U_{f}\left(x\right),v\right\rangle \nabla U_{f}\left(x\right)}{\left\Vert \nabla U_{f}\left(x\right)\right\Vert ^{2}}.\label{eq:localprojectionoperator}
\end{align}
We can check that $R_{f}\left(x\right)$ satisfies 
\begin{equation}
k\in\left\{ 1,2,\dots,d\right\} \backslash N_{f}\Longrightarrow\{R_{f}(x)v\}_{k}=v_{k}.\label{eq:sparse-collision}
\end{equation}
When suitable, we will slightly abuse notation and write $R_{f}\left(x_{f}\right)$
for $R_{f}\left(x\right)$ as $R_{f}\left(x_{f},x_{-f}\right)=R_{f}\left(x_{f},x'_{-f}\right)$
for any $x_{-f},x'_{-f}\in\mathbb{R}^{d-\left|N_{f}\right|}$, where
$\left|S\right|$ denotes the cardinality of a set $S$. Similarly,
we will use $\lambda_{f}(x_{f},v_{f})$ for $\lambda_{f}(x,v)$.

We define a collection of PP intensities based on the previous event
position $x^{(i-1)}$ and velocity $v^{(i-1)}$: $\chi_{f}(t)=\lambda_{f}(x^{\left(i-1\right)}+v^{\left(i-1\right)}t,v^{\left(i-1\right)})$.
In the local BPS, the next bounce time $\tau$ is the first arrival
of a PP with intensity $\chi(t)=\sum_{f\in F}\chi_{f}(t)$. However,
instead of modifying all velocity variables at a bounce as in the
basic BPS, we sample a factor $f$ with probability $\chi_{f}(\tau)/\chi(\tau)$
and modify only the variables connected to the sampled factor. More
precisely, the velocity $v_{f}$ is updated using $R_{f}\left(x_{f}\right)$
defined in (\ref{eq:localprojectionoperator}). A generalization of
the proof of Proposition \ref{Proposition:piinvariance} given in
the Supplementary Material shows that the local BPS algorithm results
in a $\pi-$invariant kernel. In the next subsection, we describe
various computationally efficient procedures to simulate this process.

For all these implementations, it is useful to encode trajectories
in a sparse fashion: each variable $k\in N$ only records information
at the times $t_{k}^{(1)},t_{k}^{(2)},\dots$ where an event (a bounce
or refreshment) affected it. By (\ref{eq:sparse-collision}), this
represents a sublist of the list of all event times. At each of those
times $t_{k}^{(i)},$ the component's position $x_{k}^{(i)}$ and
velocity $v_{k}^{(i)}$ right after the event is stored. Let $L_{k}$
denote a list of triplets $(x_{k}^{(i)},v_{k}^{(i)},t_{k}^{(i)})_{i\geq0}$,
where $x_{k}^{(0)}$ and $v_{k}^{(0)}$ denote the initial position
and velocity and $t_{k}^{(0)}=0$ (see Figure \ref{fig:local-rf},
where the black dots denote the set of recorded triplets). This list
is sufficient to compute $x_{k}(t)$ for $t\leq t_{k}^{(|L_{k}|+1)}$.
This procedure is detailed in Algorithm \ref{alg:calculationstatedirection}
and an example is shown in Figure \ref{fig:local-rf}, where the black
square on the first variable's trajectory shows how Algorithm \ref{alg:calculationstatedirection}
reconstructs $x_{1}(t)$ at a fixed time $t$: it identifies $i(t,1)=3$
as the index associated to the largest event time $t_{1}^{(3)}$ before
time $t$ affecting $x_{1}$ and return $x_{1}\left(t\right)=x_{1}^{(3)}+v_{1}^{(3)}(t-t_{1}^{(3)})$.

\begin{algorithm}[H]
\protect\caption{Computation of $x_{k}(t)$ from a list of events. ~\label{alg:calculationstatedirection}}

\begin{enumerate}
\item Find the index $i=i(t,k)$ associated to the largest time $t_{k}^{(i)}$
verifying $t_{k}^{(i)}\le t$. 
\item Set $x_{k}(t)\leftarrow x_{k}^{(i(t,k))}+(t-t_{k}^{(i(t,k))})v_{k}^{(i(t,k))}.$ 
\end{enumerate}
\end{algorithm}

\subsection{Local BPS: efficient implementations}

\subsubsection{Implementation via priority queue\label{subsec:Implementation-via-priorityqueue}}

We can sample arrivals from a PP with intensity $\chi(t)=\sum_{f\in F}\chi_{f}(t)$
using the superposition method of Section \ref{subsec:Simulation-using-composition},
the thinning step therein being omitted. To implement this technique
efficiently, we store potential future bounce times (called ``candidates'')
$t_{f}$, one for each factor, in a priority queue $Q$: only a subset
of these candidates will join the lists $L_{k}$ which store past,
``confirmed'' events. We pick the the smallest time in $Q$ to determine
the next bounce time and the next factor $f$ to modify. The priority
queue structure ensures that finding the minimum element of $Q$ or
inserting/updating an element of $Q$ can be performed with computational
complexity $O(\log|F|)$. When a bounce occurs, a key observation
behind efficient implementation of the local BPS is that not all the
other candidate bounce times need to be resimulated. Suppose that
the bounce was associated with factor $f$. In this case, only the
candidate bounce times $t_{f'}$ corresponding to factors $f'$ with
$N_{f'}\cap N_{f}\neq\emptyset$ need to be resimulated. For example,
consider the first bounce in Figure \ref{fig:local-rf} (shown in
purple), which is triggered by factor $f_{\text{a}}$ (rectangles
represent candidate bounce times $t_{f}$; dashed lines connect bouncing
factors to the variables that undergo an associated velocity change).
Then only the velocities for the variables $x_{1}$ and $x_{2}$ need
to be updated. Therefore, only the candidate bounce times for factors
$f_{\text{a}}$ and $f_{\text{b}}$ need to be re-simulated while
the candidate bounce time for $f_{\text{c}}$ stays constant (this
is shown by an exclamation mark in Figure \ref{fig:local-rf}).

The method is detailed in Algorithm~\ref{alg:LocalBouncyParticleSampler}.
Several operations of the BPS such as step 4, 6.iii, 6.iv and 7.ii
can be easily parallelized.

\begin{algorithm}[H]
\protect\caption{Local BPS algorithm (priority queue implementation) ~\label{alg:LocalBouncyParticleSampler}}

\begin{enumerate}
\item Initialize $\left(x^{\left(0\right)},v^{\left(0\right)}\right)$ arbitrarily
on $\mathbb{R}^{d}\times\mathbb{R}^{d}$. 
\item Initialize the global clock $T\leftarrow0$. 
\item For $k\in N$ do
\begin{enumerate}
\item Initialize the list $L_{k}\leftarrow\left(x_{k}^{\left(0\right)},v_{k}^{\left(0\right)},T\right)$. 
\end{enumerate}
\item Set $Q\leftarrow\mathbf{new\ queue}\left(x^{\left(0\right)},v^{\left(0\right)},T\right)$. 
\item Sample $t_{\mathrm{ref}}\sim\mathrm{Exp\left(\lambda^{\mathrm{ref}}\right)}$. 
\item While more events $i=1,2,\ldots$ requested do
\begin{enumerate}
\item $\left(t,f\right)\leftarrow\mathbf{smallest\ candidate\ bounce\ time}\ \mathbf{and\ associated\ factor\ in}\ Q$.\label{enu:pop} 
\item Remove $(t,f)$ from $Q$.
\item Update the global clock, $T\gets t$. 
\item If $T<t_{\mathrm{ref}}$ then \label{enu:after-pop}
\begin{enumerate}
\item $\left(v_{f}\right)_{k}\gets v_{k}^{(|L_{k}|-1)}$ for all $k\in N_{f}$. 
\item $x_{f}\gets x_{f}(T)$ (computed using Algorithm \ref{alg:calculationstatedirection}). 
\item For $k\in N_{f}$ ~do
\begin{enumerate}
\item $x_{k}^{\left(\left|L_{k}\right|\right)}\leftarrow x_{k}^{(\left|L_{k}\right|-1)}+(T-t_{k}^{(\left|L_{k}\right|-1)})v_{k}^{(\left|L_{k}\right|-1)}$,
where $t_{k}^{(\left|L_{k}\right|-1)}$ and $v_{k}^{(\left|L_{k}\right|-1)}$
are retrieved from $L_{k}$. 
\item $v_{k}^{\left(\left|L_{k}\right|\right)}\leftarrow\left\{ R_{f}\left(x_{f}\right)v_{f}\right\} _{k}$. 
\item $L_{k}\leftarrow\left\{ L_{k},\left(x_{k}^{\left(\left|L_{k}\right|\right)},v_{k}^{\left(\left|L_{k}\right|\right)},T\right)\right\} $
(add the new sample to the list). 
\end{enumerate}
\item For $f'\in F:N_{f'}\cap N_{f}\neq\emptyset$\ (note: this includes
the update of $f$) do\label{enu:Collision-recomputation}
\begin{enumerate}
\item for all $k\in N_{f'}$. 
\item $x_{f'}\gets x_{f'}(T)$ (computed using Algorithm \ref{alg:calculationstatedirection}). 
\item Simulate the first arrival time $\tau_{f'}$ of a PP of intensity
$\lambda_{f'}\left(x_{f'}+tv_{f'},v_{f'}\right)$ on $\left[0,+\infty\right)$. 
\item Set in $Q$ the candidate bounce time associated to $f'$ to the value
$t_{f'}=T+\tau_{f'}$ . 
\end{enumerate}
\end{enumerate}
\item Else
\begin{enumerate}
\item Sample $v'\sim\mathcal{N}\left(0_{d},I_{d}\right)$. 
\item $Q\leftarrow\mathbf{new\ queue}\left(x\left(t_{\mathrm{ref}}\right),v',t_{\mathrm{ref}}\right)$
where $x\left(t_{\mathrm{ref}}\right)$ is computed using Algorithm
\ref{alg:calculationstatedirection}. 
\item Set $t_{\mathrm{ref}}\leftarrow t_{\mathrm{ref}}+\tau_{\mathrm{ref}}$
where $\tau_{\mathrm{ref}}\sim\mathrm{Exp\left(\lambda^{\mathrm{ref}}\right)}$. 
\end{enumerate}
\end{enumerate}
\item Return the samples encoded as the lists $L_{k},$ $k\in N$. 
\end{enumerate}
\end{algorithm}

\begin{algorithm}[H]
\protect\caption{New Queue $\left(x,v,T\right)$ ~\label{alg:NewQueue}}

\begin{enumerate}
\item For $f\in F$ do
\begin{enumerate}
\item $\left(v_{f}\right)_{k}\gets v_{k}^{(|L_{k}|-1)}$ for all $k\in N_{f}$. 
\item $x_{f}\gets x_{f}(T)$ (computed using Algorithm \ref{alg:calculationstatedirection}). 
\item Simulate the first arrival time $\tau_{f}$ of a PP of intensity $\lambda_{f}\left(x_{f}+tv_{f},v_{f}\right)$
on $\left[0,+\infty\right)$. 
\item Set in $Q$ the time associated to $f$ to the value $T+\tau_{f}$
. 
\end{enumerate}
\item Return $Q$. 
\end{enumerate}
\end{algorithm}

\subsubsection{Implementation via thinning\label{subsec:Implementation-via-thinning}}

When the number of factors involved in Step \ref{enu:Collision-recomputation}
is large, the previous queue-based implementation can be computationally
expensive. Implementing the local BPS in this setup is closely related
to the problem of simulating stochastic chemical kinetics and innovative
solutions have been proposed in this area. We adapt here the algorithm
proposed in \cite{thanh2015simulationrejection} to the local BPS
context. For ease of presentation, we present the algorithm without
refreshment and only detail the simulation of the bounce times. This
algorithm relies on the ability to compute local-in-time upper bounds
on $\lambda_{f}$ for all $f\in F$. More precisely, we assume that
given a current position $x$ and velocity $v$, and $\Delta\in(0,\infty]$,
we can find a positive number $\bar{\chi}_{f}$, such that for any
$t\in[0,\Delta)$, we have $\bar{\chi}_{f}\ge\lambda_{f}(x+vt,v)$.
We can also use this method on a subset $G$ of $F$ and combine it
with the previously discussed techniques to sample candidate bounce
times for factors in $F$\textbackslash{}$G$ but we restrict ourselves
to $G=F$ to simplify the presentation.

\begin{algorithm}[H]
\protect\caption{Local BPS algorithm (thinning implementation) ~\label{alg:local-bps-thinning}}

\begin{enumerate}
\item Initialize $\left(x^{\left(0\right)},v^{\left(0\right)}\right)$ arbitrarily
on $\mathbb{R}^{d}\times\mathbb{R}^{d}$. 
\item Initialize the global clock $T\leftarrow0$. 
\item Initialize $\bar{T}\leftarrow\triangle$ (time until which local upper
bounds are valid). 
\item Compute local-in-time upper bounds $\bar{\chi}_{f}$ for $f\in F$
such that $\bar{\chi}_{f}\ge\lambda_{f}(x^{(0)}+v^{(0)}t,v^{\left(0\right)})$
for all $t\in[0,\Delta)$. 
\item While more events $i=1,2,\ldots$ requested do
\begin{enumerate}
\item Sample $\tau\sim\mathrm{Exp\left(\bar{\chi}\right)}$ where $\bar{\chi}=\sum_{f\in F}\bar{\chi}_{f}$.
\label{enu:Sample-next-time} 
\item If $\left(T+\tau>\bar{T}\right)$ then
\begin{enumerate}
\item $x^{(i)}\gets x^{(i-1)}+v^{(i-1)}(\bar{T}-T)$. 
\item $v^{(i)}\gets v^{(i-1)}$. 
\item For all $f\in F$, update $\bar{\chi}_{f}$ to ensure that $\bar{\chi}_{f}\ge\lambda_{f}(x^{(i)}+v^{(i)}t,v^{(i)})$
for $t\in[0,\Delta)$. 
\item Set $T\leftarrow\bar{T}$, $\bar{T}$$\leftarrow\bar{T}+\triangle.$ 
\end{enumerate}
\item Else
\begin{enumerate}
\item $x^{(i)}\gets x^{(i-1)}+v^{(i-1)}\tau$. 
\item Sample $\mathcal{F}\in F$ where $\mathbb{P}\left(\mathcal{F}=f\right)=\bar{\chi}_{f}/\bar{\chi}.$
\label{enu:Sample-which-factor} 
\item If $V<\lambda_{\mathcal{F}}\left(x^{(i)},v^{(i-1)}\right)/\bar{\chi}_{\mathcal{F}}$
where $V\sim\mathcal{U}\left(0,1\right)$ then a bounce for factor
$\mathcal{F}$ occurs at time $T$.\label{enu:LocalBPSthinning}
\begin{enumerate}
\item $v^{(i)}$$\gets R_{\mathcal{F}}\left(x^{(i)}\right)v^{(i-1)}$. 
\item For all $f'\in F:N_{f'}\cap N_{\mathcal{F}}\neq\emptyset$, update
$\bar{\chi}_{f'}$ to ensure that $\bar{\chi}_{f'}\ge\lambda_{f'}(x^{(i)}+v^{(i)}t,v^{(i)})$
for $t\in[0,\bar{T}-T-\tau)$. 
\end{enumerate}
\item Else
\begin{enumerate}
\item $v^{(i)}\gets v^{(i-1)}$. 
\end{enumerate}
\item Set $T\leftarrow T+\tau$. 
\end{enumerate}
\end{enumerate}
\end{enumerate}
\end{algorithm}

Algorithm \ref{alg:local-bps-thinning} will be particularly useful
in scenarios where summing over the bounds (Step \ref{enu:Sample-next-time})
and sampling a factor (Step \ref{enu:Sample-which-factor}) can be
performed efficiently. A scenario where it is possible to implement
these two operations in constant time is detailed in Section \ref{subsec:Logisticregression}.
Another scenario where sampling quickly from $\mathcal{F}$ is feasible
is if the number of distinct upper bounds is much smaller than the
number of factors. For example, we only need to sample a factor $\mathcal{F}$
uniformly at random if $\Lambda=\bar{\chi}_{f}=\bar{\chi}_{f^{\prime}}$
for all $f,f'$ in $F$ and $\bar{\chi}=\left|F\right|\cdot\Lambda$
(that is no factor needs to be inspected in order to execute Algorithm
\ref{enu:Sample-which-factor}) and the thinning procedure in Step
\ref{enu:LocalBPSthinning} boils down to

\begin{equation}
V\leq\frac{\left|F\right|\max\left(0,\left\langle \nabla U_{f}(x^{(i)}),v^{\left(i-1\right)}\right\rangle \right)}{\bar{\chi}}.\label{eq:remupper-1}
\end{equation}
A related approach has been adopted in \cite{zigzag} for the analysis
of big data. In this particular scenario, an alternative local BPS
can also be implemented where $s>1$ factors $\mathcal{F}=\left(\mathcal{F}_{1},\dots,\mathcal{F}_{s}\right)$
are sampled uniformly at random without replacement from $F$, the
thinning occurs with probability 
\begin{equation}
\frac{\left|F\right|}{s\bar{\chi}}\max\left(0,\sum_{j=1}^{s}\left\langle \nabla U_{\mathcal{F}_{j}}(x^{(i)}),v^{\left(i-1\right)}\right\rangle \right)\label{eq:remupper-2}
\end{equation}
and the components of $x$ belonging to $N_{\mathcal{F}}$ bounce
based on $\sum_{j=1}^{s}\nabla U_{\mathcal{F}_{j}}(x)$. One can check
that the resulting dynamics preserves $\pi$ as an invariant distribution.
In contrast to $s=1$, this is not an implementation of local BPS
described in Algorithm \ref{alg:local-bps-thinning}, but instead
this corresponds to a local BPS update for a random partition of the
factors.

\section{Numerical results\label{sec:Applications}}

\subsection{Gaussian distributions and the need for refreshment\label{subsec:Isotropic-Multivariate-Normal}}

We consider an isotropic multivariate Gaussian target distribution,
$U\left(x\right)=\left\Vert x\right\Vert ^{2}$, to illustrate the
need for refreshment. Without refreshment, we obtain from Equation
(\ref{eq:ComputeCollisionsIsoGaussian}) 
\begin{eqnarray*}
\left\langle x^{(i)},v^{(i)}\right\rangle  & = & \begin{cases}
-\sqrt{-\log V_{i}} & \text{if }\left\langle x^{(i-1)},v^{(i-1)}\right\rangle \leq0,\\
-\sqrt{\left\langle x^{(i-1)},v^{(i-1)}\right\rangle ^{2}-\log V_{i}} & \text{otherwise,}
\end{cases}
\end{eqnarray*}
and

\[
\left\Vert x^{(i)}\right\Vert ^{2}=\begin{cases}
\left\Vert x^{(i-1)}\right\Vert ^{2}-\left\langle x^{(i-1)},v^{(i-1)}\right\rangle ^{2}-\log V_{i} & \text{ if }\left\langle x^{(i-1)},v^{(i-1)}\right\rangle \leq0,\\
\left\Vert x^{(i-1)}\right\Vert ^{2}-\log V_{i} & \text{otherwise,}
\end{cases}
\]

see Supplementary Material for details. In particular, these calculations
show that if $\left\langle x^{(i)},v^{(i)}\right\rangle \le0$ then
$\left\langle x^{(j)},v^{(j)}\right\rangle \le0$ for $j>i$ so that
$\|x^{(i)}\|^{2}=\left\Vert x^{(1)}\right\Vert ^{2}-\left\langle x^{(1)},v^{(1)}\right\rangle ^{2}-\log V_{i}$
for $i\ge2$. In particular for $x^{(0)}=e_{1}$ and $v^{(0)}=e_{2}$
with $e_{i}$ being elements of standard basis of $\mathbb{R}^{d}$,
the norm of the position at all points along the trajectory can never
be smaller than $1$ as illustrated in Figure \ref{fig:reducibleTrak}.
\begin{center}
\begin{figure}
\begin{centering}
\includegraphics[height=5.5cm]{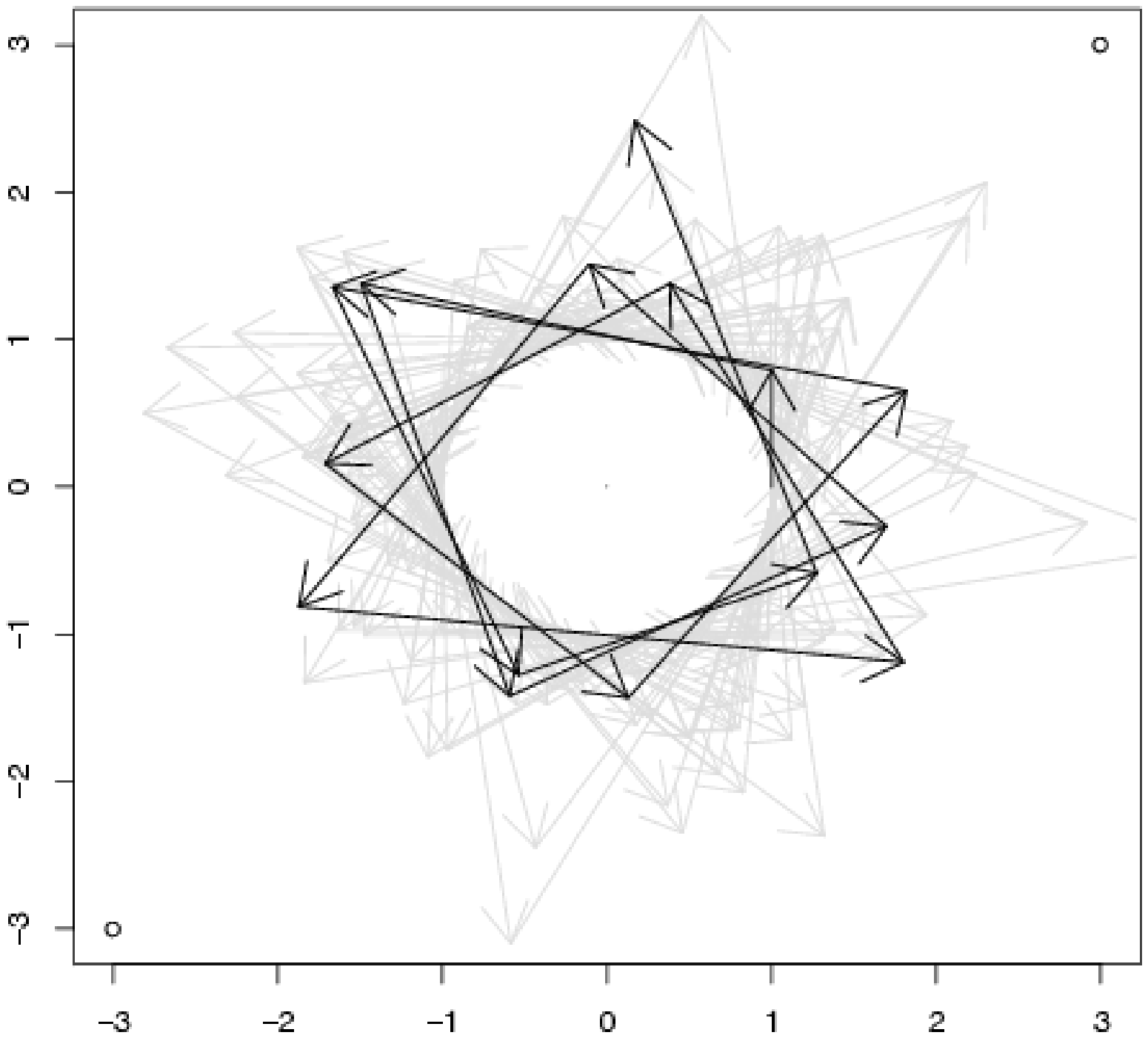}\hspace{0.5in}\includegraphics[height=5.5cm]{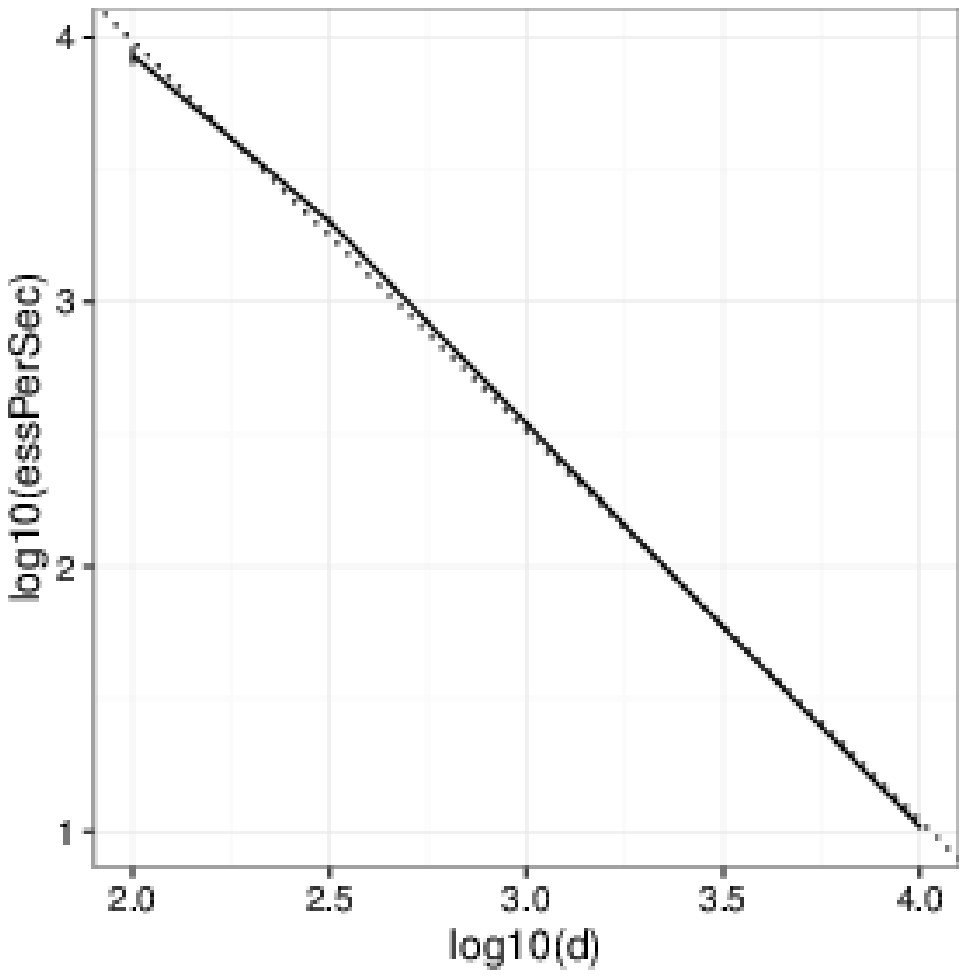} 
\par\end{centering}
\protect\caption{\label{fig:reducibleTrak}\label{fig:gaussian1}Left: the 200 first
segments/bounces of a BPS path for $\lambda^{\mathrm{ref}}=0$ (for
clarity the first 15 segments are in black, the following ones in
light grey): the center of the space is never explored. Right, solid
line: ESS per CPU second as a function of $d$ (log-log scale), along
with 95\% confidence intervals based on 40 runs (the intervals are
small and difficult to see). Dashed line: linear regression curve.
See Section \ref{subsec:Isotropic-Multivariate-Normal} for details.}
\end{figure}
\par\end{center}

In this scenario, we show that BPS without refreshment admits a countably
infinite collection of invariant distributions. Let us define $r\left(t\right)=\left\Vert x\left(t\right)\right\Vert $
and $m\left(t\right)=\left\langle x\left(t\right),v\left(t\right)\right\rangle /\left\Vert x\left(t\right)\right\Vert $
and denote by $\chi_{k}$ the probability density of the chi distribution
with $k$ degrees of freedom. 
\begin{prop}
\label{Proposition:infinitenumberinvariantmeasures.}For any dimension
$d\geq2$, the process $\left(r\left(t\right),m\left(t\right)\right)_{t\geq0}$
is Markov and its transition kernel is invariant with respect to the
probability densities $\left\{ f_{k}(r,m)\propto\chi_{k}(\sqrt{2}r)\cdot(1-\xv^{2})^{(k-3)/2};k\in\left\{ 2,3,\ldots\right\} \right\} $. 
\end{prop}
The proof is given in Appendix \ref{Proposition:infinitenumberinvariantmeasures.}.
By Theorem \ref{thm:uniqueInvariant}, we have a unique invariant
measure as soon as $\lambda^{\mathrm{ref}}>0$.

Next, we look at the scaling of the Effective Sample Size (ESS) per
CPU second of the basic BPS algorithm for $\varphi\left(x\right)=x_{1}$
when $\lambda_{\text{ref}}=1$ as the dimension $d$ of the isotropic
normal target increases. The ESS is estimated using the R package
\textit{mcmcse} \cite{ESSFlegal} by evaluating the trajectory on
a fine discretization of the sampled trajectory. The results in log-log
scale are displayed in Figure \ref{fig:gaussian1}. The curve suggests
a decay of roughly $d^{-1.47}$, slightly inferior to the $d^{-1.25}$
scaling for an optimally tuned Hamiltonian Monte Carlo (HMC) algorithm
\citep[Section III]{creutz1988}, \citep[Section 5.4.4]{Neal2011}.
It should be noted that BPS achieves this scaling without varying
any tuning parameter, whereas HMC's performance critically depends
on tuning two parameters (leap-frog stepsize and number of leap-frog
steps). Both BPS and HMC compare favorably to the $d^{-2}$ scaling
of the optimally tuned random walk MH \cite{roberts2001optimal}.

\subsection{Comparison of the global and local schemes\label{subsec:global-vs-local}}

We compare the basic ``global'' BPS of Section \ref{sec:Basic-bouncy-particle-sampler}
to the local BPS of Section \ref{sec:Local-bouncy-particle-sampler}
on a sparse Gaussian field. We use a chain-shaped undirected graphical
model of length $d=1000$ and perform separate experiments for various
pairwise precision parameters for the interaction between neighbors
in the chain. Both methods are run for 60 seconds. We compare the
Monte Carlo estimate of the variance of $x_{500}$ to its true value.
The results are shown in Figure \ref{fig:local-vs-global}. The smaller
computational complexity per local bounce of the local BPS offsets
significantly the associated decrease in expected trajectory segment
length. Moreover, both versions appear insensitive to the pairwise
precision used in this sparse Gaussian field.

\begin{figure}
\begin{centering}
\includegraphics[width=0.6\paperwidth]{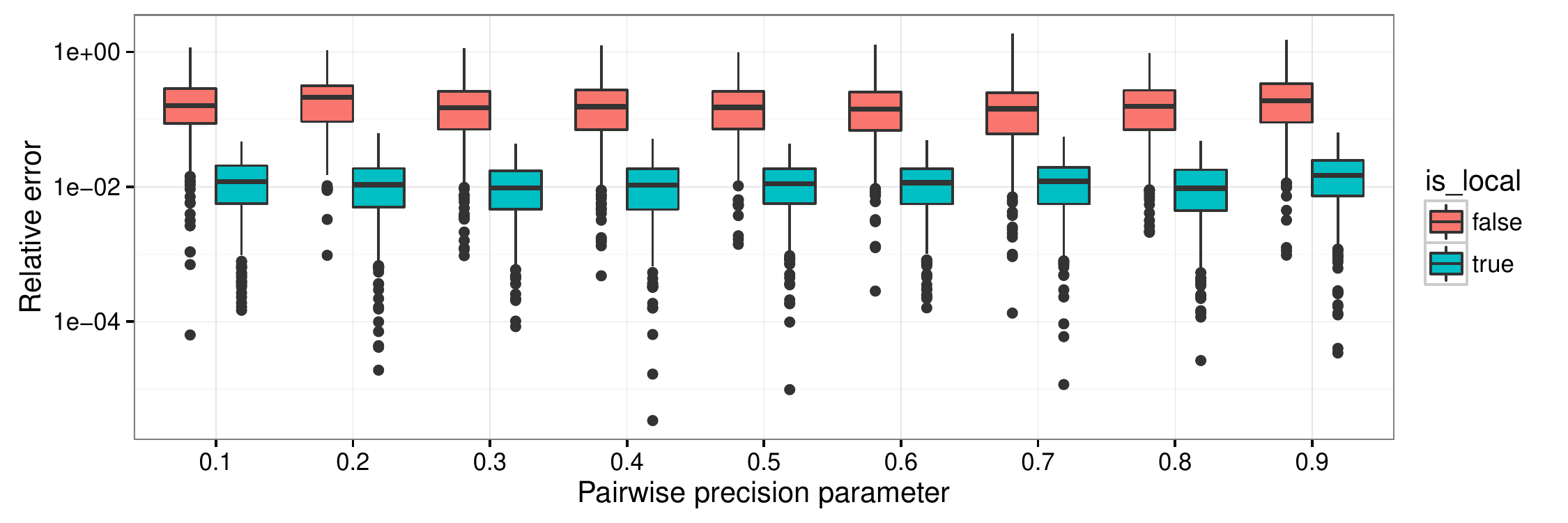} 
\par\end{centering}
\protect\caption{\label{fig:local-vs-global}Boxplots of relative errors over 100 local
BPS runs for Gaussian chain-shaped fields of pairwise precisions 0.1-0.9
.}
\end{figure}

\subsection{Comparisons of alternative refreshment schemes\label{subsec:Comparisons-of-refresh}}

In Section \ref{sec:Basic-bouncy-particle-sampler}, the velocity
was refreshed using a Gaussian distribution. We compare here this
global refreshment scheme to three alternatives: 
\begin{description}
\item [{Local refreshment:}] if the local BPS is used, the factor graph
structure can be exploited to design computationally cheaper refreshment
operators. We pick one factor $f\in F$ uniformly at random and resample
only the components of $v$ with indices in $N_{f}$. By the same
argument used in Section \ref{sec:Local-bouncy-particle-sampler},
each refreshment requires bounce time recomputation only for the factors
$f'$ with $N_{f}\cap N_{f'}\neq\emptyset$.%
{} 
\item [{Restricted refreshment:}] the velocities are refreshed according
to $\phi\left(v\right)$, the uniform distribution on $\mathcal{S}^{d-1}$,
and the BPS admits now $\rho\left(z\right)=\pi\left(x\right)\phi\left(v\right)$
as invariant distribution. 
\item [{Restricted partial refreshment:}] a variant of restricted refreshment
where we sample an angle $\theta$ by multiplying a $\text{{Beta}(\ensuremath{\alpha}, \ensuremath{\beta})}$-distributed
random variable by $2\pi.$ We then select a vector uniformly at random
from the unit length vectors that have an angle $\theta$ from $v$.
We used $\alpha=1,\beta=4$ to favor small angles. 
\end{description}
We compare these methods for different values of $\lambda^{\mathrm{ref}}$,
the trade-off being that too small a value can lead to a failure to
visit certain regions of the space, while too large a value leads
to a random walk behavior.

The rationale behind the partial refreshment procedure is to suppress
the random walk behavior of the particle path arising from a refreshment
step independent from the current velocity. Refreshment is needed
to ensure ergodicity but a ``good'' direction should only be altered
slightly. This strategy is akin to the partial momentum refreshment
strategy for HMC methods \cite{horowitz1991generalized}, \citep[Section 4.3]{Neal2011}
and could be similarly implemented for global refreshment. It is easy
to check that all of the above schemes preserve $\pi$ as invariant
distribution. We tested these schemes on the chain-shaped factor graph
described in the previous section (with the pairwise precision parameter
set to 0.5). All methods are provided with a computational budget
of 30 seconds. The results are shown in Figure \ref{fig:refresh-schemes}.
The results show that local refreshment is less sensitive to $\lambda^{\mathrm{ref}}$,
performing as well or better than global refreshment. The performance
of the restricted and partial methods appears more sensitive to $\lambda^{\mathrm{ref}}$
and generally inferior to the other two schemes.

One limitation of the results in this section is that the optimal
refreshment scheme and refreshment rate will in general be problem
dependent. Adaptation methods used in the HMC literature could potentially
be adapted to this scenario \cite{Wang2013AHMC,Hoffman2013NoUTurn},
but we leave these extensions to future work.

\begin{figure}
\begin{centering}
\includegraphics[width=0.7\paperwidth]{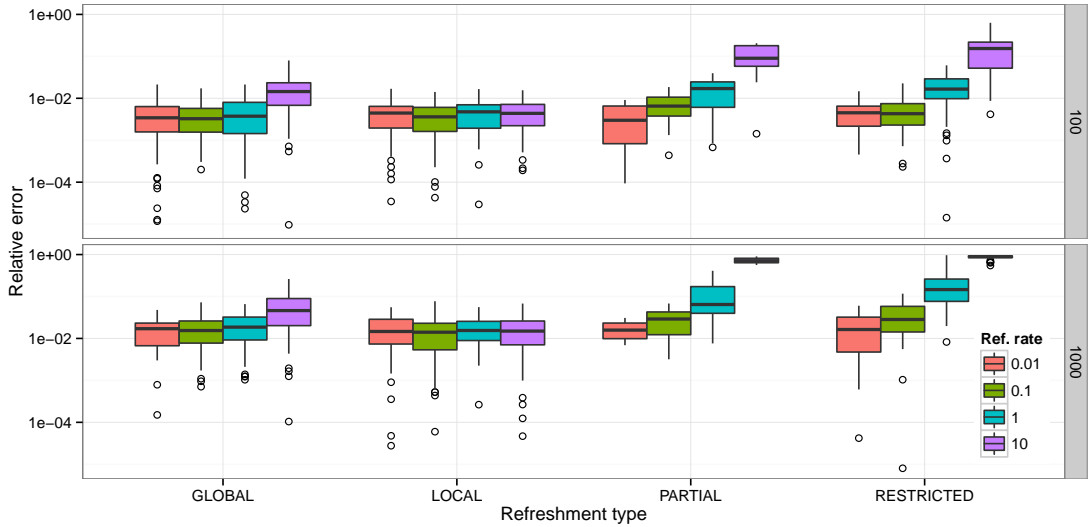} 
\par\end{centering}
\protect\caption{\label{fig:refresh-schemes}Comparison of refreshment schemes for
$d=100$ (top) and $d=1000$ (bottom). Each boxplot summarizes the
relative error for the variance estimates (in log scale) of $x_{50}$
over 100 runs of BPS.}
\end{figure}

\subsection{Comparisons with HMC methods on high-dimensional Gaussian distributions\label{subsec:Comparisons-with-HMC}}

We compare the local BPS with no partial refreshment and $\lambda^{\mathrm{ref}}=1$
to advanced adaptive versions of HMC implemented using Stan \cite{Hoffman2013NoUTurn}
on a 100-dimensional Gaussian example from \citep[Section 5.3.3.4]{Neal2011}.
For each method, we compute the relative error on the estimated marginal
variances after a wall clock time of 30 seconds, excluding from this
time the time taken to compile the Stan program. The adaptive HMC
methods use 1000 iterations of adaptation. When only the leap-frog
stepsize is adapted (``adapt=true''), HMC provides several poor
estimates of marginal variances. These deviations disappear when adapting
a diagonal metric (denoted ``fit-metric'') and/or using advanced
auxiliary variable methods to select the number of leap-frog steps
(denoted ``nuts''). Given that adaptation is critical to HMC in
this scenario, it is encouraging that BPS without adaptation is competitive
(Figure \ref{fig:radford-example}).

\begin{figure}
\begin{centering}
\includegraphics[width=0.7\paperwidth]{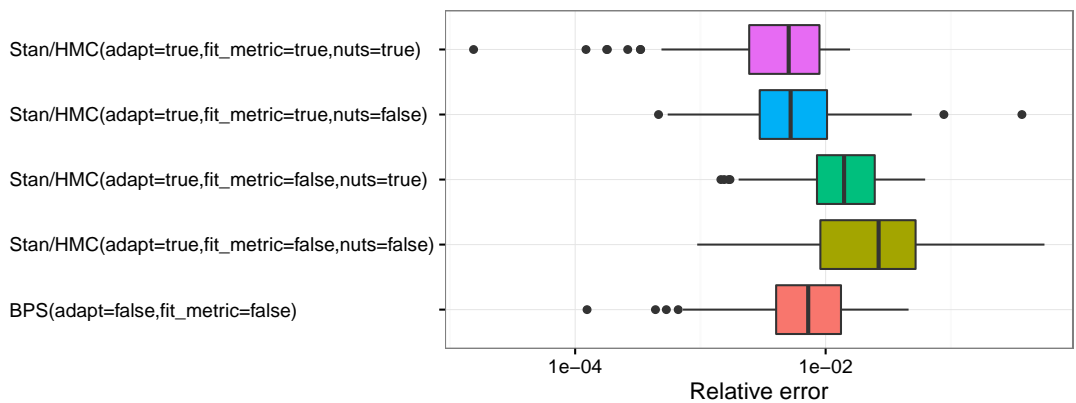} 
\par\end{centering}
\protect\caption{\label{fig:radford-example}Box plots showing the relative absolute
error of variance estimates for a fixed computational budget.}
\end{figure}

Next, we compare the local BPS to NUTS (``adapt=true,fit\_metric=true,nuts=true'')
as the dimension $d$ increases. Experiments are performed on the
chain-shaped Gaussian Random Field of Section \ref{subsec:global-vs-local}
with the pairwise precision parameter set to 0.5. We vary the length
of the chain (10, 100, 1000), and run Stan's implementation of NUTS
for 1000 iterations + 1000 iterations of adaptation. We measure the
wall-clock time (excluding the time taken to compile the Stan program)
and then run our method for the same wall-clock time 40 times for
each chain size. The absolute value of the relative error averaged
on 10 equally spaced marginal variances is measured as a function
of the percentage of the total computational budget used; see Figure
\ref{fig:vary-chain-length}. The gap between the two methods widens
as $d$ increases. To visualize the different behavior of the two
algorithms, three marginals of the Stan and BPS paths for $d=100$
are shown in Figure \ref{fig:example-paths}. Contrary to Section
\ref{subsec:Isotropic-Multivariate-Normal}, BPS outperforms here
HMC as its local version is able to exploit the sparsity of the random
field.

\begin{figure}
\begin{centering}
\includegraphics[width=0.6\paperwidth]{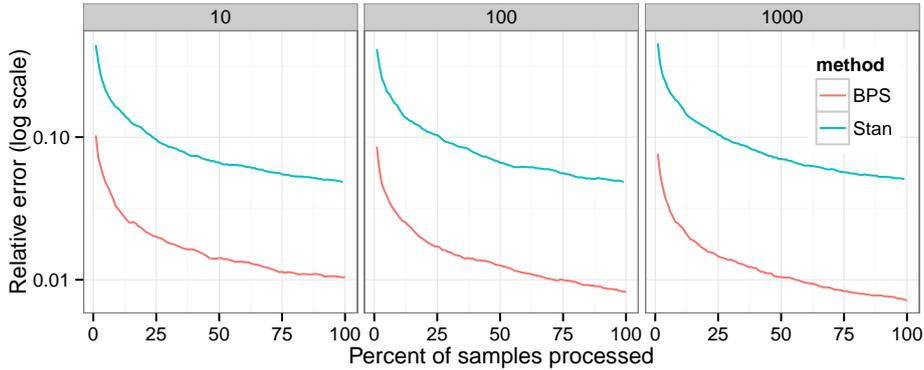} 
\par\end{centering}
\caption{\label{fig:vary-chain-length}Relative error for $d=10$ (left), $d=100$
(middle) and $d=1000$ (right), averaged over 10 of the dimensions
and 40 runs. Each $d$ uses a fixed computational budget.}
\end{figure}

\begin{figure}
\begin{centering}
\includegraphics[width=0.6\paperwidth]{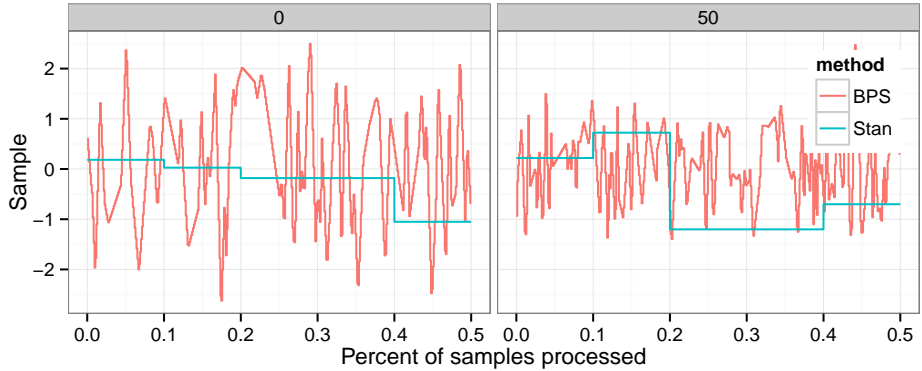} 
\par\end{centering}
\caption{\label{fig:example-paths}Simulated paths for $x_{0}$ and $x_{50}$
for $d=100$. Each state of the HMC trajectory is obtained by leap-frog
steps (not displayed), these latter cannot be used to estimate Monte
Carlo averages as HMC relies on a MH step. In contrast, BPS exploits
the full path.}
\end{figure}

\subsection{Poisson-Gaussian Markov random field\label{subsec:Poisson-Gaussian-Markov-random}}

Let $x=(x_{i,j}:i,j\in\left\{ 1,2,\dots,10\right\} )$ denote a grid-shaped
Gaussian Markov random field with pairwise interactions of the same
form as those used in the previous chain examples (pairwise precision
set to 0.5) and let $y_{i,j}$ be Poisson distributed, independent
over $i,j$ given $x$, with rate $\exp(x_{i,j})$. We generate a
synthetic dataset $y=(y_{i,j}:i,j\in\left\{ 1,2,\dots,10\right\} )$
from this model and approximate the resulting posterior distribution
of $x$. We first run Stan with default settings (``adapt=true,fit\_metric=true,nuts=true'')
for $16,32,64,\dots,4096$ iterations. For each number of Stan iterations,
we run local BPS for the same wall-clock time as Stan, using a local
refreshment ($\lambda^{\mathrm{ref}}=1$) and the method from Example
\ref{ex:exp-fam} for the bouncing time computations. We repeat these
experiments 10 times with different random seeds. Figure \ref{fig:poisson-gaussian}
displays the boxplots of the estimates of the posterior variances
of the variables $x_{0,0}$ and $x_{5,5}$ summarizing the 10 replications.
As expected, both methods converge to the same value, but BPS requires
markedly less computing time to achieve any given level of accuracy.

\begin{figure}
\begin{centering}
\includegraphics[width=0.7\paperwidth]{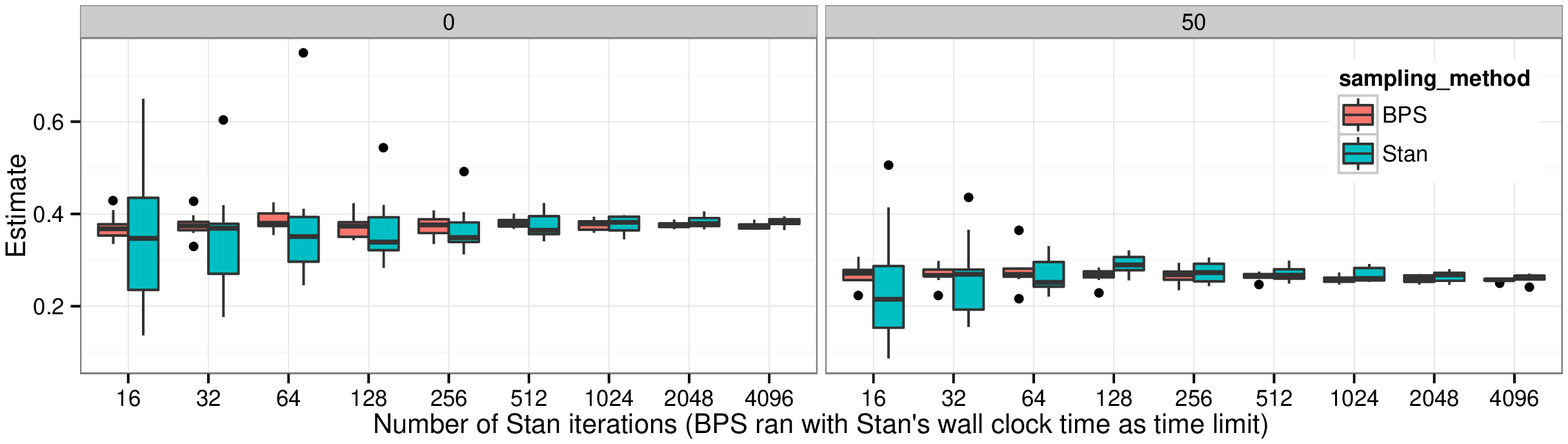} 
\par\end{centering}
\caption{\label{fig:poisson-gaussian}Boxplots of estimates of the posterior
variance of $x_{0,0}$ (left) and $x_{5,5}$ (right) using Stan implementation
of HMC and local BPS.}
\end{figure}

\subsection{Bayesian logistic regression for large data sets\label{subsec:Logisticregression}}

Consider the logistic regression model introduced in Example \ref{ex:Logistic-regression}
when the number of data $R$ is large. In this context, standard MCMC
schemes such as the MH algorithm are computationally expensive as
they require evaluating the likelihood associated to the $R$ observations
at each iteration. This has motivated the development of techniques
which only evaluate the likelihood of a subset of the data at each
iteration. However, most of the methods currently available introduce
either some non-vanishing asymptotic bias, e.g. the subsampling MH
scheme proposed in \cite{BarDouHol2014a}, or provide consistent
estimates converging at a slower rate than regular MCMC algorithms,
e.g. the Stochastic Gradient Langevin Dynamics introduced in \cite{welling2011bayesian,TehThiVol2015a}.
The only available algorithm which only requires evaluating the likelihood
of a subset of data at each iteration yet provides consistent estimates
converging at the standard Monte Carlo rate is the Firefly algorithm
\cite{maclaurin2014firefly}.

\textcolor{black}{In this context, we associate $R+1$ factors to
the target posterior distribution: one for the prior and one for each
data point with $x_{f}=x$ for all $f\in F$. }As a uniform upper
bound on the intensities of these local factors is available for restricted
refreshment, se\textcolor{black}{e Appendix \ref{subsec:Bound-on-intensity},
}we could use \textcolor{black}{\eqref{eq:remupper-2}} in conjunction
with Algorithm \ref{alg:local-bps-thinning} \textcolor{black}{to
provide an alternative to the Firefly algorithm which selects at each
bounce a subset of $s$ data points uniformly at random without replacement.
For $s=1$, a related algorithm has been recently explored in \cite{zigzag}.
In presence of outliers, this strategy can be inefficient as the uniform
upper bound becomes very large, resulting in a computationally expensive
implementation. }After a pre-computation step of complexity $O(R\log R)$
only executed once,\textcolor{black}{{} w}e show here that Algorithm
\ref{alg:local-bps-thinning} can be implemented using data-dependent
bounds \textcolor{black}{mitigating the sensitivity to outliers while
maintaining }the computational cost of each bounce independent of
$R$. We first pre-compute the sum of covariates over the data points,
$\iota_{k}^{c}=\sum_{r=1}^{R}\iota_{r,k}\1[y_{r}=c]$, for $k\in\left\{ 1,\ldots,d\right\} $
and class label $c\in\left\{ 0,1\right\} $. Using these quantities,
it is possible to compute 
\begin{align*}
\bar{\chi} & =\sum_{r=1}^{R}\bar{\chi}^{[r]}=\sum_{k=1}^{d}|v_{k}|\ \iota_{k}^{\1[v_{k}<0]},
\end{align*}

with $\bar{\chi}^{[r]}$ given in (\ref{eq:LogisticDatapotentialbound}).
If $d$ is large, we can keep the sum $\bar{\chi}$ in memory and
add-and-subtract any local updates to it. The implementation of Step
\ref{enu:Sample-which-factor} relies on the alias method \citep[Section 3.4]{Devroye1986}.
A detailed description of these derivations and of the algorithm is
presented in Appendix \ref{Appendix:Logistic}.

We compare the local BPS with thinning to the MAP-tuned Firefly algorithm
implementation provided by the authors. This version of Firefly outperforms
experimentally significantly the standard MH in terms of ESS per-datum
likelihood \cite{maclaurin2014firefly}. The two algorithms are here
compared in terms of this criterion, where the ESS is averaged over
the $d=5$ components of $x$. We generate covariates $\iota_{rk}\overset{\text{i.i.d.}}{\sim}\mathcal{U}(0.1,1.1)$
and data $y_{r}\in\{0,1\}$ for $r=1,\dots,R$ according to (\ref{eq:Logisticregressionmodel})
and set a zero-mean normal prior of covariance $\sigma^{2}I_{d}$
for $x$. For the algorithm, we set $\lambda_{\text{ref}}=0.5,\,\sigma^{2}=1$
and $\Delta=0.5$, which is the length of the time interval for which
a constant upper bound for the rate associated with the prior is used,
see Algorithm \ref{alg:local-bps-logistic}. Experimentally, local
BPS always outperforms Firefly, by about an order of magnitude for
large data sets. However, we also observe that both Firefly and local
BPS have an ESS per datum likelihood evaluation decreasing in approximately
$1/R$ so that the gains brought by these algorithms over a correctly
scaled random walk MH algorithm do not appear to increase with $R$.
The rate for local BPS is slightly superior in the regime of up to
$10^{4}$ data points, but then returns to the approximate $1/R$
rate. To improve this rate, one can adapt the control variate ideas
introduced in \textcolor{black}{\cite{BarDouHol2015}} for the MH
algorithm to these schemes. This has been proposed in \textcolor{black}{\cite{zigzag}}
for a related algorithm and in \cite{Galbraith2016} for the local
BPS.
\begin{center}
\begin{figure}
\begin{centering}
\includegraphics[width=0.7\textwidth]{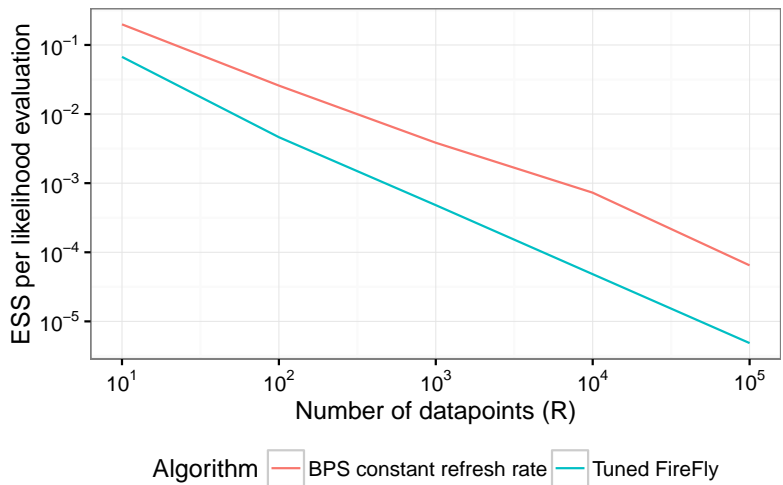} 
\par\end{centering}
\protect\caption{\label{fig:reducibleTrak-1}ESS per-datum likelihood evaluation for
Local BPS and Firefly. }
\end{figure}
\par\end{center}

\subsection{Bayesian inference of evolutionary parameters\label{subsec:Real-data-example:}}

We analyze a dataset of primate mitochondrial DNA \cite{Hayasaka1988Primates}
at the leaves of a fixed reference tree \cite{Huelsenbeck2001MB}
containing 898 sites and 12 species. We want to approximate the posterior
evolutionary parameters $x$ encoded into a rate matrix $Q$. A detailed
description of this phylogenetic model can be found in the Supplementary
Material. The global BPS is used with restricted refreshment and $\lambda^{\mathrm{ref}}=1$
in conjunction with an auxiliary variable-based method similar to
the one described in \cite{Zhao2015CTMC}, alternating between two
moves: (1) sampling continuous-time Markov chain paths along the tree
given $x$ using uniformization, (2) sampling $x$ given the path
(in which case the derivation of the gradient is simple and efficient).
The only difference compared to \cite{Zhao2015CTMC} is that we substitute
the HMC kernel by the kernel induced by running BPS for a fixed trajectory
length. This auxiliary variable method is convenient because, conditioned
on the paths, the energy function is convex and hence we can simulate
the bouncing times using the method described in Example \ref{ex:log-concave}.

We compare against a state-of-the-art HMC sampler \cite{Wang2013AHMC}
that uses Bayesian optimization to adapt the key parameters of HMC,
the leap-frog stepsize and the number of leap-frog steps, while preserving
convergence to the target distribution. Both our method and this HMC
method are implemented in Java and share the same gradient computation
code. Refer to the Supplementary Material for additional background
and motivation behind this adaptation method.

We first perform various checks to ensure that both BPS and HMC chains
are in close agreement given a sufficiently large number of iterations.
After 20 millions HMC iterations, the credible intervals estimates
from the HMC method are in close agreement with those obtained from
BPS (result not shown) and both methods pass the Geweke diagnostic
\cite{Geweke1992}.

\begin{figure}
\begin{centering}
\includegraphics[width=3in]{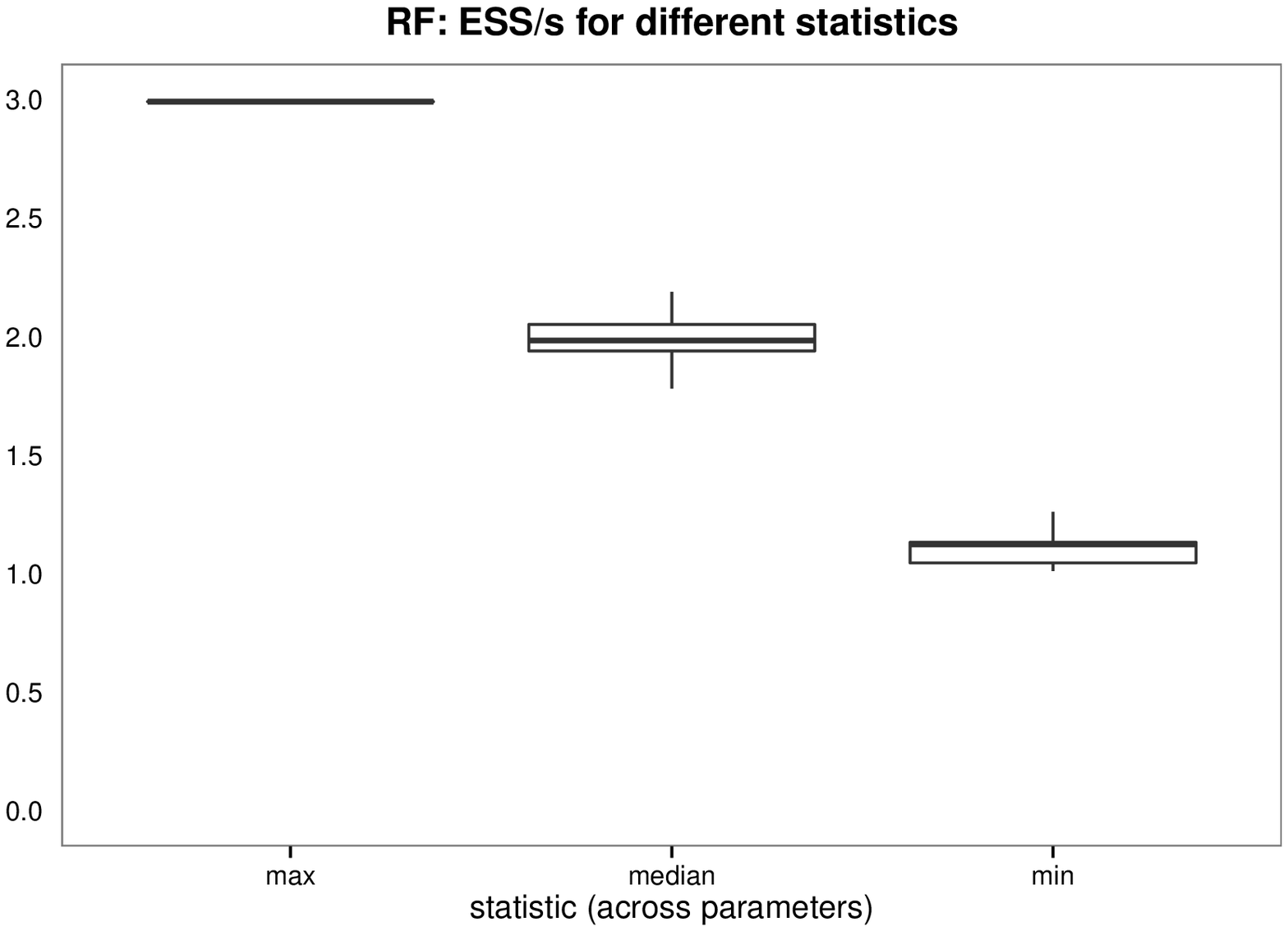}\includegraphics[width=3in]{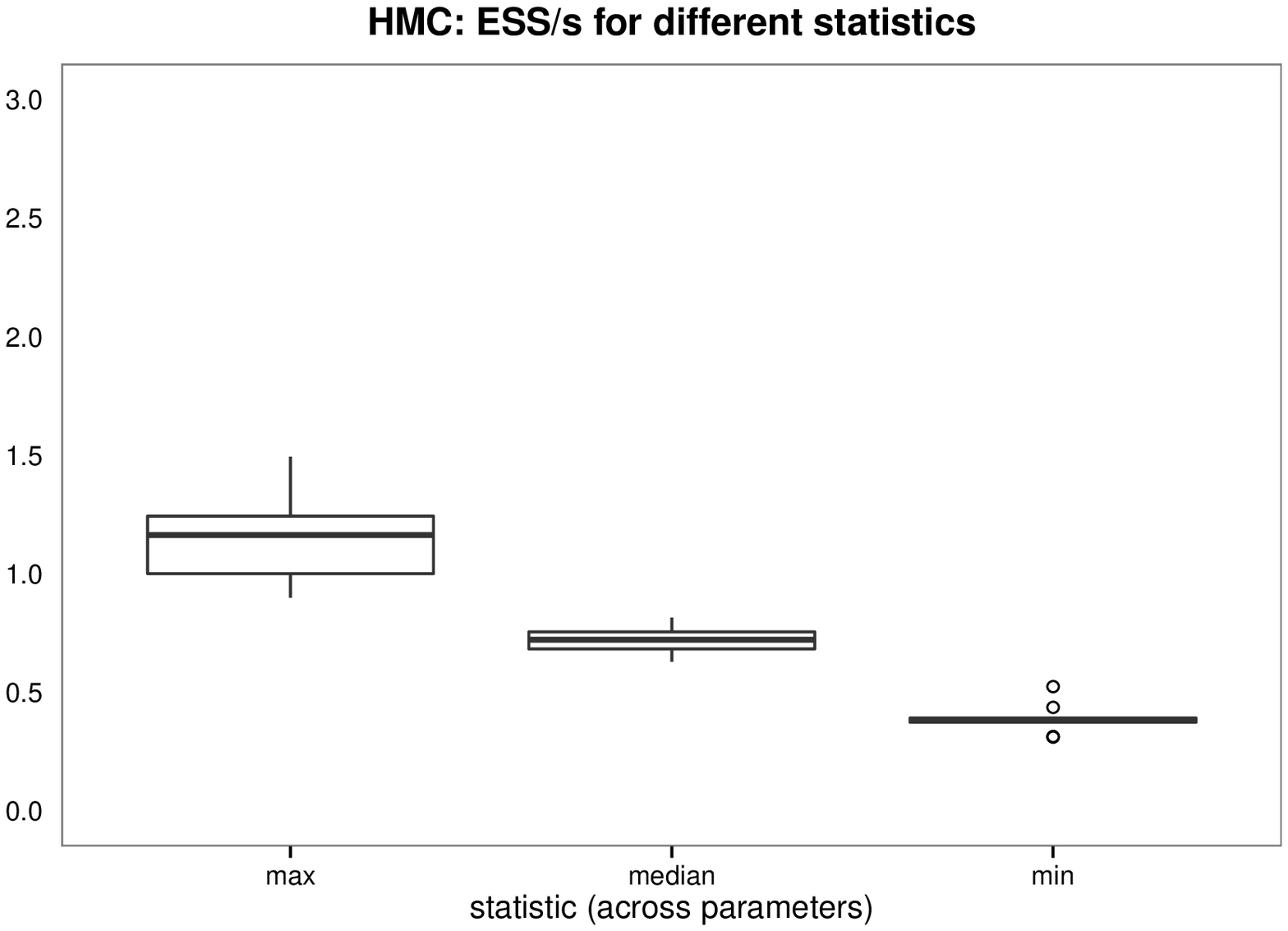} 
\par\end{centering}
\caption{\label{fig:esss}Boxplots of maximum, median and minimum ESS per second
for BPS (left) and HMC (right).}
\end{figure}

To compare the effectiveness of the two samplers, we look at the ESS
per second of the model parameters. We show the maximum, median, and
maximum over the 10 parameter components for 10 runs, for both BPS
and HMC in Figure \ref{fig:esss}. We observe a speed-up by a factor
two for all statistics considered (maximum, median, minimum). In the
supplement, we show that the HMC chain displays much larger autocorrelations
than the BPS chain.

\section{Discussion}

Most MCMC methods currently available, such as the MH and HMC algorithms,
are discrete-time reversible processes. There is a wealth of theoretical
results showing that non-reversible Markov processes mix faster and
provide lower variance estimates of ergodic averages \cite{Ottobre2016}.
However, most of the non-reversible processes studied in the literature
are diffusions and cannot be simulated exactly. The BPS is an alternative
continuous-time Markov process which, thanks to its piecewise deterministic
paths, can be simulated exactly for many problems of interest in statistics.

As any MCMC method, the BPS can struggle in multimodal scenarios and
when the target exhibits very strong correlations. However, for a
range of applications including sparse factor graphs, large datasets
and high-dimensional settings, we have observed empirically that BPS
is on par or outperforms state-of-the art methods such as HMC and
Firefly. The main practical limitation of the BPS compared to HMC
is that its implementation is model-specific and requires more than
knowing $\nabla U$ pointwise. An important open problem is therefore
whether its implementation, and in particular the simulation of bouncing
times, can be fully automated. However, the techniques described in
Section \ref{subsec:Algorithm-implementation} are already sufficient
to handle many interesting models. There are also numerous potential
methodological extensions of the method to study. In particular, it
has been shown in \cite{girolami2011riemann} how one can exploit
the local geometric structure of the target to improve HMC and it
would be interesting to investigate how this could be achieved for
BPS. More generally, the BPS is a specific continuous-time piecewise
deterministic Markov process \cite{davis1993markov}. This class
of processes deserves further exploration as it might provide a whole
new class of efficient MCMC methods.

\section*{Acknowledgements}

Alexandre Bouchard-Côté's research is partially supported by a Discovery
Grant from the National Science and Engineering Research Council.
Arnaud Doucet's research is partially supported by the Engineering
and Physical Sciences Research Council (EPSRC), grants EP/K000276/1,
EP/K009850/1 and by the Air Force Office of Scientific Research/Asian
Office of Aerospace Research and Development, grant AFOSRA/AOARD-144042.
Sebastian Vollmer's research is partially supported by the EPSRC grants
EP/K009850/1 and EP/N000188/1. We thank Markus Upmeier for helpful
discussions of differential geometry as well as Nicholas Galbraith,
Fan Wu, and Tingting Zhao for their comments.

\bibliographystyle{plain}
\bibliography{rf}

\begin{thebibliography}{10}

\bibitem{Adler2011coareaREf}
R.J. Adler and J.E. Taylor.
\newblock {\em Topological Complexity of Smooth Random Functions}, volume 2019
  of {\em Lecture Notes in Mathematics}.
\newblock Springer, Heidelberg, 2011.
\newblock {\'E}cole d'{\'E}t{\'e} de Probabilit{\'e}s de Saint-Flour XXXIX.

\bibitem{BarDouHol2014a}
R.~Bardenet, A.~Doucet, and C.C. Holmes.
\newblock Towards scaling up {MCMC}: An adaptive subsampling approach.
\newblock In {\em Proceedings of the 31st International Conference on Machine
  Learning}, 2014.

\bibitem{BarDouHol2015}
R.~Bardenet, A.~Doucet, and C.C. Holmes.
\newblock On {M}arkov chain {M}onte {C}arlo methods for tall data.
\newblock 2015.
\newblock Technical report arXiv:1505.02827.

\bibitem{zigzag}
J.~Bierkens, P.~Fearnhead, and G.~O. Roberts.
\newblock The zig-zag process and super-efficient {M}onte {C}arlo for
  {B}ayesian analysis of big data.
\newblock 2016.
\newblock Technical report arXiv:1607.03188.

\bibitem{bpsv1}
A.~Bouchard-C\^{o}t\'{e}, S.J. Vollmer, and A.~Doucet.
\newblock The bouncy particle sampler: a non-reversible rejection-free {M}arkov
  chain {M}onte {C}arlo method.
\newblock 2015.
\newblock Technical report arxiv:1510.02451v1.

\bibitem{creutz1988}
M.~Creutz.
\newblock Global {M}onte {C}arlo algorithms for many-fermion systems.
\newblock {\em Physical Review D}, 38(4):1228--1237, 1988.

\bibitem{davis1993markov}
M.H.A. Davis.
\newblock {\em Markov Models \& Optimization}, volume~49.
\newblock CRC Press, 1993.

\bibitem{Devroye1986}
L.~Devroye.
\newblock {\em Non-uniform Random Variate Generation}.
\newblock Springer-Verlag, New York, 1986.

\bibitem{birkhoff}
M.~Einsiedler and T.~Ward.
\newblock {\em Ergodic Theory: with a view towards Number Theory}, volume 259
  of {\em Graduate Texts in Mathematics}.
\newblock Springer-Verlag, London, 2011.

\bibitem{ESSFlegal}
J.M. Flegal, J.~Hughes, and D.~Vats.
\newblock {\em mcmcse: Monte Carlo Standard Errors for MCMC}, 2015.
\newblock R package version 1.1-2.

\bibitem{Galbraith2016}
N.~Galbraith.
\newblock On {E}vent-{C}hain {M}onte {C}arlo {M}ethods.
\newblock Master's thesis, Department of Statistics, Oxford University, 9 2016.

\bibitem{Geweke1992}
J.~Geweke.
\newblock Evaluating the accuracy of sampling-based approaches to the
  calculation of posterior moments.
\newblock {\em Bayesian Statistics}, 4:169--193, 1992.

\bibitem{girolami2011riemann}
M.~Girolami and B.~Calderhead.
\newblock Riemann manifold {L}angevin and {H}amiltonian {M}onte {C}arlo
  methods.
\newblock {\em Journal of the Royal Statistical Society: Series B (Statistical
  Methodology)}, 73(2):123--214, 2011.

\bibitem{Hairer2010ErgodicLectureNotes}
M.~Hairer.
\newblock Convergence of {M}arkov processes.
\newblock \url{http://www.hairer.org/notes/Convergence.pdf}, 2010.
\newblock Lecture notes, University of Warwick.

\bibitem{Hayasaka1988Primates}
K.~Hayasaka, Takashi G., and Satoshi H.
\newblock Molecular phylogeny and evolution of primate mitochondrial {DNA}.
\newblock {\em Molecular Biology and Evolution}, 5:626--644, 1988.

\bibitem{Hoffman2013NoUTurn}
M.D. Hoffman and A.~Gelman.
\newblock The no-{U}-turn sampler: Adaptively setting path lengths in
  {H}amiltonian {M}onte {C}arlo.
\newblock {\em Journal of Machine Learning Research}, 15(4):1593--1623, 2014.

\bibitem{horowitz1991generalized}
A.M. Horowitz.
\newblock A generalized guided {M}onte {C}arlo algorithm.
\newblock {\em Physics Letters B}, 268(2):247--252, 1991.

\bibitem{Huelsenbeck2001MB}
J.P. Huelsenbeck and F.~Ronquist.
\newblock {MRBAYES:} {B}ayesian inference of phylogenetic trees.
\newblock {\em Bioinformatics}, 17(8):754--755, 2001.

\bibitem{kampmann2015monte}
T.A. Kampmann, H.H. Boltz, and J.~Kierfeld.
\newblock Monte {C}arlo simulation of dense polymer melts using event chain
  algorithms.
\newblock {\em Journal of Chemical Physics}, (143):044105, 2015.

\bibitem{liu2008monte}
J.~S. Liu.
\newblock {\em Monte Carlo Strategies in Scientific Computing}.
\newblock Springer, 2008.

\bibitem{maclaurin2014firefly}
D.~Maclaurin and R.P. Adams.
\newblock Firefly {M}onte {C}arlo: Exact {MCMC} with subsets of data.
\newblock In {\em Uncertainty in Artificial Intelligence}, volume~30, pages
  543--552, 2014.

\bibitem{michel2014generalized}
M.~Michel, S.C. Kapfer, and W.~Krauth.
\newblock Generalized event-chain {M}onte {C}arlo: Constructing rejection-free
  global-balance algorithms from infinitesimal steps.
\newblock {\em Journal of Chemical Physics}, 140(5):054116, 2014.

\bibitem{michel2015event}
M.~Michel, J.~Mayer, and W.~Krauth.
\newblock Event-chain {M}onte {C}arlo for classical continuous spin models.
\newblock {\em Europhysics Letters}, (112):20003, 2015.

\bibitem{Neal2011}
R.M. Neal.
\newblock Markov chain {M}onte {C}arlo using {H}amiltonian dynamics.
\newblock In {\em Handbook of Markov Chain Monte Carlo}. Chapman \& Hall/CRC,
  2011.

\bibitem{nishikawa2015event}
Y.~Nishikawa, M.~Michel, W.~Krauth, and K.~Hukushima.
\newblock Event-chain {M}onte {C}arlo algorithm for the {H}eisenberg model.
\newblock {\em Physical Review E}, (92):063306, 2015.

\bibitem{Ottobre2016}
M.~Ottobre.
\newblock Markov chain {M}onte {C}arlo and irreversibility.
\newblock {\em Reports on Mathematical Physics}, 77:267--292, 2016.

\bibitem{PetersDeWith2012}
E.A. J.~F. Peters and G.~de~With.
\newblock Rejection-free {M}onte {C}arlo sampling for general potentials.
\newblock {\em Physical Review E}, 85:026703, 2012.

\bibitem{roberts2001optimal}
G.~O. Roberts and J.S. Rosenthal.
\newblock Optimal scaling for various {M}etropolis-{H}astings algorithms.
\newblock {\em Statistical Science}, 16(4):351--367, 2001.

\bibitem{Tavare1986GTR}
S.~Tavar\'{e}.
\newblock Some probabilistic and statistical problems in the analysis of
  {D}{N}{A} sequences.
\newblock {\em Lectures on Mathematics in the Life Sciences}, 17:56--86, 1986.

\bibitem{TehThiVol2015a}
Y.~W. Teh, A.~H. Thi\'ery, and S.~J. Vollmer.
\newblock Consistency and fluctuations for stochastic gradient {L}angevin
  dynamics.
\newblock {\em Journal of Machine Learning Research}, 17:1--33, 2016.

\bibitem{thanh2015simulationrejection}
V.H. Thanh and C.~Priami.
\newblock Simulation of biochemical reactions with time-dependent rates by the
  rejection-based algorithm.
\newblock {\em Journal of Chemical Physics}, 143(5):054104, 2015.

\bibitem{WainwrightJordan2008}
M.J. Wainwright and M.I. Jordan.
\newblock Graphical models, exponential families, and variational inference.
\newblock {\em Foundations and Trends in Machine Learning}, 1(1-2):1--305,
  2008.

\bibitem{Wang2013AHMC}
Z.~Wang, S.~Mohamed, and N.~de~Freitas.
\newblock Adaptive {H}amiltonian and {R}iemann manifold {M}onte {C}arlo.
\newblock In {\em Proceedings of the 30th International Conference on Machine
  Learning}, pages 1462--1470, 2013.

\bibitem{welling2011bayesian}
M.~Welling and Y.~W. Teh.
\newblock Bayesian learning via stochastic gradient {L}angevin dynamics.
\newblock In {\em Proceedings of the 28th International Conference on Machine
  Learning}, pages 681--688, 2011.

\bibitem{Zhao2015CTMC}
T~Zhao, Z.~Wang, A.~Cumberworth, J.~Gsponer, N.~de~Freitas, and
  A.~Bouchard-C\^{o}t\'{e}.
\newblock Bayesian analysis of continuous time {M}arkov chains with application
  to phylogenetic modelling.
\newblock {\em Bayesian Analysis}, 11:1203--1237, 2016.

\end{thebibliography}

\appendix

\section{Proofs of Section \ref{sec:Basic-bouncy-particle-sampler}\label{Appendix:proofs}}

\subsection{Proof of Proposition \ref{Proposition:piinvariance}\label{app:ProofsInvariance}}

The BPS process is a specific piecewise-deterministic Markov process
so the expression of its generator follows from \citep[Theorem 26.14]{davis1993markov}.
Its adjoint is given in \citep{PetersDeWith2012} and a derivation
of this expression from first principles can be found in the Supplementary
Material. To establish the invariance with respect to $\rho$, we
first show that $\mathcal{\int L}h(z)\rho\left(z\right){\rm d}z=0$.
We have 
\begin{eqnarray}
\mathcal{\int\int L}h(z)\rho\left(z\right){\rm d}z & = & \int\int\left\langle \nabla_{x}h\left(x,v\right),v\right\rangle \rho\left(z\right){\rm d}z\label{eq:expansiongenerator}\\
 &  & +\int\int\lambda(x,v)[h\left(x,R\left(x\right)v\right)-h\left(x,v\right)]\rho\left(z\right){\rm d}z\label{eq:expansiongenerator1}\\
 &  & +\lambda^{\mathrm{ref}}\int\int\int\left(h(x,v')-h(x,v)\right)\psi\left({\rm d}v'\right)\rho\left(z\right){\rm d}z\label{eq:expansiongenerator2}
\end{eqnarray}
As $\rho\left(z\right)=\pi\left(x\right)\psi\left(v\right)$, the
term \eqref{eq:expansiongenerator2} is trivially equal to zero, while
a change-of-variables shows that 
\begin{equation}
\int\int\lambda\left(x,v\right)h(x,R\left(x\right)v)\rho\left(z\right){\rm d}z=\int\int\lambda\left(x,R\left(x\right)v\right)h(x,v)\rho\left(z\right){\rm d}z\label{eq:BPSinvariance1}
\end{equation}

as $R^{-1}\left(x\right)v=R\left(x\right)v$ and$\left\Vert R\left(x\right)v\right\Vert =\left\Vert v\right\Vert $
implies $\psi\left(R\left(x\right)v\right)=\psi\left(v\right)$. Additionally,
by integration by parts, we obtain as $h$ is bounded 
\begin{equation}
\int\int\left\langle \nabla_{x}h\left(x,v\right),v\right\rangle \rho\left(z\right){\rm d}z=\int\int\left\langle \nabla U\left(x\right),v\right\rangle h\left(x,v\right)\rho\left(z\right){\rm d}z.\label{eq:BPSinvariance2}
\end{equation}
Substituting \eqref{eq:BPSinvariance1} and \eqref{eq:BPSinvariance2}
into \eqref{eq:expansiongenerator}-\eqref{eq:expansiongenerator1}-\eqref{eq:expansiongenerator2},
we obtain 
\[
\int\mathcal{\int L}h(z)\rho\left(z\right){\rm d}z=\int\int[\left\langle \nabla U\left(x\right),v\right\rangle +\lambda\left(x,R\left(x\right)v\right)-\lambda\left(x,v\right)]h(x,v)\rho\left(z\right){\rm d}z.
\]
Now we have 
\begin{eqnarray*}
\int\mathcal{\int L}h(z)\rho\left(z\right){\rm d}z & = & \int\int[\left\langle \nabla U\left(x\right),v\right\rangle +\lambda\left(x,R\left(x\right)v\right)-\lambda\left(x,v\right)]h(x,v)\rho\left(z\right){\rm d}z.\\
 & = & \int\int[\left\langle \nabla U\left(x\right),v\right\rangle +\max\{0,\left\langle \nabla U(x),R\left(x\right)v\right\rangle \}-\max\{0,\left\langle \nabla U(x),v\right\rangle \}]h(x,v)\rho\left(z\right){\rm d}z\\
 & = & \int\int[\left\langle \nabla U\left(x\right),v\right\rangle +\max\{0,-\left\langle \nabla U(x),v\right\rangle \}-\max\{0,\left\langle \nabla U(x),v\right\rangle \}]h(x,v)\rho\left(z\right){\rm d}z\\
 & = & 0,
\end{eqnarray*}

where we have used $\left\langle \nabla U(x),R\left(x\right)v\right\rangle =-\left\langle \nabla U(x),v\right\rangle $
and $\max\{0,-f\}-\max\{0,f\}=-f$ for any $f$. The result now follows
by \citep[Proposition 34.7]{davis1993markov}.

\subsection{Proof of Theorem \ref{thm:uniqueInvariant} }

\global\long\def\iid{\overset{\text{i.i.d.}}{\sim}}

Informally, our proof of ergodicity relies on showing that refreshments
allow the process to explore the entire space of velocities and positions,
$\R^{d}\times\R^{d}.$ To do so, it will be useful to condition on
an event $\mathscr{E}$ on which paths are ``tractable.'' Since
two refreshment events are sufficient to reach any given destination
point, we would like to focus on such paths (only one refreshment
is not sufficient since a destination point specifies both a final
position and a final velocity). In particular, refreshments are simpler
to analyze than bouncing events, so we would like to condition on
an event on which no bouncing occurs in a time interval of interest.

To formalize this idea, we make use of the time-scale implementation
of the BPS algorithm (Section \ref{subsec:Simulation-time-scale-transform}).
This allows us to express $\mathscr{E}$ in terms of simple independent
events. We start by introducing some notation for the time-scale implementation
of the BPS algorithm. Let $i\ge1$ denote the index of the current
event being simulated, and assume without loss of generality that
$\lambda^{\text{ref}}=$1. Let $x_{0},v_{0}$ denote the initial position
and velocity, while the positions and velocities at the event times,
$z_{i}=(x_{i},v_{i})$, $i\in\left\{ 1,2,\dots\right\} $ are defined
as in Algorithm \ref{alg:BasicBouncyParticleSamplerRefreshment},
we use the notation $(x_{i},v_{i})$ instead of $(x^{\left(i\right)},v^{\left(i\right)})$
to slightly simplify notation. At event time $i$, we simulate three
independent random variables: two exponentially distributed, $e_{i}^{\text{(bounce)}},\tau_{i}^{\text{(ref)}}\sim\text{Exp}(1)$,
and one $d$-dimensional normal, $n_{i}\sim\mathcal{N}(0,I)$. The
candidate time to a refreshment is given by $\tau_{i}^{\text{(ref)}}$
and the candidate time to a bounce event is defined as $\tau_{i}^{\text{(bounce)}}=\Xi_{z_{i-1}}^{-1}(e_{i}^{\text{(bounce)}})$,
where $\Xi_{z}^{-1}$ is the quantile function of $\Xi_{z}(t)=\int_{0}^{t}\chi_{z}(s)\ud s$,
and $\chi_{x,v}(s)=\lambda(x+vs,v)$. The time to the next event is
$\tau_{i}=\min\left\{ \tau_{i}^{\text{(bounce)}},\tau_{i}^{\text{(ref)}}\right\} $.
The random variable $n_{i}$ is only used if $\tau_{i}=\tau_{i}^{\text{(ref)}},$
otherwise the bouncing operator is used to update the velocity in
a deterministic fashion.

We can now define our tractable set and establish its key properties. 
\begin{lem}
\label{lemma:tractable-set}Let $t>0$, and assume the initial point
of the BPS, satisfies $\|x_{0}\|\le t$, $\|v_{0}\|\le1$. If $\|\nabla U\|^{*}=\sup\left\{ \|\nabla U(x)\|:\|x\|\le3t\right\} $
then the event 
\begin{eqnarray*}
\mathscr{E} & =\underbrace{\left(\tau_{1}^{\text{(ref)}}+\tau_{2}^{\text{(ref)}}\le t<\tau_{1}^{\text{(ref)}}+\tau_{2}^{\text{(ref)}}+\tau_{3}^{\text{(ref)}}\right)}_{\mathscr{E}_{1}}\cap & \underbrace{\bigcap_{i\in\left\{ 1,2\right\} }\left(\|n_{i}\|\le1\right)}_{\mathscr{E}_{2}}\cap\underbrace{\bigcap_{i\in{1,2,3}}\left(e_{i}^{\text{(bounce)}}\ge t\|\nabla U\|^{*}\right)}_{\mathscr{E}_{3}}
\end{eqnarray*}

has the following properties: 
\begin{enumerate}
\item \label{part:bounded-norms}On the event $\mathscr{E}$, we have $\|v(t')\|\le1$
and $\|x(t')\|\le2t$ for all $t'\in\left[0,t\right]$, 
\item \label{part:no-bounces}On the event $\mathscr{E}$, there are exactly
two refreshments and no bouncing in the interval $(0,t)$, i.e. $\tau_{1}=\tau_{1}^{\text{(ref)}}$,
$\tau_{2}=\tau_{2}^{\text{(ref)}},$ and $\tau_{1}+\tau_{2}\le t\le\tau_{1}+\tau_{2}+\tau_{3}$, 
\item \label{part:positive-pr}$\P(\mathscr{E})$ is a strictly positive
constant that does not depend on $z(0)=(x_{0},v_{0})$%
, 
\item \label{part:velocity-cond-dist}$v_{i}|\mathscr{E}\iid\psi_{\le1}(0,I)$
for $i\in\left\{ 1,2\right\} $, where $\psi_{\le1}$ denotes the
truncated Gaussian distribution, with $\|v_{i}\|\le1$, 
\item \label{part:times-cond-dist}$\left(\tau_{1}^{\text{(ref)}},\tau_{2}^{\text{(ref)}}\right)|\mathscr{E}\sim\mathcal{U}\left(\left\{ (\tau_{1},\tau_{2})\in(0,t)^{2}:\tau_{1}+\tau_{2}\le t\right\} \right)$. 
\end{enumerate}
\end{lem}
\begin{proof}
To prove Part 1 and 2, we will make use of this preliminary result:
on $\mathscr{E}_{3},$ $\|v_{i-1}\|\le1,\|x_{i-1}\|\le2t$ implies
$\tau_{i}^{\text{(bounce)}}\ge t$ for $i\in\left\{ 1,2,3\right\} $.
Indeed, $\|v_{i-1}\|\le1$ and $\|x_{i-1}\|\le2t$ imply that $\chi_{z_{i-1}}(t')\le\|\nabla U\|^{*}$
for all $t'\in\left[0,t\right]$. It follows that $\Xi_{z_{i-1}}(t)\le\|\nabla U\|^{*}t$.
Hence, by the continuity of $\Xi_{z_{i-1}}$ and standard properties
of the quantile function, $\tau_{i}^{\text{(bounce)}}=\Xi_{z_{i-1}}^{-1}(e_{i}^{\text{(bounce)}})\ge t$.

Part 1 and 2: by the assumption on $x_{0}$ and $v_{0}$ and our preliminary
result, $\tau_{1}^{\text{(bounce)}}\ge t$, and hence, combining with
$\mathscr{E}_{1}$ and $\mathscr{E}_{2}$ we have $\tau_{1}=\tau_{1}^{\text{(ref)}}\le t$
and $\|v_{1}\|=\|n_{1}\|\le1$. Also, by the triangle inequality,
$\|x_{1}\|\le\|x_{0}\|+\|x_{1}-x_{0}\|\le t+\tau_{1}^{\text{(ref)}}\le2t$.
We can therefore apply our preliminary result again and obtain $\tau_{2}^{\text{(bounce)}}\ge t$,
and hence, combining again with $\mathscr{E}_{1}$ and $\mathscr{E}_{2}$,
we have $\tau_{2}=\tau_{2}^{\text{(ref)}}$, $\tau_{1}+\tau_{2}\le t$,
$\|v_{2}\|=\|n_{2}\|\le1$. Applying the triangle inequality a second
time yields $\|x_{2}\|\le t+\tau_{1}^{\text{(ref)}}+\tau_{2}^{\text{(ref)}}\le2t$.
We apply our preliminary result one last time to obtain $\tau_{3}^{\text{\text{(bounce)}}}\ge t$.
Hence, if $\tau_{3}^{\text{(ref)}}>\tau_{3}^{\text{(bounce)}}$, $\tau_{1}+\tau_{2}+\tau_{3}=\tau_{1}^{\text{(ref)}}+\tau_{2}^{\text{(ref)}}+\tau_{3}^{\text{(bounce)}}\ge t$,
while if $\tau_{3}^{\text{(ref)}}\le\tau_{3}^{\text{(bounce)}}$,
we can use $\mathscr{E}_{1}$ to conclude that $\tau_{1}+\tau_{2}+\tau_{3}=\tau_{1}^{\text{(ref)}}+\tau_{2}^{\text{(ref)}}+\tau_{3}^{\text{(ref)}}\ge t$.
It follows from the triangle inequality that $\|x(t')\|\le2t$ for
all $t'\in\left[0,t\right]$.

Part 3, 4 and 5: these follow straightforwardly from the construction
of $\mathscr{E}$. 
\end{proof}
Note that the statement and proof of Part 4 is simple because $\mathscr{E}_{3}\in\sigma(e_{i}^{\text{(bounce)}}:i\in\left\{ 1,2,3\right\} )$.
In contrast, conditioning on conceptually simpler events of the form
$(\tau_{i}^{\text{(bounce)}}>t)\in\sigma\left(e_{i}^{\text{(bounce)}},z_{i-1}\right)$
leads to conditional distributions on $v_{k}$ which are harder to
characterize.

In the following, $B_{R}\left(x\right)$ denotes the $d$-dimensional
Euclidean ball of radius $R$ centered at $x$. 
\begin{lem}
\label{lemma:bounded-norm}For all $\epsilon,t>0$ such that $\epsilon\le t/6$,
and $v,v_{2}\in B_{1}(0)$, $x,x'\in B_{\epsilon}(0)$, $0\leq\tau_{1}\leq\frac{t}{6}$,
and $\frac{2t}{3}\leq\tau_{2}\leq\frac{5t}{6}$, we have $\left\Vert v_{1}\right\Vert \le1$,
where $v_{1}$ is defined by:

\begin{equation}
v_{1}=\frac{\left(x'-(t-\tau_{1}-\tau_{2})v_{2}\right)-\left(x+\tau_{1}v\right)}{\tau_{2}}.\label{eq:V1expression}
\end{equation}
\end{lem}
\begin{proof}
By the triangle inequality:

\begin{eqnarray*}
\|v_{1}\| & \le & \frac{\|x'\|+(t-\tau_{1}-\tau_{2})\|v_{2}\|+\|x\|+\tau_{1}\|v\|}{\tau_{2}}\\
 & \le & \frac{\epsilon+(t-\tau_{1}-\tau_{2})+\epsilon+\tau_{1}}{\tau_{2}}\le1.
\end{eqnarray*}
\end{proof}
The next lemma formalizes the idea that refreshments allow the BPS
to explore the whole space. We use $\vol(A)$ to denote the Lebesgue
volume of a measurable set $A\in\R^{d}\times\R^{d}$ and $P_{t}\left(z,A\right)$
is the probability $z\left(t\right)\in A$ given $z\left(0\right)=z$
under the BPS dynamics. 
\begin{lem}
\label{lem:refreshkernel}
\begin{enumerate}
\item \label{part:uniform-bound}For all $\epsilon,t>0$ such that $\epsilon\le t/6$,
there exist $\delta>0$ such that for all $z=\left(x,v\right)\in B=B_{\epsilon}(0)\times B_{1}(0)$
and measurable set $A\subset\R^{d}\times\R^{d}$, 
\begin{equation}
P_{t}\left(z,A\right)\ge\delta\ \vol(A\cap B).\label{eq:ErgodicityPureRefreshmentLowerBound}
\end{equation}
\item \label{part:skeletons-irreducibility}For all $t_{0}>0$, $\z=(x,v)\in\R^{\d}\times\R^{\d}$
and open set $W\subset\R^{\d}\times\R^{\d}$ there exists an integer
$n\geq1$ such that $\trans{}_{nt_{0}}(z,W)>0$. 
\end{enumerate}
\end{lem}
\begin{proof}[Proof of Lemma \ref{lem:refreshkernel}]
Part 1: We have

\begin{eqnarray*}
P_{t}(z,A) & = & \E_{z}[\1_{A}(z(t))]\\
 & \ge & \P(\mathscr{E})\E_{z}[\1_{A}(z(t))|\mathscr{E}]\\
 & = & \P(\mathscr{E})\E_{z}[\1_{A}(x+\tau_{1}^{\text{(ref)}}v+\tau_{2}^{\text{(ref)}}v_{1}+(t-\tau_{1}^{\text{(ref)}}-\tau_{2}^{\text{(ref)}})v_{2}|\mathscr{E}],
\end{eqnarray*}

where we used Part \ref{part:no-bounces} of Lemma \ref{lemma:tractable-set},
which holds since $\|x\|\le t/6\le t$. By Lemma \ref{lemma:tractable-set},
Part \ref{part:positive-pr}, it is enough to show $\E_{z}[\1_{A}(z(t))|\mathscr{E}]\ge\delta'\vol(A\cap B)$
for some $\delta'>0$.

Using Lemma \ref{lemma:tractable-set}, Parts \ref{part:velocity-cond-dist}
and \ref{part:times-cond-dist}, we can rewrite the above conditional
expectation as: 
\[
\E_{z}[\1_{A}(z(t))|\mathscr{E}]=\iiiint\1_{A\cap B}\left(x+\tau_{1}v+\tau_{2}v_{1}+(t-\tau_{1}-\tau_{2})\ v_{2}\right)\psi_{\le1}(v_{1})\psi_{\le1}(v_{2})p(\tau_{1},\tau_{2})\ud v_{2}\ud v_{1}\ud\tau_{2}\ud\tau_{1}.
\]
We will use the coarea formula to reorganize the order of integration,
see e.g. Section 3.2 of \cite{Adler2011coareaREf}. We start by introducing
some notation.

For $C^{1}$ Riemannian manifolds $M$ and $N$ of dimension $m$
and $n$, a differentiable map $F:M\rightarrow N$ and $h$ a measurable
test function, the coarea formula can be written as: 
\begin{align*}
\int_{M}h(F(y)){\rm d}\mathcal{H}_{m}(y) & =\int_{N}{\rm d}\mathcal{H}_{n}(u)h(u)\int_{F^{-1}(u)}{\rm d}\mathcal{H}_{m-n}(x)\frac{1}{JF(x)}.
\end{align*}
Here $\mathcal{H}_{m}$, $\mathcal{H}_{n}$ and $\mathcal{H}_{m-n}$
denote the volume measures associated with the Riemannian metric on
$M$, $N$ and $F^{-1}(u)$ (with the induced metric of $M$). In
the above equations, $JF$ is a generalization of the determinant
of the Jacobian $JF:=\sqrt{\det g(\nabla f_{i},\nabla f_{j})}$ where
$g$ is the corresponding Riemannian metric and $f$ is the representation
of $F$ in local coordinates, see \cite{Adler2011coareaREf}. Here
$JF\text{=\ensuremath{\sqrt{\det DF\,DF^{\top}}}}$ where $DF$ is
defined in equation (\ref{eq:defDF}) below.

We apply the coarea formula to $M=\left\{ (\tau_{1},\tau_{2})\in(0,t)^{2}:\tau_{1}+\tau_{2}\leq t\right\} \times B_{1}(0)\times B_{1}(0)$,
$N=B_{t}(x)\times B_{1}(0)$, $F_{z}(\tau_{1},\v_{1},\tau_{2},\v_{2})=\left(x+\tau_{1}v+\tau_{2}v_{1}+(t-\tau_{1}-\tau_{2})v_{2},v_{2}\right),$
$m=2d+2$, $n=2d$ and obtain:

\begin{eqnarray}
\E_{z}[\1_{A}(z(t))|\mathscr{E}] & = & \int_{N}\ud z'\1_{A\cap B}(z')\underbrace{\int_{F_{z}^{-1}(z')}\frac{{\rm d}\mathcal{H}_{2}(\tau_{1},\tau_{2},\v_{1},\v_{2})\psi_{\le1}(v_{1})\psi_{\le1}(v_{2})p(\tau_{1},\tau_{2})}{JF_{z}(\tau_{1},\tau_{2},\v_{1},\v_{2})}}_{I_{z}(z')},\label{eq:after-coarea}
\end{eqnarray}
where $p(\tau_{1},\tau_{2})$ denotes the joint conditional density
of $\tau_{1},\tau_{2}|\mathscr{E}$ described in Part \ref{part:times-cond-dist}
of Lemma \ref{lemma:tractable-set}, and:

\begin{align}
DF_{z} & =\left(\begin{array}{cccc}
v-v_{2} & \tau_{2}I & v_{1}-v_{2} & \left(t-\tau_{1}-\tau_{2}\right)I\\
0 & 0 & 0 & I
\end{array}\right),\label{eq:defDF}\\
DF_{z}\ DF_{z}^{\top} & =\left(\begin{array}{cc}
(v-v_{2})(v-v_{2})^{\top}+\left(v_{1}-v_{2}\right)\left(v_{1}-v_{2}\right)^{\top}+\left(\tau_{2}^{2}+\left(t-\tau_{1}-\tau_{2}\right)^{2}\right)I & \left(t-\tau_{1}-\tau_{2}\right)I\\
\left(t-\tau_{1}-\tau_{2}\right)I & I
\end{array}\right),\nonumber \\
\nonumber \\
JF_{z} & =\sqrt{\det DF_{z}\ DF_{z}^{\top}}.\nonumber 
\end{align}

We define $\delta'=\inf\left\{ I_{z}(z'):z,z'\in B\right\} $, and
obtain the following inequality

\begin{eqnarray*}
\int_{N}\ud z'\1_{A\cap B}(z')I_{z}(z') & \ge & \delta'\int_{N}\ud z'\1_{A\cap B}(z')\\
 & = & \delta'\ \vol(A\cap B).
\end{eqnarray*}

It is therefore enough to show that $\delta'>0$. To do so, we will
derive the following bounds related to the integral in $I_{z}(z')$: 
\begin{enumerate}
\item \label{part:bounded-domain}Its domain of integration $F_{z}^{-1}(z')$
is guaranteed to contain a set of positive $\mathcal{H}_{2}$ measure. 
\item \label{part:bounded-integrand}Its integrand is bounded below by a
strictly positive constant. 
\end{enumerate}
To establish \ref{part:bounded-domain}, we let $z'=(x',v')=(x',v_{2})$,
and notice that rearranging 
\begin{align*}
x'=x+\tau_{1}v+\tau_{2}v_{1}+(t-\tau_{1}-\tau_{2})v_{2}\\
\end{align*}
yields an expression for $v_{1}$ given in (\ref{eq:V1expression}).
From Lemma \ref{lemma:bounded-norm}, it follows that

\[
C_{z}(z'):=\left\{ \left(\tau_{1},\frac{\left(x'-(t-\tau_{1}-\tau_{2})v'\right)-\left(x+\tau_{1}v\right)}{\tau_{2}},\tau_{2},v'\right):0\leq\tau_{1}\leq\frac{t}{6},\frac{2t}{3}\leq\tau_{2}\leq\frac{5t}{6}\right\} \subseteq F_{z}^{-1}(z'),
\]

and $\mathcal{H}_{2}(C_{z}(z'))\ge(t/6)^{2}$ since the surface of
the graph of a function is larger than the surface of the domain.

To establish \ref{part:bounded-integrand}, we start by analyzing
$JF_{z}$. Exploiting its block structure, we obtain: 
\begin{align*}
JF_{z}^{2} & =\det DF_{z}\ DF_{z}^{\top}\\
 & =\det((v-v_{2})(v-v_{2})^{\top}+\left(v_{1}-v_{2}\right)\left(v_{1}-v_{2}\right)^{\top}+\tau_{2}^{2}I)\\
 & \leq\left(\frac{\left\Vert v-v_{2}\right\Vert ^{2}+\left\Vert v_{1}-v_{2}\right\Vert ^{2}}{d}+\tau_{2}^{2}\right)^{d}\leq\left(8+t^{2}\right)^{d}
\end{align*}

Moreover, it follows from basic properties of the truncated Gaussian
distribution $\psi_{\le1}$ and of $p(\tau_{1},\tau_{2})$ that
\begin{eqnarray*}
\inf\left\{ \psi_{\le1}(v_{1})\psi_{\le1}(v_{2})p(\tau_{1},\tau_{2}):(\tau_{1},v_{1},\tau_{2},v_{2})\in M\right\}  & = & K>0.
\end{eqnarray*}

Combining (\ref{eq:after-coarea}) to \ref{part:bounded-domain} and
\ref{part:bounded-integrand}, we obtain,

\begin{eqnarray*}
I_{z}(z') & \ge & \frac{K}{(8+t^{2})^{d}}\int_{F_{z}^{-1}(z')}{\rm d}\mathcal{H}_{2}(\tau_{1},\tau_{2},v_{1},v_{2})\\
 & \ge & \frac{K}{(8+t^{2})^{d}}\mathcal{H}_{2}(C_{z}(z'))\\
 & \ge & \frac{K}{(8+t^{2})^{d}}\left(\frac{t}{6}\right)^{2},
\end{eqnarray*}
and hence, $\delta'>0$.

To prove Part 2 of the lemma, we divide the trajectory of length $nt_{0}$
into three ``phases'' namely a deceleration, travel, and acceleration
phases, or respective lengths $t_{\text{d}}+t_{\text{t}}+t_{\text{a}}=nt_{0}$
defined below (see also Supplement for a figure illustrating the notation
used in this part of the lemma). This allows us to use Part 1 of the
present lemma which requires velocities bounded in norm by one. We
require an acceleration phase since $W$ may not necessarily include
velocities of norms bounded by one.

First, we show that we decelerate with positive probability by time
$t_{\text{\text{d}}}=t_{0}$. Let $\mathscr{R}$ denote the event
that there is exactly one refreshment in the interval $(0,t_{0})$,
and that the refreshed velocity has norm bounded by one. Define also
$r_{0}=t_{0}\max\left\{ 1,\|v\|\right\} $, which bounds the distance
travelled in $(0,t_{0})$ for outcomes in $\mathscr{R}$, since bouncing
does not change the norm of the velocity. We have: 
\[
P_{t_{\text{0}}}(z,B_{r_{0}}(x)\times B_{1}(0))\ge\P_{z}(\mathscr{R})=(t_{0}\lambda^{\text{ref}})\exp(-t_{0}\lambda^{\text{ref}})\psi(B_{1}(0))=:K'>0.
\]
Next, to prepare applying the first part of the lemma, set 
\begin{align*}
\epsilon & =1+\max\left\{ \|x\|+r_{0},r'\right\} ,\\
r' & =\inf\left\{ \|x'\|:\text{\ensuremath{\exists} }v'\in\R^{d}\text{ with }(x',v')\in W\right\} .
\end{align*}
Informally, $\epsilon$ is selected so that the ball of radius $\epsilon$
around the origin contains both any position attained after deceleration,
as well as ball around a point $x^{\star}$ in $W$. Indeed, since
$W$ is open and that $\epsilon>r'$, there exists some $r>0,\left(x^{\star},v^{\star}\right)\in W$
such that $B_{r}(x^{\star})\times B_{r}(v^{\star})\subseteq W$ and
$B_{\epsilon}(0)\supseteq B_{r_{0}}(x)\cup B_{r}(x^{\star})$. Let
also $n=2+\left\lceil \frac{6\epsilon}{t_{0}}\right\rceil $ .

Of the total time $nt_{0}$, we reserve time $t_{\text{a}}=\min\left\{ t_{0},(2(\|v^{\star}\|/r+1))^{-1},r/2\right\} $
to accelerate. This time is selected so that (a) $t_{0}-t_{\text{a}}\ge0$,
and (b), if we start with a position in $B_{r/2}(x^{\star})$, move
with a velocity bounded in norm by $\bar{v}=\max\left\{ 1,\|v^{\star}\|+r\right\} $
for a time $\Delta t\le\min\left\{ (2(\|v^{\star}\|/r+1))^{-1},r/2\right\} $,
we have that the final position is in $B_{r}(x^{\star})$. This holds
since the distance travelled is bounded by $\bar{v}\Delta t\le r/2$.
Hence, by a similar argument as used for deceleration, we have, for
all $z''\in B_{r/2}(x^{\star})\times B_{1}(0)$, 
\[
P_{t_{\text{a}}}(z'',B_{r}(x^{\star})\times B_{r}(v^{\star}))\ge(t_{\text{a}}\lambda^{\text{ref}})\exp(-t_{\text{a}}\lambda^{\text{ref}})\psi(B_{r}(v^{\star}))=:K''>0.
\]

With these definitions, we can apply the first part of the present
lemma with $t_{\text{t}}=(n-2)t_{0}+(t_{0}-t_{\text{a}})\ge6\epsilon$
and obtain a constant $\delta>0$ such that $P_{t_{\text{t}}}(z',A)\ge\delta\vol(A\cap(B_{\epsilon}(0)\times B_{1}(0)))$
for all $z'\in B_{\epsilon}(0)\times B_{1}(0)$ and measurable set
$A\subseteq\R^{2d}$. We thus obtain:

\begin{eqnarray*}
P_{nt_{0}}(z,W) & \ge & P_{nt_{0}}(z,B_{r}(x^{\star})\times B_{r}(v^{\star}))\\
 & = & \int\int P_{t_{\text{d}}}(z,\ud z')P_{t_{\text{t}}}(z',\ud z'')P_{t_{\text{a}}}(z'',B_{r}(x^{\star})\times B_{r}(v^{\star}))\\
 & \ge & \int_{z'\in B_{r_{0}}(x)\times B_{1}(0)}\int_{z''\in B_{r/2}(x^{\star})\times B_{1}(0)}P_{t_{\text{d}}}(z,\ud z')P_{t_{\text{t}}}(z',\ud z'')P_{t_{\text{a}}}(z'',B_{r}(x^{\star})\times B_{r}(v^{\star}))\\
 & \ge & K'K''\delta\vol(B_{r/2}(x^{\star})\times B_{1}(0))>0.
\end{eqnarray*}
\end{proof}
We can now exploit this Lemma to prove Theorem \ref{thm:uniqueInvariant}. 
\begin{proof}[Proof of Theorem \ref{thm:uniqueInvariant}]
Suppose BPS is not ergodic, then it follows from standard results
in ergodic theory that there are two measures $\mu_{1}$ and $\mu_{2}$
such that $\mu_{1}\perp\mu_{2}$ and $\mu_{i}P_{t_{0}}=\mu_{i}$;
see e.g. \citep[Theorem 1.7]{Hairer2010ErgodicLectureNotes}. Thus
there is a measurable set $A\subset\R^{d}\times\R^{d}$ such that
\begin{equation}
\mu_{1}(A)=\mu_{2}(A^{c})=0.\label{eq:ergodicitySingularSet}
\end{equation}
Let $A_{1}=A,$ $A_{2}=A^{c}$, and $B=B_{1}(0)\times B_{1}(0)$.
Because of Lemma \ref{lem:refreshkernel} Part \ref{part:skeletons-irreducibility}
and Lemma 2.2 of \cite{Hairer2010ErgodicLectureNotes} the support
of the $\mu_{i}$ is $\R^{\d}\times\R^{\d}$. Thus, $\mu_{i}(B)>0$
for $i\in\left\{ 1,2\right\} $. At least one of $A_{1}\cap B$ or
$A_{2}\cap B$ has a positive Lebesgue volume, hence, we can denote
by $i^{\star}\in\left\{ 1,2\right\} $ an index satisfying $\vol(A_{i^{\star}}\cap B)>0$.
Now, we pick $t=6$ and obtain from Lemma \ref{lem:refreshkernel}
Part \ref{part:uniform-bound} that there is some $\delta>0$ such
that $P_{t}(z,A_{i^{\star}})\ge\delta\ \vol(A_{i^{\star}}\cap B)$
for all $z\in B$. By invariance we have 
\begin{eqnarray*}
\mu_{i^{\star}}(A_{i^{\star}}) & = & \int\mu_{i^{\star}}\left(\ud z\right)P_{t}(z,A_{i^{\star}})\\
 & \ge & \int_{B}\mu_{i^{\star}}\left(\ud z\right)P_{t}(z,A_{i^{\star}})\\
 & \ge & \int_{B}\mu_{i^{\star}}\left(\ud z\right)\delta\ \vol(A_{i^{\star}}\cap B)\\
 & = & \mu_{i^{\star}}(B)\delta\ \vol(A_{i^{\star}}\cap B)>0.
\end{eqnarray*}

This contradicts that $\mu_{i}(A_{i})=0$ for $i\in\left\{ 1,2\right\} $.

The law of large numbers then follows by Birkhoff's pointwise ergodic
theorem; see e.g. \citep[Theorem 2.30, Section 2.6.4]{birkhoff}. 
\end{proof}

\section{Proof of Proposition \ref{Proposition:infinitenumberinvariantmeasures.}}

The dynamics of the BPS can be lumped into a two-dimensional Markov
process involving only the radius $r\left(t\right)=\left\Vert x\left(t\right)\right\Vert $
and $m\left(t\right)=\left\langle x\left(t\right),v\left(t\right)\right\rangle /\left\Vert x\left(t\right)\right\Vert $
for any dimensionality $d\ge2$. The variable $m\left(t\right)$ can
be interpreted (via $\arccos(m\left(t\right))$) as the angle between
the particle position $x\left(t\right)$ and velocity $v\left(t\right)$.
Because of the strong Markov property we can take $\tau_{1}=0$ without
loss of generality and let $t$ be some time between the current event
and the next, yielding: 
\begin{eqnarray}
r\left(t\right) & = & \sqrt{\left\langle x\left(0\right)+v\left(0\right)\cdot t,x\left(0\right)+v\left(0\right)\cdot t\right\rangle }=\sqrt{r_{\text{0}}^{2}+2m\left(0\right)r\left(0\right)t+t^{2}}\label{eq:IsoNormalRadialDeterministic}\\
m\left(t\right) & = & \frac{\left\langle x\left(0\right)+v\left(0\right)\cdot t,v\left(0\right)\right\rangle }{\left\Vert x\left(0\right)+v\left(0\right)\cdot t\right\Vert }=\frac{m\left(0\right)r\left(0\right)+t}{r\left(t\right)}\nonumber 
\end{eqnarray}
If there is a bounce at time $t$, then $r\left(t\right)$ is not
modified but $m\left(t\right)=-m\left(t\right)$.The bounce happens
with intensity $\Int(x+tv,v)=\max\left(0,\left\langle x+vt,v\right\rangle \right)$.
These processes can also be written as an Stochastic Differential
Equation (SDE) driven by a jump process whose intensity is coupled
to its position. This is achieved by writing the deterministic dynamics
given in \eqref{eq:IsoNormalRadialDeterministic} between events as
the following Ordinary Differential Equation (ODE): 
\begin{align*}
\frac{{\rm d}}{{\rm d}t}r\left(t\right) & =\frac{2m\left(t\right)r\left(t\right)}{2r\left(t\right)}=m\left(t\right)\\
\frac{{\rm d}}{{\rm d}t}m\left(t\right) & =\frac{r\left(t\right)-\left(r\left(t\right)m\left(t\right)\right)m\left(t\right)}{\left(r\left(t\right)\right)^{2}}=\frac{1-\left(m\left(t\right)\right)^{2}}{r\left(t\right)}.
\end{align*}
Taking the bounces into account turns this ODE into an SDE with 
\begin{align}
{\rm d}r\left(t\right) & =m\left(t\right){\rm d}t\nonumber \\
{\rm d}m\left(t\right) & =\frac{1-\left(m\left(t\right)\right)^{2}}{r\left(t\right)}{\rm d}t-2m\left(t\right)dN_{t}.\label{eq:normalSDE}
\end{align}
where $N_{t}$ is the counting process associated with a PP with intensity
$\max\left(0,r\left(t\right)m\left(t\right)\right)$.

Now consider the push forward measure of $\mathcal{N}\left(0,\frac{1}{2}I_{k}\right)\otimes\mathcal{U}(\mathcal{S}^{k-1})$
under the map $\left(x,v\right)\mapsto\left(\left\Vert x\right\Vert ,\left\langle x,v\right\rangle /\left\Vert x\right\Vert \right)$
where $\mathcal{U}(\mathcal{S}^{k-1})$ is the uniform distribution
on $\mathcal{S}^{k-1}$. This yields the collection of measures with
densities $f_{k}(r,m)$. One can check that $f_{k}(r,m)$ is invariant
for \eqref{eq:normalSDE} for all $k\ge2$.

\section{Bayesian logistic regression for large datasets\label{Appendix:Logistic}}

\subsection{Bounds on the intensity\label{subsec:Bound-on-intensity}}

We derive here a datapoint-specific upper bound $\bar{\chi}^{[r]}$
to $\chi^{[r]}(t)$. First, we need to compute the gradient for one
datapoint:\global\long\def\logistic{\textrm{logistic}}
 
\[
\nabla U^{[r]}(x)=\iota_{r}(\logistic\dotprod{\iota_{r}}{x}-y_{r}),
\]
where: 
\[
\logistic(a)=\frac{e^{a}}{1+e^{a}}.
\]

We then consider two sub-cases depending on $y_{r}=0$ or $y_{r}=1$.
Suppose first $y_{r}=0$, and let $x(t)=x+tv$ 
\begin{align*}
\chi^{[r]}(t) & =\max\{0,\dotprod{\nabla U_{r}(x(t))}{v}\}\\
 & =\max\left\{ 0,\sum_{k=1}^{d}\iota_{r,k}v_{k}\logistic\dotprod{\iota_{r}}{x(t)}\right\} \\
 & \le\sum_{k=1}^{d}\1[v_{k}>0]\iota_{r,k}v_{k}\logistic\dotprod{\iota_{r}}{x(t)}\\
 & \le\underbrace{\sum_{k=1}^{d}\1[v_{k}>0]\iota_{r,k}v_{k}}_{\bar{\chi}^{[r]}}
\end{align*}

Similarly, we have for $y_{r}=1$ 
\begin{align*}
\chi^{[r]}(t) & =\max\left\{ 0,\sum_{k=1}^{d}\iota_{r,k}v_{k}(\logistic\dotprod{\iota_{r}}{x(t)}-1)\right\} \\
 & =\max\left\{ 0,\sum_{k=1}^{d}\iota_{r,k}(-v_{k})(1-\logistic\dotprod{\iota_{r}}{x(t)})\right\} \\
 & \le\sum_{k=1}^{d}\1[v_{k}<0]\iota_{r,k}(-v_{k})(1-\logistic\dotprod{\iota_{r}}{x(t)})\\
 & =\sum_{k=1}^{d}\1[v_{k}<0]\iota_{r,k}|v_{k}|(1-\logistic\dotprod{\iota_{r}}{x(t)})\\
 & \le\underbrace{\sum_{k=1}^{d}\1[v_{k}<0]\iota_{r,k}|v_{k}|}_{\bar{\chi}^{[r]}}.
\end{align*}

Combining these terms we obtain
\[
\bar{\chi}^{[r]}=\sum_{k=1}^{d}\1[v_{k}(-1)^{y_{r}}\ge0]\iota_{r,k}|v_{k}|.
\]
When implementing Algorithm \ref{alg:local-bps-thinning}, we need
to bound $\sum_{r=1}^{R}\chi^{[r]}(t)$. We have 
\begin{align*}
\bar{\chi} & =\sum_{r=1}^{R}\bar{\chi}^{[r]}\geq\sum_{r=1}^{R}\chi^{[r]}(t)\\
 & =\sum_{r=1}^{R}\sum_{k=1}^{d}\1[v_{k}(-1)^{y_{r}}\ge0]\iota_{r,k}\ |v_{k}|\\
 & =\sum_{k=1}^{d}|v_{k}|\sum_{r=1}^{R}\1[v_{k}(-1)^{y_{r}}\ge0]\iota_{r,k}\\
 & =\sum_{k=1}^{d}|v_{k}|\ \iota_{k}^{(\1[v_{k}<0])},
\end{align*}

where
\[
\iota_{k}^{(c)}:=\sum_{r=1}^{R}\1[(-1)^{c+y_{r}}\ge0]\iota_{r,k}.
\]
The bound $\bar{\chi}$ is constant between bounce events and only
depends on the magnitude of $v$. If we further assume that we use
restricted refreshment then this bound is valid for any $t>0$ allowing
us to implement Algorithm \ref{alg:local-bps-thinning} using \eqref{eq:remupper-1}
or \eqref{eq:remupper-2}.

\subsection{Sampling the thinned factor\label{subsec:Sampling-the-thinned}}

We show here how to implement Step \ref{enu:Sample-which-factor}
of Algorithm \ref{alg:local-bps-thinning} without enumerating over
the $R$ datapoints. We begin by introducing some required pre-computed
data structures. The pre-computation is executed only once at the
beginning of the algorithm, so its running time, $O(R\log R)$ is
considered negligible (the number of bouncing events is assumed to
be greater than $R$). For each dimensionality $k$ and class label
$c$, consider the categorical distribution with the following probability
mass function over the datapoints: 
\[
\mu_{k}^{(c)}(r)=\frac{\iota_{r,k}\1[y_{r}=c]}{\iota_{k}^{(c)}}.
\]
This is just the distribution over the datapoints that have the given
label, weighted by the covariate $k$. An alias sampling data-structure
\citep[Section 3.4]{Devroye1986} is computed for each $k$ and $c$.
This pre-computation takes total time $O(R\log R)$. This allows subsequently
to sample in time $O(1)$ from the distributions $\mu_{k}^{(c)}$.

We now show how this pre-computation is used to to implement Step
\ref{enu:Sample-which-factor} of Algorithm \ref{alg:local-bps-thinning}.
We denote the probability mass function we want to sample from by
\[
q(r)=\frac{\bar{\chi}^{[r]}}{\bar{\chi}}.
\]
To sample this distribution efficiently, we construct an artificial
joint distribution over both datapoints and covariate dimension indices
\[
q(r,k)=\frac{\1[v_{k}(-1)^{y_{r}}\ge0]\iota_{r,k}\ |v_{k}|}{\bar{\chi}}.
\]
\global\long\def\tauk{q_{\textrm{k}}}
 \global\long\def\taurk{q_{\textrm{r}|\textrm{k}}}
 We denote by $\tauk(k)$, respectively $\taurk(r|k)$, the associated
marginal, respectively conditional distribution. By construction,
we have 
\[
\sum_{k=1}^{d}q(r,k)=q(r).
\]
It is therefore enough to sample $(r,k)$ and to return $r$. To do
so, we first sample (a) $k\sim\tauk(\cdot)$ and then (b) sample $r|k\sim\taurk(\cdot|k)$.

For (a), we have 
\[
\tauk(k)=\frac{|v_{k}|\ \iota_{k}^{(\1[v_{k}<0])}}{\bar{\chi}},
\]
so this sampling step again does not require looping over the datapoints
thanks the pre-computations described earlier.

For (b), we have 
\[
\taurk(r|k)=\mu_{k}^{(\1[v_{k}<0])}(r),
\]
and therefore this sampling step can be computed in $O(1)$ thanks
to the pre-computed alias sampling data structure.

\subsection{Algorithm description\label{appendix:logisticalg}}

Algorithm \ref{alg:local-bps-logistic} contains a detailed implementation
of the local BPS with thinning for the logistic regression example
(Example \ref{subsec:Logisticregression}).

\begin{algorithm}[H]
\protect\caption{Local BPS algorithm for Logistic Regression with Large Datasets ~\label{alg:local-bps-logistic}}

\begin{enumerate}
\item Precompute the alias tables $\mu_{k}^{(c)}(r)$ for $k=1,\dots,d$,
$c\in\{0,1\}$ in order to sample from $\taurk(\cdot|k)$. 
\item Initialize $\left(x^{\left(0\right)},v^{\left(0\right)}\right)$ arbitrarily
on $\mathbb{R}^{d}\times\mathbb{R}^{d}$. 
\item Initialize the global clock $T\leftarrow0$. 
\item Initialize $\bar{T}\leftarrow\triangle$. (time until which local
upper bounds are valid) 
\item Compute local-in-time upper bound on the prior factor as $\bar{\chi}_{\text{prior }}=\sigma^{-2}\max\left(0,\left\langle x^{(0)}+v^{(0)}\Delta,v^{(0)}\right\rangle \right)$
(notice the rate associated with the prior is monotonically increasing). 
\item While more events $i=1,2,\ldots$ requested do
\begin{enumerate}
\item Compute the local-in-time upper bound on the data factors in $\mathcal{O}(d)$
\[
\bar{\chi}=\sum_{k=1}^{d}|v_{k}^{\left(i-1\right)}|\ \iota_{k}^{(\1[v_{k}^{\left(i-1\right)}<0])}.
\]
\item Sample $\tau\sim\mathrm{Exp\left(\bar{\chi}_{\text{prior }}+\bar{\chi}+\lambda^{\mathrm{ref}}\right)}$. 
\item If $\left(T+\tau>\bar{T}\right)$ then
\begin{enumerate}
\item $x^{(i)}\gets x^{(i-1)}+v^{(i-1)}(\bar{T}-T)$. 
\item $v^{(i)}\gets v^{(i-1)}$. 
\item Compute the local-in-time upper bound $\bar{\chi}_{\text{prior }}=\sigma^{-2}\max\left(0,\left\langle x^{(i)}+v^{(i)}\Delta,v^{(i)}\right\rangle \right).$ 
\item Set $T\leftarrow\bar{T}$, $\bar{T}$$\leftarrow\bar{T}+\triangle.$ 
\end{enumerate}
\item Else
\begin{enumerate}
\item $x^{(i)}\gets x^{(i-1)}+v^{(i-1)}\tau$. 
\item Sample $j$ from $\text{Discrete}(\frac{\bar{\chi}}{\bar{\chi}_{\text{prior }}+\bar{\chi}+\lambda^{\mathrm{ref}}},\frac{\lambda^{\mathrm{ref}}}{\bar{\chi}_{\text{prior }}+\bar{\chi}+\lambda^{\mathrm{ref}}},\frac{\bar{\chi}_{\text{prior }}}{\bar{\chi}_{\text{prior }}+\bar{\chi}+\lambda^{\mathrm{ref}}})$
. 
\item If $j=1$
\begin{enumerate}
\item Sample $k$ according to $\tauk(k)=|v_{k}^{\left(i-1\right)}|\ \iota_{k}^{(\1[v_{k}^{\left(i-1\right)}<0])}/\bar{\chi}$. 
\item Sample $r$ $\sim\taurk(\cdot|k)$ using the precomputed alias table. 
\item If $V<\frac{\max\left(0,\left\langle \nabla U^{[r]}(x^{(i)}),v^{\left(i-1\right)}\right\rangle \right)}{\bar{\chi}^{[r]}}$
where $V\sim\mathcal{U}\left(0,1\right)$ .\\
 $v^{(i)}\gets R_{r}(x^{(i)})v^{(i-1)}$ where $R_{r}$ is the bouncing
operator associated with the $r$-th data item .\\
 
\item Else $v^{(i)}\gets v^{(i-1)}$. 
\end{enumerate}
\item If $j=2$
\begin{enumerate}
\item $v^{(i)}\sim\mathcal{N}(0_{d},I_{d})$. 
\end{enumerate}
\item If $j=3$
\begin{enumerate}
\item If $V<\frac{\sigma^{-2}\max\left(0,\left\langle x^{(i)},v^{(i-1)}\right\rangle \right)}{\chi_{\text{prior}}}$
where $V\sim\mathcal{U}\left(0,1\right)$.\\
 $v^{(i)}\gets R_{\text{prior}}(x^{(i)})v^{(i-1)}$ where $R_{\text{prior}}$
is the bouncing operator associated with the prior. \\
 
\item Else $v^{(i)}\gets v^{(i-1)}$. 
\end{enumerate}
\item Compute the local-in-time upper bound 
\[
\bar{\chi}_{\text{prior }}=\sigma^{-2}\max\left(0,\left\langle x^{(i)}+v^{(i)}\left(\bar{T}-T-\tau\right),v^{(i)}\right\rangle \right).
\]
\item Set $T\leftarrow T+\tau$. 
\end{enumerate}
\end{enumerate}
\end{enumerate}
\end{algorithm}

\pagebreak{}
\begin{center}
\textbf{\large{}Supplemental Material: The Bouncy Particle Sampler
}\\
\textbf{\large{} A Non-Reversible Rejection-Free Markov Chain Monte
Carlo Method}{\large{} }
\par\end{center}{\large \par}

\setcounter{equation}{0} \setcounter{figure}{0} \setcounter{table}{0}
\setcounter{page}{1} \makeatletter \global\long\def\theequation{S\arabic{equation}}
 \global\long\def\thefigure{S\arabic{figure}}
 \global\long\def\bibnumfmt#1{[S#1]}
 \global\long\def\citenumfont#1{S#1}
\section{Illustration for Lemma 3}
The following figure illustrates the different phases considered in the proof of Lemma 3.

\begin{figure}
\begin{centering}
\includegraphics[width=0.7\textwidth]{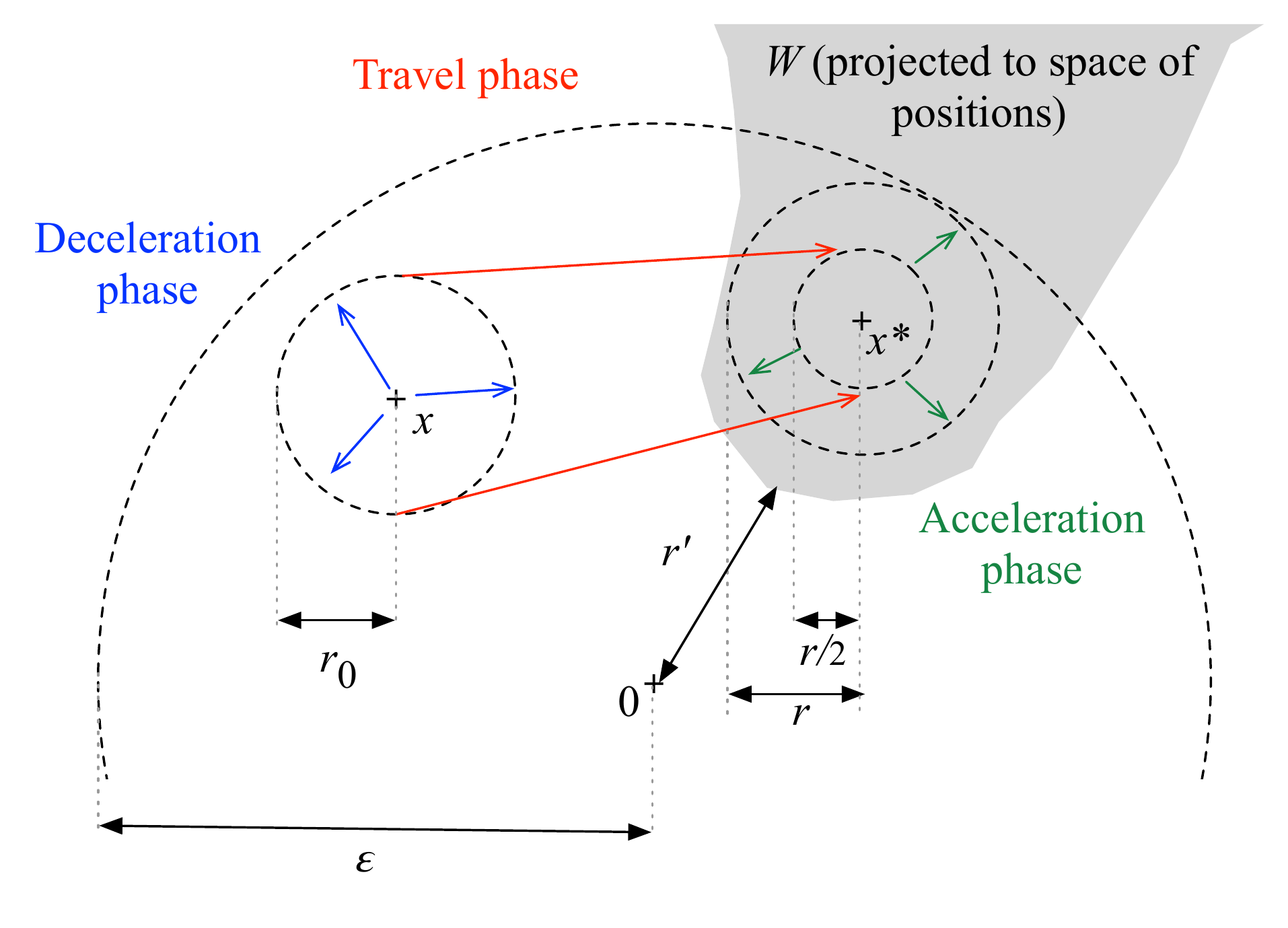}

\par\end{centering}
\caption{\label{fig:lemma3} TIllustration of the notation used in the proof of Lemma 3, part (2) (this figure is best viewed in colour). We consider trajectories divided into three phases, namely deceleration (blue), travel (red), and acceleration (green). The ball around the initial point, $x$, is used when bringing the velocity down to a norm bounded by one (deceleration phase). This allows us to use Lemma 3, part (1) in order to bound the probability of travel to a ball centered at $x^\star$ of radius $r/2$ (travel phase). Finally, since $W$ may not contain points with velocity bounded by one, the ball centered at $x^\star$ of radius $r$ is used to reach velocities contained in $W$ while ensuring position coordinates are also still in $W$.  }
\end{figure}

\section{Direct proof of invariance\label{sec:Invariance}}

Let $\mu_{t}$ be the law of $z\left(t\right)$. In the following,
we prove invariance by explicitly verifying that the time evolution
of the density $\frac{{\rm d}\mu_{t}}{{\rm d}t}=0$ is zero if the
initial distribution $\mu_{0}$ is given by $\rho(z)=\pi\left(x\right)\psi\left(v\right)$
in Proposition 1. This is achieved by deriving the forward Kolmogorov
equation describing the evolution of the marginal density of the stochastic
process. For simplicity, we start by presenting the invariance argument
when $\lambda^{\text{{ref}}}=0$.

\textbf{Notation and description of the algorithm}. We denote a pair
of position and velocity by $z=\left(x,v\right)\in\R^{\d}\times\R^{\d}$
and we denote translations by $\Phi_{t}(z)=(\Phi_{t}^{\mathrm{pos}}(z),\Phi_{t}^{\mathrm{dir}}(z))=\left(x+vt,v\right)$.
The time of the first bounce coincides with the first arrival $T_{1}$
of a PP with intensity $\chi(t)=\mathrm{\lambda}(\Phi_{t}(z))$ where:
\begin{equation}
\mathrm{\lambda}(z)=\max\left\{ 0,\left\langle \nabla U\left(x\right),v\right\rangle \right\} .\label{eq:intensity}
\end{equation}
It follows that the probability of having no bounce in the interval
$[0,t]$ is given by: 
\begin{equation}
\mathrm{N}\mathrm{o}_{t}(z)=\exp\left(-\int_{0}^{t}\lambda(\Phi_{s}(z)){\rm d}s\right),\label{eq:poisson-no-event}
\end{equation}
and the density of the random variable $T_{1}$ is given by: 
\begin{eqnarray}
q(t_{1};z)= & \mathbf{1}[t_{1}>0]\frac{{\rm d}}{{\rm d}t_{1}}\left(1-\mathrm{N}o_{t_{1}}(z)\right)\\
= & \mathbf{1}[t_{1}>0]\mathrm{N}\mathrm{o}_{t_{1}}(z)\mathrm{\lambda}(\Phi_{t_{1}}(z)).\label{eq:density}
\end{eqnarray}

If a bounce occurs, then the algorithm follows a translation path
for time $T_{1}$, at which point the velocity is updated using a
bounce operation $C(z)$, defined as: 
\begin{equation}
C\left(z\right)=\left(x,R\left(x\right)v\right)\label{eq:collisionoperator}
\end{equation}
where 
\begin{equation}
R\left(x\right)v=v-2\frac{\left\langle \nabla U\left(x\right),v\right\rangle \nabla U\left(x\right)}{\left\Vert \nabla U\left(x\right)\right\Vert ^{2}}.\label{eq:updatedirection}
\end{equation}
The algorithm then continues recursively for time $t-T_{1}$, in the
following sense: a second bounce time $T_{2}$ is simulated by adding
to $T_{1}$ a random increment with density $q(\cdot;C\circ\Phi_{t_{1}}(z))$.
If $T_{2}>t$, then the output of the algorithm is $\Phi_{t-t_{1}}\circ C\circ\Phi_{t_{1}}(z)$,
otherwise an additional bounce is simulated, etc. More generally,
given an initial point $z$ and a sequence $\mathbf{t}=(t_{1},t_{2},\dots)$
of bounce times, the output of the algorithm at time $t$ is given
by: 
\begin{equation}
\Psi_{\mathbf{t},t}\left(z\right)=\begin{cases}
\Phi_{t}(z) & \mathrm{if}~t_{1}>0~\mathrm{or~}\mathbf{t}=(\thinspace),\\
\Psi_{\mathbf{t'},t-t_{1}}\left(z\right)\circ C\circ\Phi_{t_{1}}(z) & \mathrm{otherwise,}
\end{cases}\label{eq:recursion}
\end{equation}
where $(\thinspace)$ denotes the empty list and $\mathbf{t}'$ the
suffix of $\mathbf{t}$: $\mathbf{t}'=(t_{2},t{}_{3},\dots)$. As
for the bounce times, they are distributed as follows: 
\begin{eqnarray}
T_{1} & \sim & q(\;\cdot\;;z)\\
T_{i}-T_{i-1}|T_{1:i-1} & \sim & q\Big(\;\cdot\;;\underbrace{\Psi_{T_{1:i-1},T_{i-1}}(z)}_{\text{Pos. after collision \ensuremath{i-1}}}\Big),\;\;\;i\in\{2,3,4,\dots\}
\end{eqnarray}
where $T_{1:i-1}=\left(T_{1},T_{2},\ldots,T_{i-1}\right).$

\textbf{Decomposition by the number of bounces}. Let $h$ denote an
arbitrary non-negative measurable test function. We show how to decompose
expectations of the form $\mathbb{E}[h(\Psi_{\mathbf{T},t}(z))]$
by the number of bounces in the interval $(0,t)$. To do so, we introduce
a function $\#\mathrm{Col}_{t}(\mathbf{t})$, which returns the number
of bounces in the interval $(0,t)$: 
\begin{equation}
\#\mathrm{Col}_{t}(\mathbf{t})=\min\left\{ n\ge1:t_{n}>t\right\} -1.
\end{equation}
From this, we get the following decomposition: 
\begin{eqnarray}
\mathbb{E}[h(\Psi_{\mathbf{T},t}(z))] & = & \mathbb{E}[h(\Psi_{\mathbf{T},t}(z))\sum_{n=0}^{\infty}\mathbf{1}[\#\mathrm{Col}_{t}(\mathbf{T})=n]]\\
 & = & \sum_{n=0}^{\infty}\mathbb{E}[h(\Psi_{\mathbf{T},t}(z))\mathbf{1}[\#\mathrm{Col}_{t}(\mathbf{T})=n]].\label{eq:decomposition}
\end{eqnarray}
On the event that no bounce occurs in the interval $[0,t)$, i.e.
$\#\mathrm{Col}_{t}(\mathbf{T})=0$, the function $\Psi_{\mathbf{T},t}(z)$
is equal to $\Phi_{t}(z)$, therefore: 
\begin{eqnarray}
\mathbb{E}[h(\Psi_{\mathbf{T},t}(z))\mathbf{1}[\#\mathrm{Col}_{t}(\mathbf{T})=0]] & = & h(\Phi_{t}(z))\mathbb{P}(\#\mathrm{Col}_{t}(\mathbf{T})=0)\\
 & = & h(\Phi_{t}(z))\mathrm{N}\mathrm{o}_{t}(z).\label{eq:first-case}
\end{eqnarray}

\smallskip{}
 Indeed, on the event that $n\ge1$ bounces occur, the random variable
$h(\Phi_{t}(z))$ only depends on a finite dimensional random vector,
$(T_{1},T_{2},\dots,T_{n})$, so we can write the expectation as an
integral with respect to the density $\widetilde{q}(t_{1:n};t,z)$
of these variables: {\footnotesize{}{}{} 
\begin{eqnarray}
 &  & \mathbb{E}[h(\Psi_{\mathbf{T},t}(z))\mathbf{1}[\#\mathrm{Col}_{t}(\mathbf{T})=n]]\\
 &  & =\mathbb{E}\left[h(\Psi_{\mathbf{T},t}(z))\mathbf{1}[0<T_{1}<\dots<T_{n}<t<T_{n+1}]\right]\nonumber \\
 &  & =\idotsint{}_{0<t_{1}<\dots<t_{n}<t<t_{n+1}}h(\Psi_{t_{1:n},t}(z))q(t_{1};z)\prod_{i=2}^{n+1}q(t-t_{i-1};\Psi_{t_{1:i-1},t_{i-1}}(z)){\rm d}t_{1:n+1}\nonumber \\
 &  & =\idotsint{}_{0<t_{1}<\dots<t_{n}<t}h(\Psi_{t_{1:n},t}(z))\widetilde{q}(t_{1:n};t,z){\rm d}t_{1:n},\label{eq:second-case}
\end{eqnarray}
} where: 
\[
\widetilde{q}(t_{1:n};t,z)=q(t_{1};z)\times\begin{cases}
\mathrm{N}\mathrm{o}_{t-t_{1}}(\Phi_{t_{1}}(z)) & \mathrm{if}~n=1\\
\mathrm{N}\mathrm{o}_{t-t_{n}}(\Phi_{t_{1:n},t_{n}}(z)) & \prod_{i=2}^{n}q\left(t_{i}-t_{i-1};\Psi_{t_{1:i-1},t_{i-1}}(z)\right)\ \mathrm{if}~n\geq2.
\end{cases}
\]

\smallskip{}
 To include Equations~(\ref{eq:first-case}) and (\ref{eq:second-case})
under the same notation, we define $t_{1:0}$ to the empty list, $(\thinspace)$,
$\widetilde{q}((\thinspace);t,z)=\mathrm{N}\mathrm{o}_{t}(z)$, and
abuse the integral notation so that for all $n\in\{0,1,2,\dots\}$:
\begin{equation}
\mathbb{E}[h(\Psi_{\mathbf{T},t}(z))\mathbf{1}[\#\mathrm{Col}_{t}(\mathbf{T})=n]]=\idotsint{}_{0<t_{1}<\dots<t_{n}<t}\thinspace h(\Psi_{t_{1:n},t}(z))\widetilde{q}(t_{1:n};t,z){\rm d}t_{1:n}.\label{eq:integral-1}
\end{equation}
\textbf{Marginal density. }Let us fix some arbitrary time $t>0$.
We seek a convenient expression for the marginal density at time $t$,
$\mu_{t}(z)$, given an initial vector $Z\sim\rho$, where $\rho$
is the hypothesized stationary density $\rho(z)=\pi\left(x\right)\psi\left(v\right)$
on $Z$.%
{} To do so, we look at the expectation of an arbitrary non-negative
measurable test function $h$: {\footnotesize{}{}{} 
\begin{eqnarray}
\mathbb{E}[h(\Psi_{\mathbf{T},t}(Z))] & = & \mathbb{E}\Big[\mathbb{E}[h(\Psi_{\mathbf{T},t}(Z))|Z]\Big]\label{eq:integral}\\
 & = & \sum_{n=0}^{\infty}\mathbb{E}\Big[\mathbb{E}[h(\Psi_{\mathbf{T},t}(Z))\mathbf{1}[\#\mathrm{Col}_{t}(\mathbf{T})=n]|Z]\Big]\\
 & = & \sum_{n=0}^{\infty}\int_{\mathcal{Z}}\rho(z)\idotsint{}_{0<t_{1}<\dots<t_{n}<t}h(\Psi_{t_{1:n},t}(z))\widetilde{q}(t_{1:n};t,z){\rm d}t_{1:n}{\rm d}z\\
 & = & \sum_{n=0}^{\infty}\;\;\idotsint{}_{0<t_{1}<\dots<t_{n}<t}\int_{\mathcal{Z}}\rho(z)h(\Psi_{t_{1:n},t}(z))\widetilde{q}(t_{1:n};t,z){\rm d}z{\rm d}t_{1:n}\\
 & = & \sum_{n=0}^{\infty}\;\;\idotsint{}_{0<t_{1}<\dots<t_{n}<t}\int_{\mathcal{Z}}\rho(\Psi_{t_{1:n},t}^{-1}(z'))h(z')\widetilde{q}(t_{1:n};t,\Psi_{t_{1:n},t}^{-1}(z'))\left|\det D\Psi_{t_{1:n},t}^{-1}\right|{\rm d}z'{\rm d}t_{1:n}\nonumber \\
 & = & \int_{\mathcal{Z}}h(z')\underbrace{\sum_{n=0}^{\infty}\;\;\idotsint{}_{0<t_{1}<\dots<t_{n}<t}\rho(\Psi_{t_{1:n},t}^{-1}(z'))\widetilde{q}(t_{1:n};t,\Psi_{t_{1:n},t}^{-1}(z')){\rm d}t_{1:n}}_{\mu_{t}(z')}{\rm d}z'.\label{eq:rhot}
\end{eqnarray}
}\smallskip{}
 We used the following in the above derivation successively the law
of total expectation, equation~(\ref{eq:decomposition}), equation~(\ref{eq:integral}),
Tonelli's theorem and the change of variables, $z'=\Psi_{t_{1:n},t}(z)$,
justified since for any fixed $0<t_{1}<t_{2}<\dots<t_{n}<t<t_{n+1}$,
$\Psi_{t_{1:n},t}(\cdot)$ is a bijection (being a composition of
bijections). Now the absolute value of the determinant is one since
$\Psi_{\mathbf{t},t}\left(z\right)$ is a composition of unit-Jacobian
mappings and, by using Tonelli's theorem again, we obtain that the
expression above the brace is necessarily equal to $\mu_{t}(z')$
since $h$ is arbitrary.

\textbf{Derivative}. Our goal is to show that for all $z'\in\mathcal{Z}$
\[
\frac{{\rm d}\mu_{t}(z')}{{\rm d}t}=0.
\]
Since the process is time homogeneous, once we have computed the derivative,
it is enough to show that it is equal to zero at $t=0$. To do so,
we decompose the computation according to the terms $I_{n}$ in Equation~(\ref{eq:rhot}):
\begin{eqnarray}
\mu_{t}(z') & = & \sum_{n=0}^{\infty}I_{n}(z',t)\label{eq:sum}\\
I_{n}(z',t) & = & \idotsint{}_{0<t_{1}<\dots<t_{n}<t}\rho(\Psi_{t_{1:n},t}^{-1}(z'))\widetilde{q}(t_{1:n};t,\Psi_{t_{1:n},t}^{-1}(z')){\rm d}t_{1:n}.
\end{eqnarray}
The categories of terms in Equation~(\ref{eq:sum}) to consider are:

\emph{No bounce}: $n=0$, $\Psi_{t_{1:n},t}(z)=\Phi_{t}(z)$, or,

\emph{Exactly one bounce}: $n=1$, $\Psi_{t_{1:n},t}(z)=F_{t,t_{1}}:=\Phi_{t-t_{1}}\circ C\circ\Phi_{t_{1}}(z)$
for some $t_{1}\in(0,t)$, or,

\emph{Two or more bounces}: $n\ge2$, $\Psi_{t_{1:n},t}(z)=\Psi_{t-t_{2}}\circ C\circ F_{t_{2},t_{1}}(z)$
for some $0<t_{1}<t_{2}<t$

In the following, we show that the derivative of the terms in the
third category, $n\ge2$, are all equal to zero, while the derivative
of the first two categories cancel each other.

\textbf{No bounce in the interval}. From Equation~(\ref{eq:first-case}):
\begin{eqnarray}
I_{0}(z',t) & = & \rho(\Phi_{-t}(z'))\mathrm{N}\mathrm{o}_{t}(\Phi_{-t}(z')).\label{eq:I0}
\end{eqnarray}
We now compute the derivative at zero of the above expression: 
\begin{eqnarray}
\left.\frac{{\rm d}}{{\rm d}t}I_{0}(z',t)\right|_{t=0} & =\mathrm{N}\mathrm{o}_{0} & (\Phi_{0}(z'))\left.\frac{{\rm d}\rho(\Phi_{-t}(z'))}{{\rm d}t}\right|_{t=0}+\nonumber \\
 &  & \rho(\Phi_{0}(z'))\left.\frac{{\rm d}\mathrm{N}\mathrm{o}_{t}(\Phi_{-t}(z'))}{{\rm d}t}\right|_{t=0}\label{eq:derivative}
\end{eqnarray}
The first term in the above equation can be simplified as follows:
\begin{eqnarray}
\mathrm{N}\mathrm{o}_{0}(\Phi_{0}(z'))\frac{{\rm d}\rho(\Phi_{-t}(z'))}{{\rm d}t} & = & \frac{{\rm d}\rho(\Phi_{-t}(z'))}{{\rm d}t}\\
 & = & \left\langle \frac{\partial\rho(\Phi_{-t}(z'))}{\partial\Phi_{-t}^{\mathrm{pos}}(z')},\frac{{\rm d}\Phi_{-t}^{\mathrm{pos}}(z')}{{\rm d}t}\right\rangle +\nonumber \\
 &  & \left\langle \frac{\partial\rho(\Phi_{-t}(z'))}{\partial\Phi_{-t}^{\mathrm{dir}}(z')},\underbrace{\frac{{\rm d}\Phi_{-t}^{\mathrm{dir}}(z')}{{\rm d}t}}_{=\boldsymbol{0}}\right\rangle \\
 & = & \left\langle \frac{\partial\rho(z)}{\partial x},-v'\right\rangle \\
 & = & \left\langle \frac{\partial}{\partial x}\frac{1}{Z}\exp\left(-U(x)\right)\psi\left(v\right),-v'\right\rangle \nonumber \\
 & = & \rho(\Phi_{-t}(z'))\left\langle \nabla U(x),v'\right\rangle ,
\end{eqnarray}
where $x=\Phi_{-t}^{\mathrm{pos}}(z')$. The second term in Equation~(\ref{eq:derivative})
is equal to: 
\begin{eqnarray}
\rho(\Phi_{0}(z'))\left.\frac{{\rm d}\mathrm{N}\mathrm{o}_{t}(\Phi_{-t}(z'))}{{\rm d}t}\right|_{t=0} & = & -\rho(\Phi_{0}(z'))\mathrm{N}\mathrm{o}_{0}(z')\lambda(\Phi_{0}(z'))\\
 & = & -\rho(z')\lambda(z'),
\end{eqnarray}
using Equation~(\ref{eq:density}). In summary, we have: 
\[
\left.\frac{{\rm d}}{{\rm d}t}I_{0}(z',t)\right|_{t=0}=\rho(z')\left\langle \nabla U(x'),v'\right\rangle -\rho(z')\mathrm{\lambda}(z').
\]

\textbf{Exactly one bounce in the interval}. From Equation~(\ref{eq:second-case}),
the trajectory consists in a bounce at a time $T_{1}$, occurring
with density (expressed as before as a function of the final point
$z'$) $q(t_{1};F_{t,t_{1}}^{-1}(z'))$, followed by no bounce in
the interval $(T_{1},t]$, an event of probability: 
\begin{eqnarray}
\mathrm{N}\mathrm{o}_{t-t_{1}}(C\circ\Phi_{t_{1}}(z)) & = & \mathrm{N}\mathrm{o}_{t-t_{1}}(C\circ\Phi_{t_{1}}\circ F_{t,t_{1}}^{-1}(z'))\\
 & = & \mathrm{N}\mathrm{o}_{t-t_{1}}(\Phi_{t_{1}-t}(z')),
\end{eqnarray}
where we used that $C^{-1}=C$. This yields: 
\begin{eqnarray}
I_{1}(z',t) & = & \int_{0}^{t}q(t_{1};F_{t,t_{1}}^{-1}(z'))\rho(\Psi_{t_{1:1},t}^{-1}(z'))\mathrm{N}\mathrm{o}_{t-t_{1}}(\Phi_{t_{1}-t}(z')){\rm d}t_{1}.\label{eq:I1}
\end{eqnarray}

To compute the derivative of the above equation at zero, we use again
Leibniz's rule: 
\begin{eqnarray*}
\left.\frac{{\rm d}}{{\rm d}t}I_{1}(z',t)\right|_{t=0}=\rho(C(z'))\mathrm{\lambda}(C(z')).
\end{eqnarray*}

\textbf{Two or more bounces in the interval. }For a number of bounce,
we get: {\footnotesize{}{}{} 
\begin{equation}
I_{n}(z',t)=\int_{0}^{t}\Bigg[\;\;\underbrace{\idotsint{}_{t_{2:n}:t_{1}<t_{2}\dots<t_{n}<t}\rho(\Psi_{t_{1:n},t}^{-1}(z'))\widetilde{q}(t_{1:n};t,\Psi_{t_{1:n},t}^{-1}(z')){\rm d}t_{2:n}}_{\tilde{I}(t_{1},t,z')}\Bigg]{\rm d}t_{1},
\end{equation}
} and hence, using Leibniz's rule on the integral over $t_{1}$: 
\begin{equation}
\left.\frac{{\rm d}}{{\rm d}t}I_{n}(z',t)\right|_{t=0}=\tilde{I}(0,0,z')=0.
\end{equation}

\textbf{Putting all terms together}. Putting everything together,
we obtain: 
\begin{eqnarray}
\left.\frac{{\rm d}\mu_{t}(z')}{{\rm d}t}\right|_{t=0}=\rho(z')\left\langle \nabla U(x'),v'\right\rangle \underbrace{-\rho(z')\lambda(z')+\rho(C(z'))\lambda(C(z')).}
\end{eqnarray}
From the expression of $\mathrm{\lambda}(\cdot)$, we can rewrite
the two terms above the brace as follows: 
\begin{align*}
 & -\rho(z')\mathrm{\lambda}(z')+\rho(C(z'))\lambda(C(z'))\\
= & -\rho(z')\lambda(z')+\rho(z')\lambda(C(z'))\\
= & -\rho(z')\max\{0,\left\langle \nabla U(x'),v'\right\rangle \}+\rho(z')\max\{0,\left\langle \nabla U(x'),R\left(x'\right)v'\right\rangle \}\\
= & -\rho(z')\max\{0,\left\langle \nabla U(x'),v'\right\rangle \}+\rho(z')\max\{0,\left\langle \nabla U(x'),R\left(x'\right)v'\right\rangle \}\\
= & -\rho(z')\max\{0,\left\langle \nabla U(x'),v'\right\rangle \}+\rho(z')\max\{0,-\left\langle \nabla U(x'),v'\right\rangle \}\\
= & -\rho(z')\left\langle \nabla U(x'),v'\right\rangle ,
\end{align*}
where we used that $\rho(z')=\rho(C(z'))$, $\left\langle \nabla U(x'),R\left(x'\right)v'\right\rangle =-\left\langle \nabla U(x'),v'\right\rangle $
and $-\max\{0,f\}+\max\{0,-f\}=-f$ for any function $f$. Hence we
have $\left.\frac{{\rm d}\mu_{t}(z')}{{\rm d}t}\right|_{t=0}=0$,
establishing that that the bouncy particle sampler $\lambda^{\text{ref}}=0$
admits $\rho$ as invariant distribution. The invariance for $\lambda^{\text{ref}}>0$
then follows from Lemma \ref{lem:mixMKernel} given below. 
\begin{lem}
\label{lem:mixMKernel} Suppose $P_{t}$ is a continuous time Markov
kernel and $Q$ is a discrete time Markov kernel which are both invariant
with respect to $\mu.$ Suppose we construct for $\lambda^{\text{{ref}}}>0$
a Markov process $\hat{P}_{t}$ as follows: at the jump times of an
independent PP with intensity $\lambda^{\text{{ref}}}$ we make a
transition with $Q$ and then continue according to $P_{t}$, then
$\hat{P}_{t}$ is also $\mu$-invariant. 
\end{lem}
\begin{proof}
The transition kernel is given by 
\begin{eqnarray*}
\hat{P}_{t} & = & e^{-\lambda t}P_{t}+\int_{0}^{t}{\rm d}t_{1}\lambda e^{\lambda t_{1}}e^{-\lambda(t-t_{1})}P_{t-t_{1}}QP_{t_{1}}\\
 &  & +\int_{0}^{t}{\rm d}t_{1}\int_{t_{1}}^{t_{2}}{\rm d}t_{2}\lambda^{2}e^{\lambda t_{1}}e^{\lambda(t_{2}-t_{1})}e^{-\lambda(t-t_{2})}P_{t-t_{2}}QP_{t_{2}-t_{1}}QP_{t_{1}}+\dots
\end{eqnarray*}
Therefore 
\begin{eqnarray*}
\mu\hat{P}_{t} & = & \mu\left(e^{-\lambda t}+\lambda te^{-\lambda t}+\frac{\left(\lambda t\right)^{2}}{2}e^{-\lambda t}\dots\right)\\
 & = & \mu.
\end{eqnarray*}
Hence $\hat{P}_{t}$ is $\mu$-invariant. 
\end{proof}

\section{Invariance of the local sampler\label{subsec:Local-SamplerInvariance}}

The generator of the local BPS is given by 
\begin{eqnarray}
\mathcal{L}h(z) & = & \left\langle \nabla_{x}h\left(x,v\right),v\right\rangle \label{eq:localBPSgenerator}\\
 &  & +\sum_{f\in F}\lambda_{f}\left(x,v\right)\left\{ h(x,R_{f}\left(x\right)v)-h(x,v)\right\} \nonumber \\
 &  & +\lambda^{\mathrm{ref}}\int\left(h(x,v')-h(x,v)\right)\psi\left(\ud v'\right).\nonumber 
\end{eqnarray}
The proof of invariance of the local BPS is very similar to the proof
of Propostion 1. We have 
\begin{eqnarray}
\mathcal{\int L}h(z)\rho\left(z\right){\rm d}z & = & \int\int\left\langle \nabla_{x}h\left(x,v\right),v\right\rangle \rho\left(z\right){\rm d}z\label{eq:local1}\\
 &  & +\int\int\sum_{f\in F}\lambda_{f}\left(x,v\right)\left\{ h(x,R_{f}\left(x\right)v)-h(x,v)\right\} ]\rho\left(z\right){\rm d}z\label{eq:local2}\\
 &  & +\lambda^{\mathrm{ref}}\int\int\int\left(h(x,v')-h(x,v)\right)\psi\left(\ud v'\right)\rho\left(z\right){\rm d}z\label{eq:local3}
\end{eqnarray}
where the term \eqref{eq:local3} is straightforwardly equal to 0
while, by integration by parts, the term \eqref{eq:local1} satisfies
\begin{equation}
\int\int\left\langle \nabla_{x}h\left(x,v\right),v\right\rangle \rho\left(z\right){\rm d}z=\int\int\left\langle \nabla U\left(x\right),v\right\rangle h\left(x,v\right)\rho\left(z\right){\rm d}z.\label{eq:local1change}
\end{equation}
as $h$ is bounded. Now a change-of-variables shows that for any $f\in F$
\begin{equation}
\int\int\lambda_{f}\left(x,v\right)h(x,R_{f}\left(x\right)v)\rho\left(z\right){\rm d}z=\int\int\lambda\left(x,R_{f}\left(x\right)v\right)h(x,v)\rho\left(z\right){\rm d}z\label{eq:local2changeterm}
\end{equation}

as $R_{f}^{-1}\left(x\right)v)=R\left(x\right)v$ and $\left\Vert R_{f}\left(x\right)v\right\Vert =\left\Vert v\right\Vert $
implies $\psi\left(R_{f}\left(x\right)v\right)=\psi\left(v\right)$.
So the term \eqref{eq:local2} satisfies 
\begin{eqnarray}
 &  & \int\int\sum_{f\in F}\lambda_{f}\left(x,v\right)\left\{ h(x,R_{f}\left(x\right)v)-h(z)\right\} ]\rho\left(z\right){\rm d}z\nonumber \\
 & = & \int\int\sum_{f\in F}[\lambda\left(x,R_{f}\left(x\right)v\right)-\lambda\left(x,v\right)]h(x,v)\rho\left(z\right){\rm d}z\nonumber \\
 & = & \int\int\sum_{f\in F}[\max\{0,\left\langle \nabla U_{f}(x),R\left(x\right)v\right\rangle \}-\max\{0,\left\langle \nabla U_{f}(x),v\right\rangle \}]h(x,v)\rho\left(z\right){\rm d}z\nonumber \\
 & = & \int\int\sum_{f\in F}[\max\{0,-\left\langle \nabla U_{f}(x),v\right\rangle \}-\max\{0,\left\langle \nabla U_{f}(x),v\right\rangle \}]h(x,v)\rho\left(z\right){\rm d}z\nonumber \\
 & = & -\int\int\sum_{f\in F}[\left\langle \nabla U_{f}\left(x\right),v\right\rangle ]h(x,v)\rho\left(z\right){\rm d}z\nonumber \\
 & = & -\int\int\left\langle \nabla U\left(x\right),v\right\rangle ]h(x,v)\rho\left(z\right){\rm d}z,\label{eq:local2change}
\end{eqnarray}

where we have used $\left\langle \nabla U_{f}(x),R_{f}\left(x\right)v\right\rangle =-\left\langle \nabla U_{f}(x),v\right\rangle $
and $\max\{0,-f\}-\max\{0,f\}=-f$ for any $f$. Hence, summing \eqref{eq:local1change}-\eqref{eq:local2change}-\eqref{eq:local3},
we obtain$\mathcal{\int L}h(z)\rho\left(z\right){\rm d}z=0$ and the
result follows by \citep[Proposition 34.7]{davis1993markov}.

\section{Calculations in the isotropic normal case}

As we do not use refreshment, it follows from the definition of the
collision operator that 
\begin{eqnarray*}
\left\langle x^{(i)},v^{(i)}\right\rangle  & = & \left\langle x^{(i)},v^{(i-1)}-\frac{2\left\langle x^{(i)},v^{(i-1)}\right\rangle }{\left\Vert x^{(i)}\right\Vert ^{2}}x^{(i)}\right\rangle \\
 & = & -\left\langle x^{(i)},v^{(i-1)}\right\rangle =-\left\langle x^{(i-1)},v^{(i-1)}\right\rangle -\tau_{i}\\
 & = & \begin{cases}
-\sqrt{-\log V_{i}} & \text{if }\left\langle x^{(i-1)},v^{(i-1)}\right\rangle \leq0\\
-\sqrt{\left\langle x^{(i-1)},v^{(i-1)}\right\rangle ^{2}-\log V_{i}} & \text{otherwise}
\end{cases},
\end{eqnarray*}
and therefore

\[
\left\Vert x^{(i)}\right\Vert ^{2}=\begin{cases}
\left\Vert x^{(i-1)}\right\Vert ^{2}-\left\langle x^{(i-1)},v^{(i-1)}\right\rangle ^{2}-\log V_{i} & \text{ if }\left\langle x^{(i-1)},v^{(i-1)}\right\rangle \leq0\\
\left\Vert x^{(i-1)}\right\Vert ^{2}-\log V_{i} & \text{otherwise.}
\end{cases}.
\]

It follows that $\left\langle x^{(j)},v^{(j)}\right\rangle \le0$
for $j>0$ if $\left\langle x^{(0)},v^{(0)}\right\rangle \le0$ so,
in this case, we have

\begin{align*}
\left\Vert x^{(i)}\right\Vert ^{2} & =\left\Vert x^{(i-1)}\right\Vert ^{2}-\left\langle x^{(i-1)},v^{(i-1)}\right\rangle ^{2}-\log V_{i}\\
 & =\left\Vert x^{(i-1)}\right\Vert ^{2}+\log V_{i-1}-\log V_{i}\\
 & =\left\Vert x^{(i-2)}\right\Vert ^{2}-\left\langle x^{(i-1)},v^{(i-1)}\right\rangle ^{2}-\log V_{i-1}+\log V_{i-1}-\log V_{i}\\
\vdots & \vdots\\
 & =\left\Vert x^{(1)}\right\Vert ^{2}-\left\langle x^{(1)},v^{(1)}\right\rangle ^{2}-\log V_{i}
\end{align*}

In particular for $x^{(0)}=e_{1}$ and $v^{(0)}=e_{2}$ with $e_{i}$
being elements of standard basis of $\mathbb{R}^{d}$, the norm of
the position at all points along the trajectory can never be smaller
than 1.

\section{Supplementary information on the evolutionary parameters inference
experiments\label{Appendix:Evolutionaryparameters}}

\subsection{Model}

We consider an over-parameterized generalized time reversible rate
matrix \cite{Tavare1986GTR} with $d=10$ corresponding to 4 unnormalized
stationary parameters $x_{1},\dots,x_{4}$, and 6 unconstrained substitution
parameters $x_{\{i,j\}}$, which are indexed by sets of size 2, i.e.
where $i,j\in\left\{ 1,2,3,4\right\} ,\thinspace i\neq j$. Off-diagonal
entries of $Q$ are obtained via $q_{i,j}=\pi_{j}\exp\left(x_{\{i,j\}}\right)$,
where 
\[
\pi_{j}=\frac{\exp\left(x_{j}\right)}{\sum_{k=1}^{4}\exp\left(x_{k}\right)}.
\]

We assign independent standard Gaussian priors on the parameters $x_{i}.$
We assume that a matrix of aligned nucleotides is provided, where
rows are species and columns contains nucleotides believed to come
from a shared ancestral nucleotide. Given $x=\left(x_{1},\dots,x_{4},x_{\{1,2\}},\dots,x_{\{3,4\}}\right),$
and hence $Q$, the likelihood is a product of conditionally independent
continuous time Markov chains over $\{$A, C, G, T$\}$, with ``time''
replaced by a branching process specified by the phylogenetic tree's
topology and branch lengths. The parameter $x$ is unidentifiable,
and while this can be addressed by bounded or curved parameterizations,
the over-parameterization provides an interesting challenge for sampling
methods, which need to cope with the strong induced correlations.

\subsection{Baseline}

We compare the BPS against a state-of-the-art HMC sampler \cite{Wang2013AHMC}
that uses Bayesian optimization to adapt the the leap-frog stepsize
$\epsilon$ and trajectory length $L$ of HMC. This sampler was shown
in \cite{Zhao2015CTMC} to be comparable or better to other state-of-the-art
HMC methods such as NUTS. It also has the advantage of having efficient
implementations in several languages. We use the author's Java implementation
to compare to our Java implementation of the BPS. Both methods view
the objective function as a black box (concretely, a Java interface
supporting pointwise evaluation and gradient calculation). In all
experiments, we initialize at the mode and use a burn-in of 100 iterations
and no thinning. The HMC auto-tuner yielded $\epsilon=0.39$ and $L=100$.
For our method, we use the global sampler and the global refreshment
scheme.

\subsection{Additional experimental results}

\begin{figure}
\centering\includegraphics[width=2.3in]{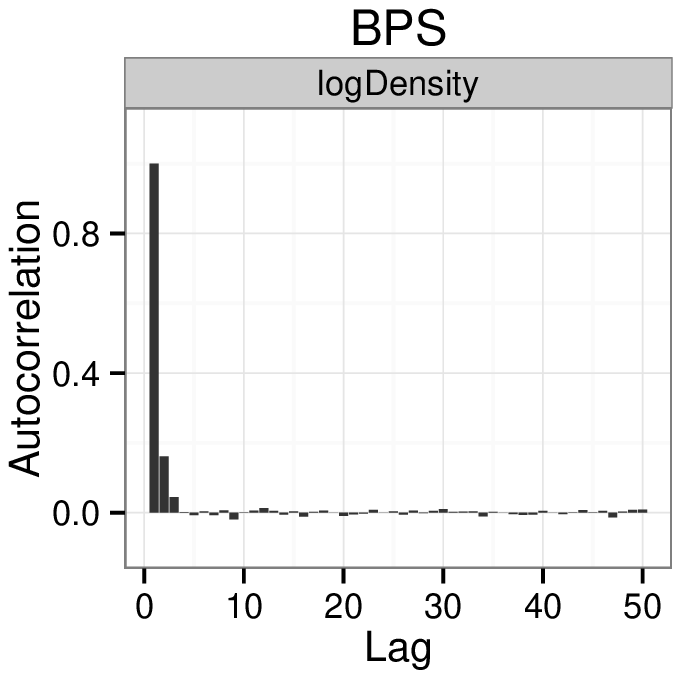}\includegraphics[width=2.3in]{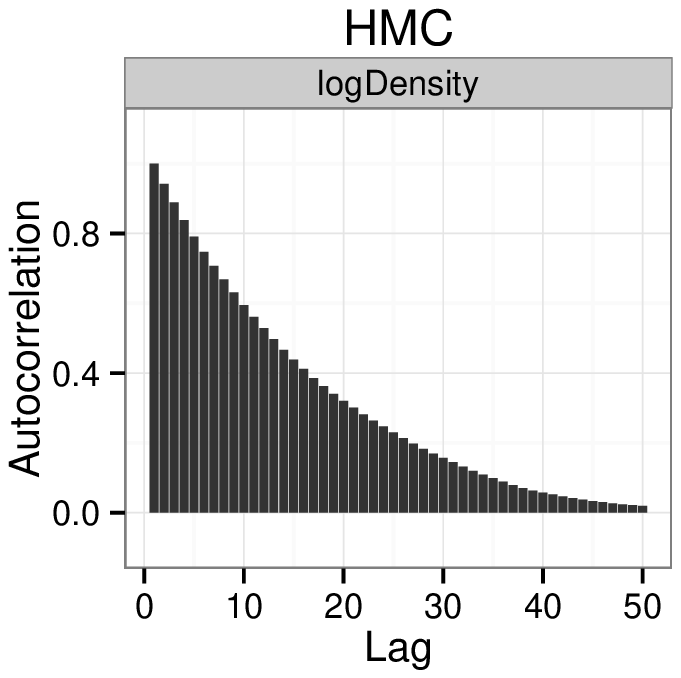}

\caption{\label{fig:acf}Estimate of the ACF of the log-likelihood statistic
for BPS (left) and HMC (right). A similar behavior is observed for
the ACF of the other statistics.}
\end{figure}

To ensure that BPS outperforming HMC does not come from a faulty auto-tuning
of HMC parameters, we look at the ESS/s for the log-likelihood statistic
when varying the stepsize $\epsilon$. The results in Figure \ref{fig:sensitivity}(right)
show that the value selected by the auto-tuner is indeed reasonable,
close to the value 0.02 found by brute force maximization. We repeat
the experiments with $\epsilon=0.02$ and obtain the same conclusions.
This shows that the problem is genuinely challenging for HMC.

\begin{figure}
\begin{centering}
\includegraphics[width=3in]{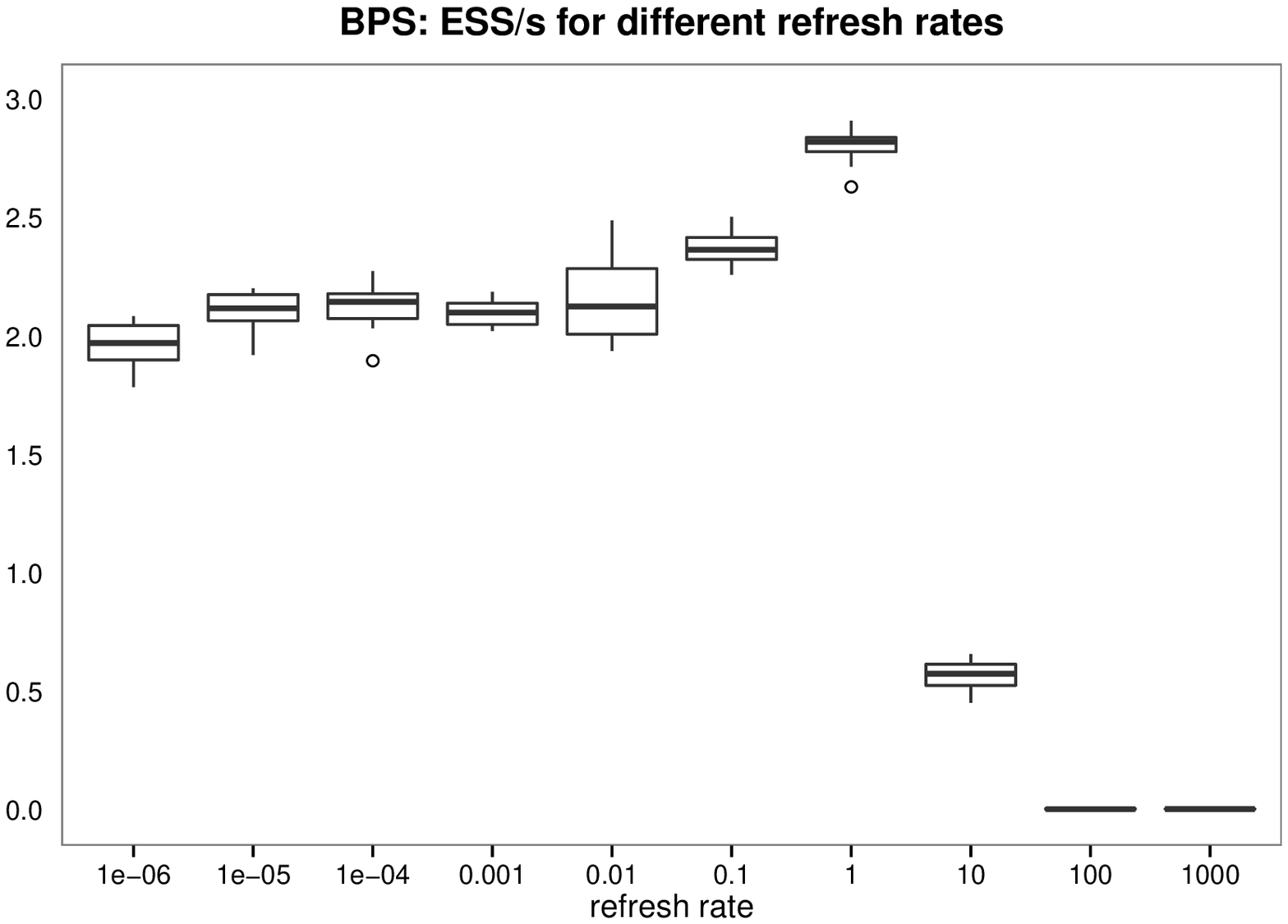}\includegraphics[width=3in]{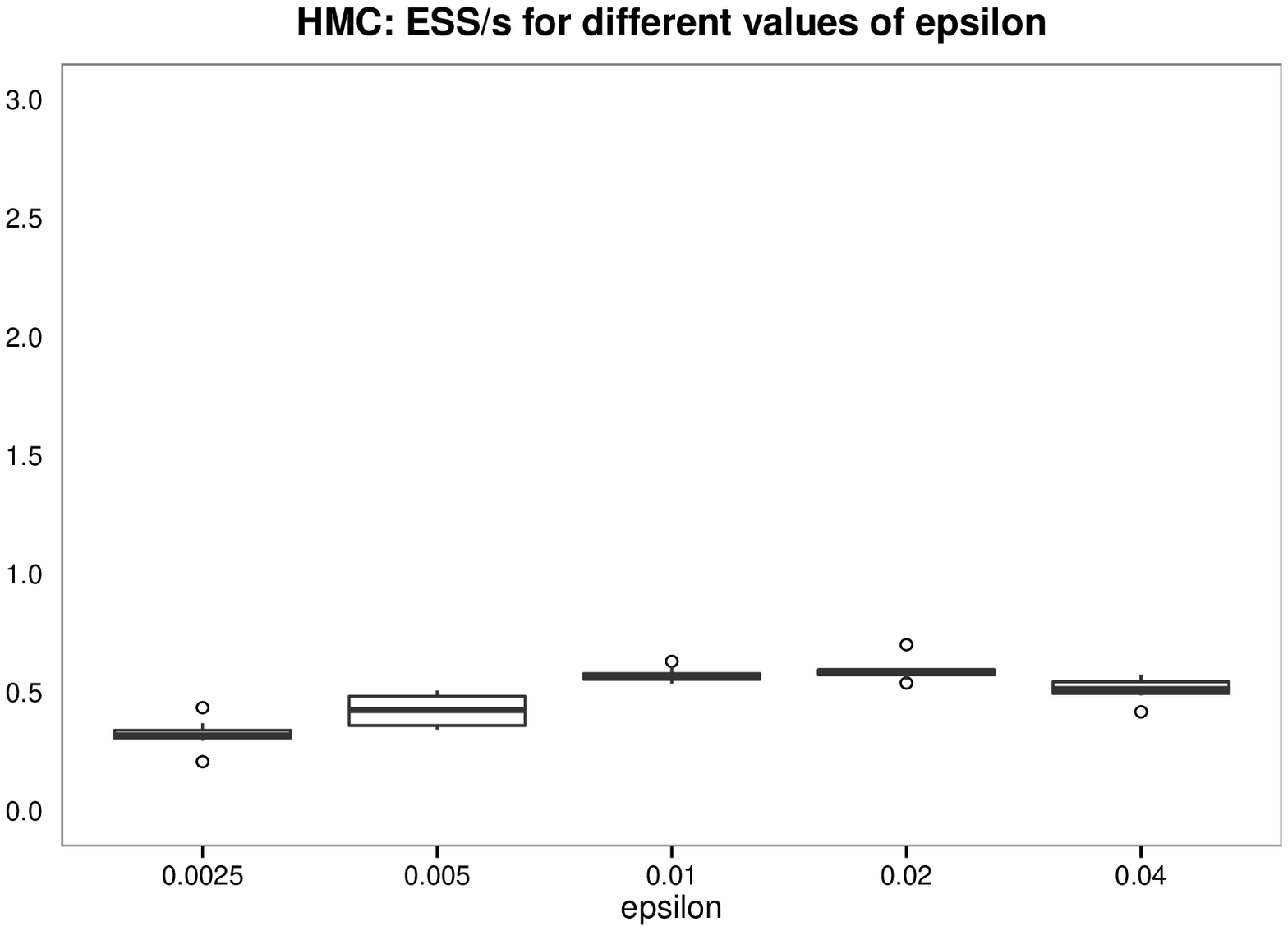} 
\par\end{centering}
\caption{\label{fig:sensitivity}Left: sensitivity of BPS's ESS/s on the log
likelihood statistic. Right: sensitivity of HMC's ESS/s on the log
likelihood statistic. Each setting is replicated 10 times with different
algorithmic random seeds.}
\end{figure}

The BPS algorithm also exhibits sensitivity to $\lambda^{\mathrm{ref}}$.
We analyze this dependency in Figure \ref{fig:sensitivity}(left).
We observe an asymmetric dependency, where values higher than 1 result
in a significant drop in performance, as they bring the sampler closer
to random walk behavior. Values one or more orders of magnitudes lower
than 1 have a lower detrimental effect. However for a range of values
of $\lambda^{\mathrm{ref}}$ covering six orders of magnitudes, BPS
outperforms HMC at its optimal parameters.
\end{document}